%% file: arxiv-main.tex
\documentclass[11pt]{article}
\usepackage{ashwin}
\usepackage{tikz}

\usetikzlibrary{arrows.meta,calc,positioning, matrix, shapes.geometric}
\pgfdeclarelayer{edgelayer}
\pgfsetlayers{background,main,edgelayer}
\tikzset{
  pathBox/.style ={fill=red!18, draw=red!70, very thick},
  stepLabel/.style={font=\scriptsize\bfseries,
                    fill=red!80, text=white,
                    circle, inner sep=1.3pt}
}

\providecommand{\email}[1]{\texttt{#1}}

\title{Learning Partition Trees for Nearest Neighbor Search} 

\author{Sanjeev Khanna\thanks{NYU. Supported in part by National Science Foundation (NSF) award CCF-2625203 and AFOSR award FA9550-25-1-0107. \{\email{sanjeev.khanna@nyu.edu}\}} \and Ashwin Padaki\thanks{University of Pennsylvania. Supported by the National Science Foundation (NSF) GRFP under Grant No. DGE-2236662, and Grant No. CCF-2337993. \{\email{apadaki@seas.upenn.edu}\}}\and Erik Waingarten\thanks{University of Pennsylvania. Supported by the National Science Foundation (NSF) under Grant No. CCF-2337993. \{\email{ewaingar@seas.upenn.edu}\}}}

\begin{document} 

\maketitle 

\input{sections/abstract}

\newpage
\tableofcontents
\newpage

\input{sections/intro}
\input{sections/tech-overview}
\input{sections/prelims}
\input{sections/cut-hard}
\input{sections/cut-algo}
\input{sections/dasgupta}
\input{sections/final-learning}

\bibliographystyle{alpha} 
\bibliography{arxiv-refs, waingarten} 

\newpage

\appendix 

\input{sections/poly-approx}
\input{sections/missing-cut-algo}

\input{sections/missing-dasgupta}

\end{document}

%% file: sections/abstract.tex

\begin{abstract}
    We study nearest neighbor search from the perspective of data-driven algorithm design: given a dataset $P \subset \R^d$ of size $n$ and sample access to a query distribution over $\R^d$, the goal is to learn a data structure optimized for queries drawn from that specific distribution. We focus on the class of balanced halfspace trees, which naturally abstracts space-partitioning frameworks like locality-sensitive hashing. Assuming Gaussian-like marginal conditions on the dataset and query distribution, we give an efficient algorithm that learns a tree achieving $o(nd)$ query time, provided that a perfect tree exists.
    
    At the core of our algorithmic approach is the balanced halfspace cut problem, where we are given a distribution over $\R^d \times \R^d$ and must find a balanced halfspace that minimizes the fraction of cut pairs. We prove that without distributional assumptions, finding the optimal balanced halfspace is \nphard. To circumvent this computational barrier, we design an efficient improper learning algorithm: if the optimal halfspace cuts an $\alpha$ fraction of pairs, our algorithm outputs a balanced polynomial threshold function of degree $\Ot(1/\eps^2)$ that cuts at most an $O(\sqrt{\alpha+\eps})$ fraction.
\end{abstract}

%% file: sections/intro.tex

\section{Introduction}\label{sec:intro}

The focus of this paper is \emph{nearest neighbor search}. For a dimensionality $d\in\N$, we receive as input a dataset of $n$ points $P = \{p_1,\ldots,p_n\}\subset\R^d$. We aim to build a data structure which can support nearest neighbor queries: given a query point $q\in\R^d$, return the point $p\in P$ minimizing $\|q-p\|_2$ over all points in the dataset.\footnote{All distances in this paper are Euclidean (i.e., $\ell_2$ distance), but the specific choice of distance plays only a minimal role.} The goal is to design data structures that answer queries in time that is \emph{sublinear} in $n$ (i.e., significantly faster than scanning the entire dataset) while keeping the space complexity manageable, ideally near-linear in $n$.

Nearest neighbor search has been studied extensively from the theoretical computer science perspective, dating back to the work of Minsky and Papert~\cite{MP69} and Knuth~\cite{knuth1973}, continuing throughout the 80s and 90s~\cite{LT80, S84, C88, M93}, and leading to relatively recent developments~\cite{IM98, AI06, AR15, ALRW17, R18}. Historically, the driving question has been to understand the conditions under which one can achieve sublinear-time queries using polynomial space for \emph{worst-case} datasets and queries. This line of work has yielded powerful algorithmic frameworks and techniques in settings of (i) low dimensionality, (ii) low intrinsic dimensionality, and (iii) approximate nearest neighbors in high dimensions (see~\cite{AIR18} for a more thorough overview).

In this work, we consider \emph{beyond worst-case} data structures for nearest neighbor search, specifically taking the perspective of \emph{data-driven algorithm design}~\cite{GR17,balcan21}. Our motivation stems from a growing body of work outside theoretical computer science on fast and accurate heuristic algorithms for nearest neighbor search~\cite{MY18,JSDS19,BIGANN21,BIGANN23}. These works do not achieve worst-case performance guarantees (see~\cite{IX23}). Instead, performance is measured on benchmark datasets and query sets via curves comparing \emph{recall} (the fraction of correctly found nearest neighbors) with \emph{qps} (queries per second). This methodology crucially relies on an implicit assumption: that future ``online'' query performance will reflect the performance measured on an ``offline'' training set of queries. From a theoretical perspective, this raises a fundamental data structure design question:
\begin{quote}
    Given a dataset and sample access to a query distribution, are there (efficient) algorithms to build (fast) nearest neighbor data structures for queries drawn from that same distribution? 
\end{quote}

Crucially, sample access to the query distribution lets one cross-validate any single data structure; the challenge, and the goal of this paper, is to compete with the best data structure in a class.

\paragraph{Learning Nearest Neighbor Data Structures.} We consider a fixed dataset $P \subset \R^d$ of size $n$, and provide the algorithm with
\begin{itemize}
\item Sample access to a distribution $\calD$ over $\R^d \times P$ such that $(\bq, \nn) \sim \calD$ corresponds to a random query and its nearest neighbor, i.e., $\nn \in P$ minimizes $\| \bq - \nn\|_2$, and 
\item A class $\calC$ of ``bounded-complexity'' nearest neighbor data structures which return the nearest neighbor $p^\star \in P$ for any query $q \in \R^d$.
\end{itemize}
Our goal is to design a polynomial-time algorithm which produces a data structure for $P$ that performs nearly as well on $\calD$ as the best data structure in the class $\calC$, whenever $\calC$ contains a good one. To make this concrete, we first need to specify the class $\calC$.

\paragraph{The Class $\calC$ of Balanced Halfspace Trees.} We study the class of balanced halfspace tree data structures, which is a natural abstraction of (idealized) locality-sensitive hashing data structures (many of which use repeated halfspace partitions~\cite{C02, DIIM04}).\footnote{Beyond halfspace partitions, capturing the class of graph-based methods is an intriguing open direction (See Subsection~\ref{sec:related})} A balanced halfspace tree $T$ for a dataset $P \subset \R^d$ is a rooted binary tree with labeled edges and $n$ leaves. Each internal node $v \in T$ contains a halfspace $f_v \colon \R^d \to \{-1,1\}$ and two outgoing edges labeled with $+1$ and $-1$. For a fixed tree $T$ and $x \in \R^d$, the point $x$ defines the natural root-to-leaf path which starts at the root node $v$ and iteratively proceeds to the child of $v$ labeled by $f_v(x)$; $T(x)$ denotes the leaf of this path. The $n$ leaves of the tree are labeled with the points of $P$ such that $T(p) = p$, and the tree $T$ is \emph{$\beta$-balanced} if each internal node has at least $\beta$-fraction of its descendant leaves on each of its children's subtrees. For a nearest neighbor pair distribution $\calD$ over $\R^d \times P$, we say that $T$ is \emph{perfect} for $\calD$ if $(\bq, \nn) \sim \calD$ always satisfies $\nn = T(\bq) \in P$.

\begin{figure}[t]
\centering
\begin{tikzpicture}[x=1cm,y=1cm,scale=0.62]
\fill[white] (0,0) rectangle (10.0,10.0);
\draw[black,line width=0.45pt,line cap=round] (3.671,1.952)--(1.907,1.602);
\draw[black,line width=0.45pt,line cap=round] (1.947,2.019)--(0.000,2.245);
\draw[black,line width=0.45pt,line cap=round] (8.917,0.000)--(7.147,1.815);
\draw[black,line width=0.45pt,line cap=round] (5.607,2.770)--(4.510,0.000);
\draw[black,line width=0.45pt,line cap=round] (7.211,3.967)--(6.649,2.490);
\draw[black,line width=0.45pt,line cap=round] (8.591,5.467)--(8.554,3.503);
\draw[black,line width=0.45pt,line cap=round] (5.944,3.862)--(5.068,5.065);
\draw[black,line width=0.45pt,line cap=round] (3.074,3.448)--(3.012,4.831);
\draw[black,line width=0.45pt,line cap=round] (7.694,6.646)--(4.781,5.678);
\draw[black,line width=0.45pt,line cap=round] (6.120,7.206)--(4.958,8.299);
\draw[black,line width=0.45pt,line cap=round] (10.000,6.813)--(8.197,7.615);
\draw[black,line width=0.45pt,line cap=round] (8.540,8.381)--(6.369,8.966);
\draw[black,line width=0.45pt,line cap=round] (3.595,4.897)--(2.590,7.540);
\draw[black,line width=0.45pt,line cap=round] (2.249,6.984)--(0.000,6.928);
\draw[black,line width=0.45pt,line cap=round] (4.879,7.958)--(3.678,8.399);
\draw[black,line width=0.45pt,line cap=round] (4.034,10.000)--(1.701,8.094);
\draw[black,line width=0.65pt,line cap=round] (2.110,3.706)--(1.752,0.000);
\draw[black,line width=0.65pt,line cap=round] (7.388,2.292)--(6.230,0.000);
\draw[black,line width=0.65pt,line cap=round] (10.000,3.005)--(6.370,4.257);
\draw[black,line width=0.65pt,line cap=round] (4.573,5.009)--(4.137,3.163);
\draw[black,line width=0.65pt,line cap=round] (6.879,8.073)--(4.203,5.018);
\draw[black,line width=0.65pt,line cap=round] (9.263,10.000)--(7.734,6.577);
\draw[black,line width=0.65pt,line cap=round] (2.464,7.619)--(1.459,4.653);
\draw[black,line width=0.65pt,line cap=round] (4.244,10.000)--(3.233,7.139);
\draw[black,line width=1.0pt,line cap=round] (4.042,0.000)--(3.404,3.360);
\draw[black,line width=1.0pt,line cap=round] (7.547,5.348)--(4.954,2.944);
\draw[black,line width=1.0pt,line cap=round] (8.381,5.443)--(5.779,10.000);
\draw[black,line width=1.0pt,line cap=round] (4.508,6.345)--(0.000,9.154);
\draw[black,line width=1.6pt,line cap=round] (10.000,1.592)--(0.000,4.272);
\draw[black,line width=1.6pt,line cap=round] (5.349,10.000)--(4.191,4.965);
\draw[black,line width=2.5pt,line cap=round] (10.000,5.628)--(0.000,4.487);
\draw[black,line width=1.1pt] (0,0) rectangle (10.0,10.0);
\draw[querblue,dotted,line width=0.8pt] (2.457,1.203)--(2.882,0.855);
\draw[querblue,dotted,line width=0.8pt] (2.627,2.067)--(2.723,2.609);
\draw[querblue,dotted,line width=0.8pt] (0.893,1.636)--(0.910,1.087);
\draw[querblue,dotted,line width=0.8pt] (1.313,3.523)--(0.987,3.079);
\draw[querblue,dotted,line width=0.8pt] (7.622,1.121)--(7.431,0.605);
\draw[querblue,dotted,line width=0.8pt] (9.252,0.821)--(8.804,1.140);
\draw[querblue,dotted,line width=0.8pt] (5.505,0.941)--(5.922,1.300);
\draw[querblue,dotted,line width=0.8pt] (4.830,1.515)--(4.408,1.869);
\draw[querblue,dotted,line width=0.8pt] (7.883,2.644)--(8.422,2.755);
\draw[querblue,dotted,line width=0.8pt] (5.690,3.381)--(6.238,3.343);
\draw[querblue,dotted,line width=0.8pt] (9.293,4.940)--(9.320,4.391);
\draw[querblue,dotted,line width=0.8pt] (8.289,4.721)--(7.775,4.525);
\draw[querblue,dotted,line width=0.8pt] (5.481,3.797)--(4.954,3.956);
\draw[querblue,dotted,line width=0.8pt] (6.666,5.027)--(6.186,4.758);
\draw[querblue,dotted,line width=0.8pt] (2.393,4.004)--(1.892,4.230);
\draw[querblue,dotted,line width=0.8pt] (3.202,4.030)--(3.737,4.157);
\draw[querblue,dotted,line width=0.8pt] (6.006,5.432)--(6.486,5.700);
\draw[querblue,dotted,line width=0.8pt] (6.146,6.341)--(6.451,6.799);
\draw[querblue,dotted,line width=0.8pt] (4.861,6.343)--(5.094,6.841);
\draw[querblue,dotted,line width=0.8pt] (5.930,9.128)--(5.844,8.585);
\draw[querblue,dotted,line width=0.8pt] (8.897,6.975)--(8.896,6.425);
\draw[querblue,dotted,line width=0.8pt] (8.792,8.361)--(9.341,8.392);
\draw[querblue,dotted,line width=0.8pt] (8.060,8.175)--(7.548,7.975);
\draw[querblue,dotted,line width=0.8pt] (7.802,8.856)--(7.617,9.374);
\draw[querblue,dotted,line width=0.8pt] (2.220,5.302)--(2.521,5.762);
\draw[querblue,dotted,line width=0.8pt] (4.123,6.405)--(3.656,6.113);
\draw[querblue,dotted,line width=0.8pt] (0.434,6.085)--(0.931,5.850);
\draw[querblue,dotted,line width=0.8pt] (1.118,7.221)--(0.975,7.752);
\draw[querblue,dotted,line width=0.8pt] (3.562,7.589)--(4.091,7.435);
\draw[querblue,dotted,line width=0.8pt] (4.472,8.513)--(4.551,9.057);
\draw[querblue,dotted,line width=0.8pt] (3.135,9.025)--(3.078,8.478);
\draw[querblue,dotted,line width=0.8pt] (2.250,9.314)--(1.700,9.320);
\fill[black] (2.882,0.855) circle (1.9pt);
\fill[querblue] (2.457,1.203) circle (1.9pt);
\fill[black] (2.723,2.609) circle (1.9pt);
\fill[querblue] (2.627,2.067) circle (1.9pt);
\fill[black] (0.910,1.087) circle (1.9pt);
\fill[querblue] (0.893,1.636) circle (1.9pt);
\fill[black] (0.987,3.079) circle (1.9pt);
\fill[querblue] (1.313,3.523) circle (1.9pt);
\fill[black] (7.431,0.605) circle (1.9pt);
\fill[querblue] (7.622,1.121) circle (1.9pt);
\fill[black] (8.804,1.140) circle (1.9pt);
\fill[querblue] (9.252,0.821) circle (1.9pt);
\fill[black] (5.922,1.300) circle (1.9pt);
\fill[querblue] (5.505,0.941) circle (1.9pt);
\fill[black] (4.408,1.869) circle (1.9pt);
\fill[querblue] (4.830,1.515) circle (1.9pt);
\fill[black] (8.422,2.755) circle (1.9pt);
\fill[querblue] (7.883,2.644) circle (1.9pt);
\fill[black] (6.238,3.343) circle (1.9pt);
\fill[querblue] (5.690,3.381) circle (1.9pt);
\fill[black] (9.320,4.391) circle (1.9pt);
\fill[querblue] (9.293,4.940) circle (1.9pt);
\fill[black] (7.775,4.525) circle (1.9pt);
\fill[querblue] (8.289,4.721) circle (1.9pt);
\fill[black] (4.954,3.956) circle (1.9pt);
\fill[querblue] (5.481,3.797) circle (1.9pt);
\fill[black] (6.186,4.758) circle (1.9pt);
\fill[querblue] (6.666,5.027) circle (1.9pt);
\fill[black] (1.892,4.230) circle (1.9pt);
\fill[querblue] (2.393,4.004) circle (1.9pt);
\fill[black] (3.737,4.157) circle (1.9pt);
\fill[querblue] (3.202,4.030) circle (1.9pt);
\fill[black] (6.486,5.700) circle (1.9pt);
\fill[querblue] (6.006,5.432) circle (1.9pt);
\fill[black] (6.451,6.799) circle (1.9pt);
\fill[querblue] (6.146,6.341) circle (1.9pt);
\fill[black] (5.094,6.841) circle (1.9pt);
\fill[querblue] (4.861,6.343) circle (1.9pt);
\fill[black] (5.844,8.585) circle (1.9pt);
\fill[querblue] (5.930,9.128) circle (1.9pt);
\fill[black] (8.896,6.425) circle (1.9pt);
\fill[querblue] (8.897,6.975) circle (1.9pt);
\fill[black] (9.341,8.392) circle (1.9pt);
\fill[querblue] (8.792,8.361) circle (1.9pt);
\fill[black] (7.548,7.975) circle (1.9pt);
\fill[querblue] (8.060,8.175) circle (1.9pt);
\fill[black] (7.617,9.374) circle (1.9pt);
\fill[querblue] (7.802,8.856) circle (1.9pt);
\fill[black] (2.521,5.762) circle (1.9pt);
\fill[querblue] (2.220,5.302) circle (1.9pt);
\fill[black] (3.656,6.113) circle (1.9pt);
\fill[querblue] (4.123,6.405) circle (1.9pt);
\fill[black] (0.931,5.850) circle (1.9pt);
\fill[querblue] (0.434,6.085) circle (1.9pt);
\fill[black] (0.975,7.752) circle (1.9pt);
\fill[querblue] (1.118,7.221) circle (1.9pt);
\fill[black] (4.091,7.435) circle (1.9pt);
\fill[querblue] (3.562,7.589) circle (1.9pt);
\fill[black] (4.551,9.057) circle (1.9pt);
\fill[querblue] (4.472,8.513) circle (1.9pt);
\fill[black] (3.078,8.478) circle (1.9pt);
\fill[querblue] (3.135,9.025) circle (1.9pt);
\fill[black] (1.700,9.320) circle (1.9pt);
\fill[querblue] (2.250,9.314) circle (1.9pt);
\node[querblue,font=\footnotesize] at (0.593,1.686) {$\bq$};
\node[font=\footnotesize] at (1.310,1.037) {$\nn$};
\end{tikzpicture}

\caption{A perfect balanced halfspace tree in $\R^2$ for a nearest neighbor pair distribution $\Dpair$ over $32$ pairs. The black points are the dataset, and each query $\bq$ is connected to its nearest neighbor $\nn$ with a dotted line. The thick diagonal is the hyperplane at the root, and lines get thinner with depth.}
\label{fig:balanced-halfspace-tree}
\end{figure}
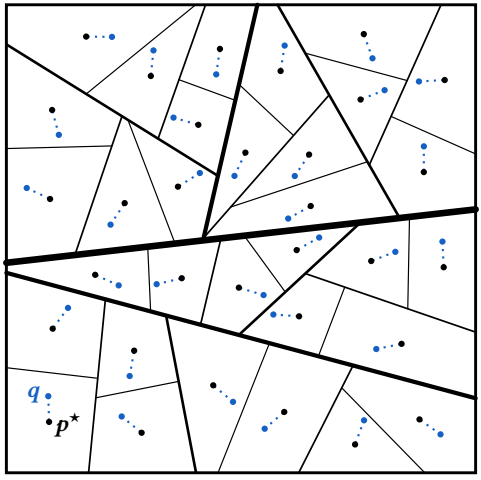

\begin{observation}
A dataset $P \subset \R^d$ and nearest neighbor distribution $\calD$ which admit a perfect $\beta$-balanced halfspace tree yield an $O(nd)$-space nearest neighbor data structure with query time $O(d \log n / \beta)$.
\end{observation}

Our concrete problem is the following: given a dataset $P$ and sample access to a nearest-neighbor distribution $\calD$ for which $(P,\calD)$ admits a perfect $\beta$-balanced halfspace tree (the realizable regime), design a data structure (in polynomial time) that uses polynomial space and achieves sublinear query time.

\subsection{Related Work}\label{sec:related}

Our work adopts the framework of Gupta and Roughgarden~\cite{GR17}, who cast algorithm selection from a parameterized family as a statistical learning problem over an instance distribution (here, the parameterized family is $\calC$ and the instance distribution is $\calD$).\footnote{A related but distinct line-of-work studies learning-augmented algorithms, or algorithms with predictions~\cite{MV22} in which a (possibly unreliable) predictor is supplied as a side input to improve performance while retaining worst-case guarantees. In contrast, we do not assume access to a predictor; our goal is to learn the data structure itself from samples of the query distribution.} Within theoretical computer science, the idea of tailoring a nearest neighbor data structure to the dataset appears in various contexts: spectral algorithms for semi-random settings~\cite{AAKK14}, data-dependent locality-sensitive hashing~\cite{AINR14, AR15, ALRW17, ANNRW18}, data-dependent decompositions and embeddings~\cite{ANRW21, JWZ24, AJW25, AN25, BBMWWZ25}. However, while these works adapt the data structure to the dataset $P$, they remain worst-case with respect to queries; they make no assumptions on the query distribution and derive no benefit from having access to it. A complementary, more empirical thread learns partitions directly (e.g., NeuralLSH~\cite{DIRW20}) and more broadly, there is an extensive body of work on \emph{learning to hash} (see~\cite{WZSSS17} for a survey). 

To our knowledge, the only prior data-driven work in our setting (i.e., that uses the query distribution to inform the data structure) is that of Cayton and Dasgupta~\cite{CD07}. They use the query distribution to select hash bin widths after applying a random projection. Their method is tied to the specific partitioning scheme (i.e., random projection and bin) and comes without guarantees relative to a general class of data structures. The central question we take up is how to compete with the \emph{best} data structure in a class---in this work, the class of balanced halfspace trees---which requires both formulating the intermediate learning objectives and designing efficient algorithms for them. An interesting direction for future work is whether efficient learning algorithms exist for the class given by using $k$ balanced halfspace trees, or an appropriate class to capture \emph{graph-based methods}~\cite{IX23,DGMMS24, KPW26, CDFJLMSSW26}.

\subsection{Our Results} 

Our main result is a polynomial-time learning algorithm which can output a sublinear query time (i.e., $o(nd)$ time) nearest neighbor data structure whenever there exists a perfect balanced halfspace tree, and when the query distribution and dataset have ``Gaussian-like'' marginals (which we elaborate on shortly).

\begin{theorem}[Main Theorem---informal version of Theorem~\ref{thm:main-formal}]\label{thm:main-intro}
    Fix a dataset $P\subset\R^d$ of size $n$ where $d = \polylog(n)$. Let $\Dpair$ be a nearest neighbor distribution for $P$, and suppose: \begin{itemize}
        \item there exists a perfect $1/3$-balanced halfspace tree for $\Dpair$, and
        \item the nearest neighbor pair distribution $\calD$ has ``Gaussian-like'' marginals.
    \end{itemize}
    Then, there is a polynomial-time algorithm that outputs a nearest neighbor data structure using $O(nd)$ space and $o(nd)$ query time, which returns the true nearest neighbor $\nn$ with  probability $1-o(1)$ over $(\bq, \nn) \sim \Dpair$.
\end{theorem}

To elaborate on the second condition, the underlying algorithm relies on low-degree polynomial approximations of (functions of) halfspaces in order to learn efficiently. Hence, we will require both marginals of $\Dpair$ to exhibit certain concentration and anticoncentration properties (generalizing properties of high-dimensional Gaussian distributions, see Appendix~\ref{sec:poly-approx}). This is a standard assumption made in learning theory (e.g., for agnostic learning of halfspaces and functions of halfspaces~\cite{klivans2002,KKMS08, DKZ20}). More formally, for small $\eps > 0$, if $f \colon \R^d \to \{-1,1\}$ is any function where the set $f^{-1}(1)$ is the union of $\poly(1/\eps)$ disjoint geometric regions, each described by an intersection of $O(\log(1/\eps))$ halfspaces, there should exist a $\poly(1/\eps)$-degree polynomial $A \colon \R^d \to \R$ which approximates $f$ on a random dataset point and random query (we require approximability in terms of first, second and fourth moments). 

\begin{remark}[Why Barely Sublinear]\label{rm:barely-sublinear}
Although a perfect tree would support $O(d \log n)$ time per query, Theorem~\ref{thm:main-intro} does not recover the perfect halfspace tree. Rather, each cut of the tree is learned only improperly and approximately (via optimizing over low-degree polynomials). Roughly speaking, the final query algorithm traverses an $O(\log(1/\eps))$-depth tree by evaluating $\poly(1/\eps)$-degree polynomials (each can take $d^{\poly(1/\eps)}$-time), and this narrows down the search to an $\eps$-fraction of the dataset. The query time is of the form $d^{\poly(1/\eps)} + \eps n d$ and we can set $\eps = (\log n)^{-\Omega(1)}$ when $d = \polylog(n)$ so that $d^{\poly(1/\eps)} = n^{o(1)}$. While this degrades significantly from the ideal $O(d \log n)$ time to $o(nd)$, it is still a substantial improvement over linear scan for exact nearest neighbor search. (See Footnote~\ref{foot-open} for a concrete challenge.)
\end{remark}

\paragraph{Balanced Halfspace Cut.} As we discuss in the technical overview, Theorem~\ref{thm:main-intro} will build a hierarchical clustering of the dataset $P$ which minimizes an asymmetric version of Dasgupta's objective for hierarchical clustering~\cite{dasgupta2016}. Roughly speaking, given a query and nearest neighbor distribution $\Dpair$, Dasgupta's objective will seek to minimize the expected number of dataset points within the smallest hierarchical partition containing both $(\bq, \nn) \sim \Dpair$. This objective function is well-studied, and we adapt an argument from~\cite{charikar2017}, which shows that recursive sparsest cuts (of the underlying pair distribution $\Dpair$) approximately minimize Dasgupta's objective. While we defer the details of the reduction, this motivates the key algorithmic problem, \emph{balanced halfspace cut}, which forms the technical bulk of this paper.

\begin{quote}
    Given sample access to a distribution $\Dpair$ over pairs $(\bq,\nn)\in \R^d \times P$, find a \emph{balanced} halfspace $f \colon \R^d\to\{-1,1\}$ with respect to the second marginal of $\Dpair$ that minimizes the \emph{cut error}, defined as the probability that $f(\bq)\ne f(\nn)$.
\end{quote}

Specifically, let $\opt \geq 0$ denote the minimum cut error $\Pr_{(\bq,\nn)\sim\Dpair}[f(\bq) \neq f(\nn)]$ over halfspaces $f$ which are $1/3$-balanced with respect to $P$, i.e., $\Pr_{(\bq, \nn)\sim \Dpair}[f(\nn) = b] \geq 1/3$ for $b \in \{-1,1\}$. The goal of the balanced halfspace cut problem is to produce a halfspace $f$ of constant balance and error $\approx \opt$. We present two main results for this problem.

 \begin{itemize}
     \item \textbf{Hardness of Proper Learning (Theorem~\ref{thm:strong-hardness}).} We establish a strong computational barrier: given a collection of pairs $(x_1, y_1), \dots, (x_m, y_m) \in \R^{d} \times \R^{d}$,\footnote{Here, both $m$ and $d$ are considered asymptotic parameters.} it is \nphard to find an $\Omega(1)$-balanced halfspace of small cut probability (with respect to the uniform distribution over pairs), even when $\opt = 0$.
     \item \textbf{Improper Learning via PTFs (Theorem~\ref{thm:improper-alg}).} We circumvent the above hardness with an improper learner under the distributional assumption. Provided that $\Dpair$ admits $\poly(1/\eps)$-degree polynomial approximations of halfspaces, our algorithm outputs an $\Omega(1)$-balanced $\poly(1/\eps)$-degree PTF of error $O(\sqrt{\opt + \eps})$ and runs in $d^{\poly(1/\eps)}$ time. 
 \end{itemize}

 We remark that the obtained $\sqrt{\opt}$ guarantee of Theorem~\ref{thm:improper-alg} is tight assuming \sseh. In particular, Corollary~\ref{cor:square-root-tight} (which directly follows from~\cite{raghavendra2012}) implies that there is some constant $c > 0$ such that for any $\alpha > 0$, it is \nphard (assuming \sseh) to find any $\Omega(1)$-balanced cut with error less than $c \sqrt{\alpha}$ when a $1/2$-balanced cut of error $\alpha$ exists---this hardness holds even when $\Dpair$ admits degree-1 approximations of all cuts. Despite the tightness of the approximation guarantee of Theorem~\ref{thm:improper-alg}, there remains a natural open problem: whether the running time of can be improved from $d^{\poly(1/\eps)}$ to $d^{\polylog(1/\eps)}$, or even $\poly(d,1/\eps)$.\footnote{\label{foot-open} We note that a $\poly(d,1/\eps)$-time algorithm which can (improperly) learn low cut error, $\Omega(1)$-balanced unions-of-intersections of halfspaces would directly improve on the final running time of Theorem~\ref{thm:improper-alg}  (see Remark~\ref{rm:barely-sublinear}). At the moment, we do not even know whether the following admits a $\poly(d,1/\eps)$-time algorithm: given $\Dpair$ over $\R^d \times \R^d$ whose marginals are standard Gaussian and admitting a perfect $1/2$-balanced halfspace, find a $\Omega(1)$-balanced cut of error at most $\eps$. Furthermore, we note that the \sseh-hardness of improving the $\sqrt{\opt}$ dependence does not preclude a polynomial-time algorithm whose approximation guarantee is $r(d) \cdot \opt$ for some growing function $r(\cdot)$.}

%% file: sections/tech-overview.tex

\section{Technical Overview}
\label{sec:technical-overview}

At a high level, the proof of Theorem~\ref{thm:main-intro} has three conceptual steps. First, we reduce the data structure objective to a hierarchical cut objective: a tree is useful exactly when query-neighbor pairs remain together deep into the tree. Second, we show how to learn a \emph{single} balanced geometric cut. Properly learning the best balanced halfspace is \nphard, so we instead learn a low-degree polynomial
whose threshold gives a balanced low-error cut. Third, we recursively apply this cut learner and charge the resulting tree to the perfect halfspace tree using a Dasgupta-style sparsest cut argument. The main technical point is that the recursive certificates are no longer single halfspaces, but unions-of-intersections of halfspaces, which explains the stronger polynomial approximation assumption in the final theorem.

Recall that a balanced halfspace tree routes each point along a root-to-leaf path, and that searching for the nearest neighbor of a query $\bq$ amounts to routing $\bq$ to a leaf and scanning the dataset point stored there. A \emph{perfect} balanced tree never cuts a pair $(\bq, \nn) \sim \Dpair$, so $\nn$ lies alone in the leaf $T(\bq)$ and the search is immediate; an imperfect tree separates some pairs earlier, and the search must scan every dataset point that still shares a node with $\bq$ when its path diverges from $\nn$. This motivates minimizing the following asymmetric Dasgupta objective function~\cite{dasgupta2016},\footnote{In hierarchical clustering, one is given a weighted graph $G = (V, E)$ with $w \colon E \to \R_{\geq 0}$ and Dasgupta's objective asks to find a hierarchical (tree) partition of $V$ which minimizes $\sum_{e=(u,v) \in E} w(e) |V^{\lca_T(u,v)}|$. We say the objective in our setting is asymmetric because the size of the subtree is only measured with respect to dataset points $P$.} which measures the expected weight of dataset points sharing the deepest node (i.e. least common ancestor or LCA) common to $(\bq, \nn)$,  
\[
    \cost_{\Dpair}(T) := \Ex_{(\bq,\nn)\sim\Dpair}\left[\Prx_{(\bq',\bp') \sim \calD}\left[ \bp' \in P^{\lca_T(\bq, \nn)}\right]\right] \in [0,1].
\] 

In the case that the $\nn$-marginal of $\calD$ is uniform over $P$, this cost is exactly proportional to the number of dataset points in the LCA node of $\bq$ and $\nn$. Indeed, up to the cost of routing through the tree, $\cost_{\Dpair}(T)$ dictates the cost of nearest neighbor search (see Claim~\ref{claim:cost-to-search}) and reduces the data structure design question to finding a low-cost balanced cut tree. 

\subsection{Step 1: Finding a Single Cut} 

We use a greedy, top-down approach to building the tree and hence focus on the single-node, balanced halfspace cut problem described in Section~\ref{sec:intro}. The objective is to compute a geometric cut $f \colon \R^d \to \{-1,1\}$ which competes with halfspaces (later generalized to unions-of-intersections of halfspaces) and
\begin{itemize}
\item Minimizes the cut error $\err_{\Dpair}(f) := \Prx_{(\bq, \nn)\sim\Dpair}[f(\bq)\neq f(\nn)]$, so that a query and its nearest neighbor are often deep in the tree together;
\item Satisfies a balance constraint $\bal_{\Dpair}(f) := \min_{b\in\{-1,1\}} \{ \Prx_{(\bq,\nn)\sim\Dpair}[f(\nn)=b]\} \geq \Omega(1)$, so that the dataset decreases geometrically down the tree; and
\item Maintains a low computational complexity, so the time to evaluate $f(\bq)$ (which factors into the query time) is faster than a linear scan over the dataset.
\end{itemize}
The first two requirements (low error and balance) have direct analogues in graph algorithms (i.e., the balanced separator problem~\cite{LR99, ARV09}), but there is no direct mechanism for imposing the geometric (low-complexity) constraint. We will instead use the rounding idea of Cheeger's inequality~\cite{AM85}: first, find a real-valued function $F \colon \R^d \to \R$ of low computational complexity which minimizes
\[ \Rpair(F) := \dfrac{\Ex_{(\bq, \nn)\sim\Dpair}\left[ (F(\bq) - F(\nn))^2 \right]}{\Varx_{(\bq,\nn)\sim\Dpair}[F(\nn)]}. \] 
Then, round $F$ via the transformation $F(x) \mapsto \sign(F(x) - v) \cdot (F(x) - v)^2$ and a random threshold $\btau \sim [-R, R]$ for a sufficiently large $R > 0$. The final cut becomes $f(x) = -1$ if and only if $\sign(F(x) - v) \cdot (F(x) - v)^2 \leq \tau$ for fixed $v, \tau \in \R$, which inherits the low computational complexity of $F$.\footnote{The function $f$ given by $\sign( \sign(F-v) \cdot (F-v)^2 - \tau)$ is not immediately a PTF (as claimed in Theorem~\ref{thm:improper-alg}), but since $\sign(F-v) \cdot (F-v)^2$ is monotone in $F$, there is a threshold $\theta$ where $f(x) = \sign(F(x) - \theta)$, which \emph{is} a PTF when $F$ is a polynomial.} 

The plan becomes: minimize $\Rpair(F)$ over degree-$k$ polynomials $F$. This is both efficient, as it is a clean eigenvalue computation on a $\binom{d+k}{k} \times \binom{d+k}{k}$ psd matrix, and it automatically satisfies the low-complexity constraint, since $f$ can be evaluated in $O((d+k)^k)$ time. We use the distributional assumption on $\Dpair$ to ensure that the optimal halfspace cut $f^{\star}$ has a degree-$k(\eps)$ polynomial approximation $A^{\star}$, where $k(\eps)=\poly(1/\eps)$, whose squared error is at most $\eps$ under both marginals of $\Dpair$. This implies that the minimum value of $\Rpair(F)$ over degree-$k(\eps)$ polynomials is at most $O(\opt+\eps)$; see Appendix~\ref{subsec:polynomial-approx}. Finally, the analysis of Cheeger's rounding guarantees that the \emph{sparsity} of the output cut, measuring the ratio $\err_{\Dpair}(f) / \bal_{\Dpair}(f)$, is at most $O(\sqrt{\Rpair(F)}) = O(\sqrt{\opt + \eps})$. 

\textbf{Minimizing $\err_{\Dpair}(f) / \bal_{\Dpair}(f)$ may give unbalanced cuts.} The immediate issue with the above approach is that finding a degree-$\poly(1/\eps)$ polynomial $F$ which minimizes $\Rpair(F)$ may result in a sparse cut which approximately minimizes the ratio $\err_{\Dpair}(f) / \bal_{\Dpair}(f)$ but is very far from balanced. This is often not an issue in approximation algorithms for graphs, since an iterative peeling argument~\cite{LR99} can accumulate a $\Omega(1)$-balanced separator by a union of at most $n$ sparse cuts; however, this is precisely an argument we must avoid (since it would correspond to a slow algorithm for evaluating the cut function $f$). Instead, we add one more constraint to the search for a degree-$\poly(1/\eps)$ polynomial $F$ which ensures the resulting cut is $\Omega(1)$-balanced. Instead of minimizing $\Rpair(F)$, we seek a degree-$\poly(1/\eps)$ polynomial which minimizes
\begin{align} \Ex_{(\bq,\nn)\sim\Dpair}\left[ (F(\bq) - F(\nn))^2 \right] \qquad \text{subject to} \qquad \begin{array}{ll} \text{(i)} & \displaystyle \Ex_{(\bq,\nn)\sim\Dpair}[F(\nn)] = 0 \\
       \text{(ii)} & \displaystyle \Ex_{(\bq,\nn)\sim\Dpair}[F(\nn)^2] = 1 \\
       \text{(iii)} & \displaystyle \Ex_{(\bq,\nn) \sim\Dpair}[F(\nn)^4] \leq C \end{array}, \label{eq:ideal-poly-program}
\end{align}
for a large enough (constant) $C > 1$ (concretely, we set $C$ to $\Theta(1/\beta^2)$ whenever we want to compete with a $\beta$-balanced halfspace, and we will set $\beta=1/3$ by convention). 

Note, constraints (i) and (ii) enforce the variance of $F(\nn)$ is $1$. Constraint (iii), which upper bounds the fourth moment, guarantees that the unit variance does \emph{not} simply come from a tiny fraction of $\nn$ which have extremely high $F(\nn)^2$ while the rest of the dataset points have $F$-values tightly concentrated near 0. Specifically, an upper bound on the fourth moment lets us draw the random threshold $\btau$ from a restricted interval which guarantees $\bal_{\Dpair}(f) = \Omega(1/C)$ (Lemma~\ref{lem:shift-existence}). Fortunately, the same polynomial constructions for approximating halfspaces (using the anticoncentration and concentration of marginals of $\Dpair$) guarantee that for the \emph{optimal} balanced halfspace $f^{\star}$, there is a $\poly(1/\eps)$-degree polynomial $A^{\star}$ satisfying $\Ex_{\bx}\left[ (f^{\star}(\bx) - A^{\star}(\bx))^4\right] \leq \eps$ for $\bx$ drawn from both marginals of $\Dpair$ (see Appendix~\ref{subsec:poly-approx}). This is enough to show that the minimizing objective value is $O(\opt + \eps)$, so rounding guarantees $\err_{\Dpair}(f) = O(\sqrt{\opt + \eps})$ and $\bal_{\Dpair}(f) = \Omega(1)$.

\textbf{Final Challenge: (\ref{eq:ideal-poly-program}) is not a convex program.} The goal of (\ref{eq:ideal-poly-program}) is to optimize over the $\binom{d+k}{k}$ coefficients (corresponding to all monomials) in a degree-$k$ polynomial $F \colon \R^d \to \R$. Given a vector of coefficients $\phi$ and a point $x \in \R^d$, one evaluates $F(x)$ by $\langle \phi, \psi(x)\rangle$ where \smash{$\psi \colon \R^d \to \R^{\binom{d+k}{k}}$} is the appropriate mapping which prepares the monomials with input $x$. However, as described, (\ref{eq:ideal-poly-program}) is non-convex (e.g., if $F$ satisfies all constraints, so does $-F$, but not $(F-F)/2 = 0$).\footnote{Note that it is constraint (ii) which makes the above non-convex. On the other hand, without the fourth-moment constraint (iii), (\ref{eq:ideal-poly-program}) is solved via minimizing $\Rpair(F)$ using an eigenvalue computation, then translating and rescaling so (i) and (ii) are satisfied.} To handle all constraints simultaneously, we will solve a semidefinite programming relaxation and then round the output to a polynomial.

Instead of assigning each of the $\binom{d+k}{k}$ coefficients of $F$ a real number, we assign each coefficient a vector. The prior vector of coefficients $\phi$ now becomes a collection of vectors; \smash{$\psi \colon \R^d \to \R^{\binom{d+k}{k}}$} stays the same, so the prior inner products $F(x) = \langle \phi, \psi(x)\rangle$ becomes a linear combination of vectors in $\phi$. The objective and first two constraints now become
\[  \Ex_{(\bq,\nn)\sim\Dpair}[\|F(\bq) - F(\nn)\|_2^2]; \qquad \left\| \Ex_{(\bq,\nn)\sim\Dpair}[F(\nn)]\right\|_2^2 = 0 \qquad \text{and}\qquad \Ex_{(\bq,\nn) \sim \Dpair} [ \| F(\nn) \|_2^2 ] = 1, \]
and the final constraint can be expressed with
\[ \sup_{v \in \R^{P},\, \|v\|_{2} = 1}\left\{ \sum_{p\in P} v_p \cdot \Prx_{(\bq,\nn)\sim\Dpair}[\nn=p]^{1/2} \cdot \|F(p)\|_2^2\right\} = \Ex_{(\bq,\nn)\sim\Dpair}\left[ \|F(\nn)\|_2^4 \right]^{1/2} \leq C^{1/2}. \]
By expanding squared norms $\|\cdot\|_2^2 = \langle \cdot, \cdot\rangle$, all objectives and constraints are linear (in)equalities of inner products among vectors of $\phi$, and we can find a positive semidefinite matrix $M$ encoding all inner products. The final step, to obtain the polynomial $F \colon \R^d \to \R$ (with real-valued coefficients), is to draw a Gaussian vector with covariance $M$ for the coefficients. The semidefinite program is precisely set up to satisfy (\ref{eq:ideal-poly-program}) in expectation (over the draw of the Gaussian for the coefficients), and a simple argument shows that the rounding scheme affects the objective and constraints by only a constant factor (Claim~\ref{claim:sdp-rounding}).

\subsection{Recursive Sparsest Cut}

Given the balanced cut algorithm, our high-level strategy is simple: recursively apply the algorithm to obtain a top-down ``greedy'' tree. The key challenge is to \emph{charge} the cost of this tree (specifically, the contribution of each cut) to the optimal halfspace tree $T^{\star}$. We argue this using an adaptation of an argument by~\cite{charikar2017} to which we refer the reader. Namely, we show that $T^{\star}$ induces, at each of its levels, a clustering of the dataset, and that these clusters yield sparse ``certificate'' cuts of the clusters appearing in our recursive tree.

An important subtlety is that the clusters induced by $T^{\star}$ are \emph{not} halfspaces: each is carved out by the halfspace cuts along several root-to-node paths of $T^{\star}$, so the corresponding certificate cuts are given by a \emph{union-of-intersections of halfspaces}. This is precisely why in our final result (Theorem~\ref{thm:main-formal}), we require that polynomials approximate not just halfspaces (which is sufficient for a single cut) but unions-of-intersections of halfspaces. We defer additional details to Section~\ref{sec:dasgupta}.

%% file: sections/prelims.tex

\section{Preliminaries}

\subsection{The Balanced Halfspace Cut Problem}
\label{subsec:balanced-cut-def}

For a dimension $d \in \N$, a \emph{halfspace} $f \colon \R^d \to \{-1,1\}$ is specified by a function $f(x) = \sign(\langle h, x \rangle - \theta)$, where $h \in \R^d$ and $\theta \in \R$. More generally, a degree-$k$ \emph{polynomial threshold function} $f \colon \R^d \to \{-1,1\}$ is specified by $f(x) = \sign(A(x))$, where $A \colon \R^d \to \R$ is a degree-$k$ polynomial (and note, halfspaces are degree-1 polynomial threshold functions). We let $\calP_1$ be the set of halfspaces and $\calP_k$ the set of degree-$k$ polynomial threshold functions.

Let $\Dpair$ denote a distribution over pairs $(\bx,\by)\in\R^d\times\R^d$. In the context of nearest neighbor search, we will think of $\bx$ as a query and $\by$ as its nearest neighbor in a finite dataset, but our results on the balanced halfspace problem will be for arbitrary distributions over pairs. We use the notation $\Dx$ and $\Dy$ for the $\bx$- and $\by$-marginals of $\Dpair$, respectively.

\begin{definition}[Cut Error and Balance]
    Let $f : \R^d \to \{-1,1\}$. The \emph{error} of $f$ on $\Dpair$ is: \[\err_{\Dpair}(f) := \Prx_{(\bx,\by)\sim\Dpair}[f(\bx) \ne f(\by)].\] The \emph{balance} of $f$ is defined with respect to $\by$-marginal of $\Dpair$ as: \[\bal_{\Dpair}(f) := \min \ \cbr{\Prx_{\by\sim \Dy}[f(\by)=1],\Prx_{\by\sim \Dy}[f(\by)=-1]}.\]    
\end{definition}

\begin{definition}[Sparsest Balanced Halfspace Cut]
\label{def:opt-halfspace}
    Given a distribution $\Dpair$, we use $\opt_\beta(\calP_1)$ to denote the minimum error of any halfspace of balance $\ge \beta$. Formally: \[\opt_\beta(\calP_1) := \min_{f\in\calP_1}\cbr{\err_{\Dpair}(f) \mid \bal_{\Dpair}(f) \ge \beta}.\] When clear from context, we simply write $\opt_\beta$.
\end{definition}

\begin{tcolorbox}[colback=cyan!05]
    \begin{problem}[Balanced Halfspace Cut]
    \label{problem:balanced-halfspace-cut}
        Let $\Dpair$ be a distribution over pairs $(\bx,\by)\in\R^d\times\R^d$, let $\beta\in(0,1/2]$ be a balance target, and let $\eps > 0$ be an error parameter. Given $\poly(d)$ i.i.d.\ samples from $\Dpair$, the task is to design a $\poly(d)$-time algorithm that outputs a cut $f \colon \R^d \to \{-1,1\}$ satisfying: 
        \begin{itemize}
            \item \textbf{Approximate Optimality:} $\err_{\Dpair}(f) \le \Psi(\opt_\beta(\calP_1)) + \eps$, for a non-decreasing function $\Psi$.
            \item \textbf{Balance:} $\bal_{\Dpair}(f) \ge \psi(\beta)$, for a non-decreasing function $\psi$.
            \item \textbf{Cut Complexity:} In the \emph{proper} setting, we require $f\in \calP_1$. In the \emph{improper} setting, $f$ may be a cut of higher complexity, such as a low-degree PTF.
        \end{itemize}
    \end{problem}
\end{tcolorbox}

%% file: sections/cut-hard.tex

\section{Hardness of Proper Learning: Theorem~\ref{thm:strong-hardness}}
\label{sec:balanced-cut-hard}

While one would hope for an algorithm which outputs a low-error halfspace when one exists, i.e., $\err_{\Dpair}(f) \le O(\opt_\beta) + \eps$ for an error parameter $\eps$, our first result (Theorem~\ref{thm:strong-hardness}) gives a strong computational barrier: it is \nphard to obtain a small constant additive approximation or any multiplicative approximation to the optimal halfspace cut.

\begin{theorem}
    \label{thm:strong-hardness}
    There exists a universal constant $c>0$ such that, for any $\beta\in (0,1/2]$, the following problem is \nphard. Given pairs $(x_1,y_1),\ldots,(x_m,y_m)$ in $\R^d\times\R^d$, let $\Dpair$ denote the uniform distribution over $\{(x_i,y_i)\}_{i=1}^m$. Distinguish between the following cases:
    \begin{itemize}
        \item There exists $f\in\calP_1$ such that $\bal_{\Dpair}(f)=1/2$ and $\err_{\Dpair}(f) = 0$.
        \item Every $f\in\calP_1$ with $\bal_{\Dpair}(f) \ge \beta$ satisfies $\err_{\Dpair}(f)\ge c\beta$.
    \end{itemize}
\end{theorem}

    \subsection{Description of the Reduction} 

We prove Theorem~\ref{thm:strong-hardness} by reducing from the \setsplit problem (equivalent to NAE-SAT without negated variables).

\begin{definition} An instance of $3$-\setsplit is given by a family $\calF$ of size-$3$ subsets of $[n]$. The task is to find a partition $[n] = U_1\sqcup U_2$ maximizing the number of sets $S\in\calF$ that are \emph{split} (i.e., $S\cap U_1\neq \emptyset$ and $S\cap U_2\neq \emptyset$).
\end{definition}

\begin{lemma}[\cite{guruswami03, charikar2005}]\label{lem:set-split-hard}
    There is a constant $B\in\N$ such that given a $3$-\setsplit instance where each element $i\in[n]$ occurs in at most $B$ sets $S\in\calF$, it is \nphard to distinguish between the following cases: \begin{itemize}
        \item There exists a partition which splits all sets.
        \item Any partition splits at most $\frac{24}{25}$-fraction of the sets.
    \end{itemize}
\end{lemma}

We may assume without loss of generality that every element of $[n]$ belongs to at least one set of $\calF$. In particular, $n\le 3|\calF|$. We now describe a polynomial-time reduction from $3$-\setsplit to the balanced halfspace cut problem. Given a family $\calF$ of size-$3$ subsets of $[n]$, we construct a multiset of pairs in $\R^{n+1}\times\R^{n+1}$ where $\{e_0,\ldots,e_n\}$ denotes the standard basis of $\R^{n+1}$. As we describe, the $0$-th coordinate will be a normalization and coordinates $1, \dots, n$ will be associated with the elements in $[n]$. At a high level, the reduction will relate a halfspace $f(x) = \sign(\langle h, x\rangle - \theta)$ to a partition of $[n]$ via the signs of coefficients $h_1,\dots, h_n$ which specify $h$. The construction is as follows:
\begin{itemize}
    \item \textbf{Centering Pairs.} For $T := (1-\beta)/\beta$ and $M:=\lfloor T\cdot(n+|\calF|) \rfloor + 1$, we create $M$ copies of each of the following self-pairs: $L_0 := (-e_0/10,\, -e_0/10)$ and $R_0 := (e_0/10,\, e_0/10).$ $L_0$ and $R_0$ are never cut, but they force any sufficiently balanced halfspace $f(x) = \sign(\langle h, x \rangle - \theta)$, after normalization, to have $h_0 = 1$ and $|\theta| \leq 1/10$ (see the first paragraph in the proof of Lemma~\ref{lem:split-soundness}).
    \item \textbf{Element Pairs.} For every element $i\in[n]$ we create two pairs: $L_i := (e_0+e_i/2,\, -e_0+e_i/2)$ and $R_i := (e_0+e_i/3,\,e_0-e_i/3).$ These pairs help enforce that any halfspace $f(x) = \sign(\langle h, x \rangle - \theta)$ with $h_0 = 1$ and $|\theta| \leq 1/10$ will have $|h_i|$ relatively close to $2.5$ (see Claim~\ref{claim:element-pair-cut}). 
    \item \textbf{Set Pairs.} For every set $\{i,j,k\}\in\calF$, let $e_{i+j+k} := e_i+e_j+e_k$. We create two pairs: $L_{i,j,k} := (-5e_0+e_{i+j+k},\, -11e_0+e_{i+j+k})$ and $R_{i,j,k} := (5e_0+e_{i+j+k},\,11e_0+e_{i+j+k}).$ A halfspace $f(x) = \sign(\langle h, x \rangle - \theta)$ with $|\theta| \leq 1/10$ and $|h_i|, |h_j|, |h_k|$ close to $2.5$ will induce a partition (via the sign of $h$) which splits $\{ i, j, k\}$ iff both $L_{i,j,k}$ and $R_{i,j,k}$ are uncut (see Claim~\ref{claim:set-pair-cut}).
    \item \textbf{Symmetrization Pairs.} Finally, let $\calS_0$ denote the multiset comprising all element, set, and centering pairs constructed above. We define the final multiset of pairs as $\calS = \calS_0 \cup \{(-x, -y) \mid (x,y) \in \calS_0\}$. This symmetrizes the construction to ensure $\bal_{\Dpair}(f^*) = 1/2$ in the completeness case (see Lemma~\ref{lem:split-completeness}). Now, we define $\Dpair=\Dpair(\calF)$ to be the uniform distribution over the final multiset $\calS$.
\end{itemize}

\subsubsection{Analysis of the Reduction}

We begin with the two helper claims discussed above, which guarantee that a halfspace $f(x) = \sign(\langle h, x\rangle - \theta)$, normalized so $h_0 = 1$ and $|\theta| \leq 1/10$, will have $|h_i|$ close to $2.5$ (unless pairs were cut by $f$). Then, we prove the completeness of the reduction in Lemma~\ref{lem:split-completeness} and the soundness of the reduction in Lemma~\ref{lem:split-soundness}. 

\begin{claim}
\label{claim:element-pair-cut}
    Suppose $f(z)=\sign(\ip{h}{z}-\theta)$ satisfies $h_0 = 1$ and $|\theta|\le 1/10$. Then, for every $i\in[n]$: \begin{itemize}
        \item If $L_i$ and $R_i$ are uncut by $f$, then $|h_i|\in [1.8,3.3]$.
        \item If $|h_i| \in (2.2, 2.7)$ then $L_i$ and $R_i$ are uncut by $f$.
    \end{itemize}
\end{claim}
\begin{proof}
    Fix any $i\in[n]$. The pair $L_i$ is uncut when $\sign(1+h_i/2-\theta) = \sign(-1+h_i/2-\theta),$ which implies $|h_i/2-\theta| \ge 1$ and thus $|h_i| \ge 2-2|\theta| \ge 1.8$. Similarly, the pair $R_i$ is uncut when $\sign(1+h_i/3-\theta) = \sign(1-h_i/3-\theta),$ which implies $|1-\theta| \ge |h_i|/3$ and thus $|h_i| \le |3-3\theta| \le 3.3$. We conclude that $|h_i| \in [1.8,3.3].$ For the reverse direction, suppose that $|h_i| \in (2.2,2.7)$. Then since $|h_i/2-\theta| > 1$, we have $\sign(1+h_i/2-\theta) = \sign(-1+h_i/2-\theta),$ so $L_i$ is uncut by $f$. Similarly, $|1-\theta| > |h_i/3|$, meaning $\sign(1+h_i/3-\theta) = \sign(1-h_i/3-\theta),$ so $R_i$ is uncut by $f$.
\end{proof}

\begin{claim}
\label{claim:set-pair-cut}
    Suppose $f(z)=\sign(\ip{h}{z}-\theta)$ satisfies $h_0=1$ and $|\theta|\le 1/10$. Define the partition $$U^+ := \{i\in[n] : h_i\ge 0\},\quad U^- := \{i\in[n] : h_i < 0\}.$$Then, for any $\{i,j,k\}\in\calF$ with $|h_i|,|h_j|,|h_k|\in [1.8,3.3]$, the following holds:$$L_{i,j,k}\text{ and } R_{i,j,k}\text{ are both uncut by } f \iff \{i,j,k\}\text{ is split by } U^+\sqcup U^-.$$
\end{claim}

\begin{proof}
    Fix $\{i,j,k\}\in\calF$ and write $s = h_i+h_j+h_k$. The pairs $L_{i,j,k}$ and $R_{i,j,k}$ are both uncut if and only if $\sign(s-5-\theta) = \sign(s-11-\theta)$ and $\sign(s+5-\theta)=\sign(s+11-\theta).$ Assume first that $\{i,j,k\}$ is split by $U^+\sqcup U^-$. Then, since $|h_i|,|h_j|,|h_k|\in[1.8,3.3]$, we have $|s| \le 3.3+3.3-1.8 = 4.8$ which gives $|s-\theta| \le 4.9$. It follows that $s-11-\theta < s-5-\theta < 0$ and $0 < s+5-\theta < s+11-\theta,$ which proves that $L_{i,j,k}$ and $R_{i,j,k}$ are both uncut. Now, suppose that $\{i,j,k\}$ is not split by $U^+ \sqcup U^-$. Then, $h_i,h_j,h_k$ are either all nonnegative or all negative. If they are all nonnegative, then $s \ge 3\cdot 1.8 = 5.4$, so $s-\theta \ge 5.3$. Also, $s \le 3\cdot 3.3 = 9.9$, giving $s-\theta \le 10$. Hence, $s-5-\theta > 0$ but $s-11-\theta < 0,$ and hence $L_{i,j,k}$ is cut. If $h_i,h_j,h_k$ are instead all negative, then $s \le -3\cdot 1.8 = -5.4$ and $s \ge -3\cdot 3.3 = -9.9$. Thus, $s-\theta \in [-10,-5.3]$ and therefore $s+5-\theta<0$ but $s+11-\theta > 0$. Hence, $R_{i,j,k}$ is cut in this case.
\end{proof}

\begin{lemma}[Completeness]
\label{lem:split-completeness}
    If there exists a partition splitting all sets in $\calF$, then there exists $f\in\calP_1$ with $\err_{\Dpair}(f) = 0$ and $\bal_{\Dpair}(f) = 1/2.$
\end{lemma}

\begin{proof}
    Let $[n] = U_1\sqcup U_2$ denote the partition splitting all sets in $\calF$. Consider the halfspace $f(z) = \sign(\ip{h}{z})$ where $h\in\R^{n+1}$ is defined as follows: first, $h_0 = 1$ (the normalization), and for $i \in [n]$, $h_i = 2.5$ if $i \in U^+$ and $h_i = -2.5$ if $i \in U^-$.
    Since $2.5\in [2.2,2.7]$, Claim~\ref{claim:element-pair-cut} implies that no element pair is cut by $f$. Since every set in $\calF$ is split by $U_1\sqcup U_2$ and $2.5\in [1.8,3.3]$, Claim~\ref{claim:set-pair-cut} implies that no set pair is cut by $f$. The centering pairs are self-pairs, so they are never cut. 
    Moreover, for any original pair $(x,y)$ uncut by $f$, its negated pair $(-x,-y)$ is also uncut since $f(-z) = \sign(\ip{h}{-z}) = -f(z)$ (noting $\ip{h}{z} \ne 0$). Hence, $\err_{\Dpair}(f) = 0.$
    To see that $f$ has a balance of exactly $1/2$, observe that the $\by$-marginal distribution $\Dy$ is entirely symmetric around the origin due to the addition of the symmetrization pairs. Since $f(-y) = -f(y)$ for all points in the support, exactly half of the points in $\Dy$ evaluate to $1$ and the other half evaluate to $-1$. Thus, $\bal_{\Dpair}(f) = 1/2.$
\end{proof}

\begin{lemma}[Soundness]
\label{lem:split-soundness}
    If every partition splits at most a $24/25$-fraction of the sets in $\calF$, then every $f\in\calP_1$ with $\bal_{\Dpair}(f) \ge \beta$ satisfies $\err_{\Dpair}(f)\ge \beta / (600B).$
\end{lemma}
\begin{proof}
    Let $f(z) = \sign(\ip{h}{z}-\theta)$. We first show that $f$ must separate the two original centering pairs $L_0$ and $R_0$. Suppose instead that $f(e_0/10) = f(-e_0/10).$ Then, the copies of $L_0$, $R_0$, and their negations all evaluate to the same sign. The total number of points from the centering pairs in the marginal $\Dy$ is $4M$. The total number of points in $\Dy$ overall is $4(n+|\calF| + M)$. This implies that the fraction of points on the minority side of $f$ is bounded by the fraction of non-centering points:
    $$\bal_{\Dpair}(f) \le \frac{4(n+|\calF|)}{4(n+|\calF| + M)} = \frac{n+|\calF|}{n+|\calF| + M}.$$
    Since $T := (1-\beta) / \beta$ and $M := \lfloor T\cdot(n+|\calF|) \rfloor + 1$ during the reduction, we ensure $M > \frac{1-\beta}{\beta}(n+|\calF|)$. Substituting this gives
    $$\bal_{\Dpair}(f) < \frac{n+|\calF|}{n+|\calF| + \frac{1-\beta}{\beta}(n+|\calF|)} = \beta,$$
    which 
    would contradict the condition that $\bal_{\Dpair}(f) \ge \beta$. Hence, $f(e_0/10)\neq f(-e_0/10)$. In particular, $h_0\ne 0$, so by rescaling $(h,\theta)$ we may assume that $h_0 = 1$ and $|\theta| \le 1/10$. 
    
    Now, define the partition $U^+ := \{i\in[n] : h_i\ge 0\}$ and $U^- :=\{i\in[n] : h_i < 0\}.$ By assumption, this leaves at least $|\calF|/25$ sets unsplit. Fix any unsplit set $\{i,j,k\}\in\calF$. We claim that one of the eight original pairs $L_i,\,R_i,\,L_j,\,R_j,\,L_k,\,R_k,\,L_{i,j,k},\,R_{i,j,k}$ must be cut by $f$. Indeed, if $L_i,R_i,L_j,R_j,L_k,R_k$ are all uncut by $f$, then Claim~\ref{claim:element-pair-cut} yields $|h_i|,|h_j|,|h_k|\in[1.8,3.3]$. Then, since $\{i,j,k\}$ is unsplit by $U^+ \sqcup U^-$, Claim~\ref{claim:set-pair-cut} implies that either $L_{i,j,k}$ or $R_{i,j,k}$ is cut. Therefore, each unsplit set cuts at least one of the original pairs in $\calS_0$. A set pair can be charged at most once, whereas an element pair can be charged at most $B$ times. Therefore, the number of cut pairs is at least $|\calF| / (25B)$. 
    Meanwhile, the total number of pairs in $\Dpair$ is $P = 4(n+|\calF|+M)$. Using $M \le \frac{1-\beta}{\beta}(n+|\calF|) + 1$, we have:
    $$P \le 4(n+|\calF|) + 4\left(\frac{1-\beta}{\beta}(n+|\calF|) + 1\right) = \frac{4}{\beta}(n+|\calF|) + 4.$$
    Using $n\le 3|\calF|$, and noting that $|\calF|\ge 1$ and $\beta \le 1/2$, we can bound the total number of pairs by $P \le 16|\calF|/\beta + 4 \le 24|\calF| / \beta$.
    Therefore, $$\err_{\Dpair}(f) \ge \frac{|\calF| / (25B)}{24|\calF| / \beta} = \frac{\beta}{600B},$$ as desired.
\end{proof}

\begin{proof}[Proof of Theorem~\ref{thm:strong-hardness}]
    Lemma~\ref{lem:set-split-hard} together with Lemmas~\ref{lem:split-completeness} and~\ref{lem:split-soundness} give the theorem, with $c = 1 / (600B)$.
\end{proof}

%% file: sections/cut-algo.tex

\section{An Improper Learning Algorithm: Theorem~\ref{thm:improper-alg}}
\label{sec:balanced-cut-algo}

Theorem~\ref{thm:strong-hardness} rules out polynomial-time proper learning for the balanced halfspace cut problem (Problem~\ref{problem:balanced-halfspace-cut}). We therefore consider a relaxation in which the learner is allowed to be improper and will output a low-degree PTF instead of a halfspace. Our algorithmic framework naturally extends beyond halfspaces, to any function class $\calF$ of cuts $f : \R^d\to\{-1,1\}$, which we will refer to as a cut class. We will take advantage of this generality in Section~\ref{sec:dasgupta}.

\begin{definition}[$\calF$-Approximation Degree]
\label{def:general-approx-deg}
    Let $\calF$ be a class of cuts $f : \R^d \to \{-1,1\}$. We say a distribution $\Dpair$ over $\R^d \times \R^d$ has \emph{$\calF$-approximation degree} $k(\cdot) \colon \R \to \N$ if for every $\eps > 0$ and every $f \in \calF$, there exists a polynomial $A \colon \R^d \to \R$ of degree at most $k(\eps)$ such that for all $\ell \in \{1,2,4\}$
    \[
    \Ex_{\bx \sim \Dx}\big[|A(\bx)-f(\bx)|^\ell\big] \le \eps \qquad \text{and} \qquad \Ex_{\by \sim \Dy}\big[|A(\by)-f(\by)|^\ell\big] \le \eps.
    \]
\end{definition}

We are particularly interested in settings where the distribution $\Dpair$ satisfies $k(\eps) = \poly(1/\eps)$ for the function class $\calF$ being the set of halfspaces $\calP_1$ (and will later generalize to certain functions of halfspaces). This restricts the class of input distributions $\Dpair$ that we can handle but corresponds to common assumptions made in learning theory. We refer the reader to Appendix~\ref{sec:poly-approx}, where we show that when the two marginals of $\Dpair$ ($\Dx$ and $\Dy$) satisfy standard anticoncentration and subgaussianity conditions (generalizing properties of the standard Gaussian), then $k(\eps) = \Ot(1/\eps^2)$ for halfspaces.

\begin{theorem}
\label{thm:improper-alg}
    There exists a randomized algorithm, \ptfcut, which receives sample access to a distribution $\Dpair$ over pairs $\R^d \times \R^d$, a balance parameter $\beta \in (0, 1/2]$ and an error parameter $\eps \in (0,1)$, and satisfies:
    \begin{itemize}
        \item Let $\calF$ be a cut class and $k(\cdot) \colon \R \to \N$ be the $\calF$-approximation degree of $\Dpair$. With probability at least $9/10$, \ptfcut outputs a PTF $\bg \in \calP_{k(\eps)}$ satisfying
        \[\err_{\Dpair}(\bg) \le O\del{\sqrt{(\opt_\beta(\calF)+\eps)/\beta}} \qquad\text{and}\qquad\bal_{\Dpair}(\bg) \ge \Omega(\beta^2).\]
        \item \ptfcut uses $O((d+k(\eps))^{k(\eps)} (1/\eps^2 + 1/\beta^4))$ samples from $\Dpair$ and runs in time $\poly((d+k(\eps))^{k(\eps)}, 1/(\beta\eps))$.
    \end{itemize}
\end{theorem}

Theorem~\ref{thm:improper-alg} gives the improper learner which outputs a low-degree PTF. It is useful to think of the balance parameter $\beta$ as a constant (it will be set to $1/3$ in our nearest neighbor application) and $\eps$ as a small additive error. In this case, \ptfcut outputs a $\Omega(1)$-balanced PTF whose degree is at most $k(\eps)$ and achieves error $O(\sqrt{\opt_{1/3}(\calF) + \eps})$. In Corollary~\ref{cor:square-root-tight}, we show that the square-root dependence on $\opt_\beta(\calF)$ in the error term is optimal assuming the Small Set Expansion Hypothesis (\sseh), even when $\Dpair$ has $\calF$-approximation degree $k(\eps) = 1$.\footnote{See Footnote~\ref{foot-open}.}

\subsection{Proof of Theorem~\ref{thm:improper-alg}}

In this subsection, we prove Theorem~\ref{thm:improper-alg} assuming Lemma~\ref{lem:sdp-algorithm} which we prove in Subsection~\ref{subsec:sdp-algorithm}. We adopt the following notation regarding samples from $\Dpair$. For a parameter $m \in \N$ which will be our sample complexity, we let $\calbS \sim \Dpair^m$ denote the collection of $m$ independent samples from $\Dpair$. We abuse notation to let $\calbS$ denote the uniform distribution over these samples. We use $\Sx$ and $\Sy$ to denote the $\bx$- and $\by$-marginals of $\calbS$. We will make use of the following two notions, the energy and kurtosis of a function with respect to a distribution over pairs.

\begin{definition}[Energy]
\label{def:rayleigh-quotient}
    Given $F : \R^d\to\R$, we define the following \emph{energy function} with respect to $\Dpair$:
    \[ \Rpair(F) := \frac{\Ex_{(\bx,\by)\sim\Dpair}[(F(\bx)-F(\by))^2]}{\Var_{\by\sim\Dy}[F(\by)]}. \]
\end{definition}

\begin{definition}[Kurtosis]
    Let $\calM$ be a distribution over $\R^d$. Let $F : \R^d\to\R$ with $\mu := \Ex_{\by\sim\calM}[F(\by)]$. The \emph{kurtosis} of $F$ is given by the following ratio: \[\kurt(F;\,\calM) := \frac{\Ex_{\by\sim\calM}[(F(\by)-\mu)^4]}{\Var_{\by\sim\calM}[F(\by)]}\]
\end{definition}

We will divide the proof of Theorem~\ref{thm:improper-alg} into the following steps. \begin{itemize}
    \item \textbf{Low energy implies sparse cut.} Consider a fixed set of samples $\calbS \sim \Dpair^m$. In Lemma~\ref{lem:shift-existence}, we show that any function $F \colon \R^d \to \R$ with energy $\calR_{\calbS}(F)$ and constant kurtosis on $\Sy$ (the second marginal of $\calbS$) admits a threshold cut with error $\approx \sqrt{\calR_{\calbS}(F)}$ and constant balance. The proof follows from a Cheeger-style rounding argument.
    \item \textbf{Existence of low-energy polynomial over the distribution.} In Lemma~\ref{lem:polynomial-existence}, we use the definition of $\calF$-approximation degree to show the existence of a degree-$k(\eps)$ polynomial $A^\star$ with constant kurtosis and $\Rpair(A^\star) \le O(\opt_{\beta}(\calF) + \eps)$.
    \item \textbf{Computing a low-energy polynomial over samples.} This step is the core algorithmic contribution. Lemma~\ref{lem:sdp-algorithm} gives an algorithm, based on a semidefinite programming relaxation followed by a rounding algorithm, to find a polynomial $A$ of degree $k(\eps)$ with $\calR_{\calbS}(A) \le O(\opt_{\beta}(\calF) + \eps)$ and constant kurtosis on $\calbS_{Y}$ for $\calbS \sim \Dpair^m$.
    \item \textbf{Generalizing back to the distribution.} We conclude that some threshold cut of $A$ has error $O(\sqrt{\opt_{\beta}(\calF)+\eps})$ and constant balance on $\calbS$. In Lemma~\ref{lem:uniform-convergence}, we give a uniform convergence argument showing that the same holds for the true distribution $\Dpair$.
\end{itemize}

\begin{restatable}[Cheeger-style Rounding, proof in Appendix~\ref{subsec:cheeger-rounding}]{lemma}{shiftExistence}
\label{lem:shift-existence}
    For any $F : \R^d\to\R$ and any distribution $\calS$ over pairs in $\R^d \times \R^d$, there is a threshold $\theta\in\R$ such that the PTF $g_{F, \theta} \colon \R^d \to \{-1,1\}$ given by $g_{F, \theta}(x) = \sign(F(x) - \theta)$ satisfies
    \[\err_{\calS}(g_{F,\theta}) \le O\!\del{\sqrt{\calR_{\calS}(F)}} \quad\text{and}\quad \bal_{\calS}(g_{F,\theta}) \ge \Omega\del{\frac{1}{\kurt(F;\,\calS_{Y})}}.\]
\end{restatable}
    
\begin{lemma}[Low-Energy and Kurtosis Polynomial Finder, proof in Subsection~\ref{subsec:sdp-algorithm}]
\label{lem:sdp-algorithm}
    There is a randomized algorithm that receives $\calbS \sim \Dpair^m$ where $\Dpair$ has $\calF$-approximation degree $k(\cdot)$, parameters $\beta \in (0,1/2]$ and $\eps \in (0, \beta/8)$, and outputs in $\poly((d+k(\eps))^{k(\eps)}, m, 1/\eps)$ a degree-$k(\eps)$ polynomial $\bA$ satisfying
        \[
    \calR_{\calbS}(\bA) \le O(1/\beta)\cdot (\opt_\beta(\calF)+\eps)
    \quad\text{and}\quad
    \kurt(\bA;\,\Sy) \le O(1/\beta^2).
    \]
    with probability at least $1/10$.
\end{lemma}

\begin{restatable}[Uniform Convergence, proof in Appendix~\ref{subsec:proof-uniform-convergence}]{lemma}{uniformConvergence}
\label{lem:uniform-convergence}
    For any $\eps \in (0,1)$ and $k \in \N$, let $\calbS \sim\Dpair^m$ where $m = O(((d+k)^k+\log(1/\delta))/ \eps^2)$. Then, with probability $1-\delta$, every $g\in\calP_k$ satisfies \[|\err_{\calbS}(g) - \err_{\Dpair}(g)| \le \eps\quad\text{and}\quad |\bal_{\calbS}(g) - \bal_{\Dpair}(g)| \le \eps.\]
\end{restatable}

\begin{proof}[Proof of Theorem~\ref{thm:improper-alg} assuming above lemmas]
    Set $k = k(\eps)$ to simplify notation, and notice it suffices to consider the case where $\eps$ is a small enough constant factor of $\beta$ (otherwise, the required guarantee on $\err_{\Dpair}(\bg)$ becomes $1$ and is trivially satisfied). We repeat the following procedure $T = O(\log(1/\delta))$ times for a success probability $1-\delta$, and output the overall minimum:
    \begin{enumerate}
        \item Draw a sample $\calbS \sim \Dpair^m$ of size $m = O((d+k)^k \log T/\min\{\eps, \beta^2\}^2)$ (where we will use Lemma~\ref{lem:uniform-convergence} with error guarantee $\min\{ \eps, c_0 \beta^2 \}$ for a small enough constant $c_0 > 0$) and error probability $\delta/(2T)$, establishing the desired sample complexity.
        \item Run the algorithm from Lemma~\ref{lem:sdp-algorithm} on $\calbS$ to obtain a degree-$k$ polynomial $\bA \colon \R^d \to \R$.
        \item Evaluate $\bA(\bx)$ and $\bA(\by)$ for all $(\bx,\by) \in \calbS$. Sort these $2m$ values to obtain $2m+1$ candidate thresholds.
        \item Iterate over these thresholds to find $\theta \in \R$ and let $g_{\bA, \theta}\colon \R^d \to \{-1,1\}$ be $g_{\bA,\theta}(x) = \sign(\bA(x) - \theta)$ minimizing $\err_{\calbS}(g_{\bA, \theta})$ subject to $\bal_{\calbS}(g_{\bA,\theta}) \ge \Omega(\beta^2)$.
    \end{enumerate}

Finally, we output the degree-$k$ PTF $\bg_{\bA, \theta}$ achieving the minimal empirical error $\err_{\calbS}(\bg_{\bA, \theta})$ over all $T$ runs. To analyze the correctness of the algorithm, we first assume that for all $T$ repetitions of the execution above, the conclusions of Lemma~\ref{lem:uniform-convergence} are satisfied (and note this occurs except with probability at most $\delta/2$ via a union bound). Then, we consider a particular repetition among the $T$ where the conclusions of Lemma~\ref{lem:sdp-algorithm} are satisfied, which occurs except with probability $\delta/2$. We henceforth analyze the execution where conclusions of Lemma~\ref{lem:sdp-algorithm} hold.

By Lemma~\ref{lem:sdp-algorithm}, the degree-$k$ polynomial $\bA$ returned by Step 2 satisfies $\calR_{\calbS}(\bA) \le O(1/\beta) \cdot (\opt_\beta(\calF) + \eps)$  and $\kurt(\bA; \Sy) \le O(1/\beta^2)$. 
Applying Lemma~\ref{lem:shift-existence}, there exists a threshold $\theta \in \R$ such that 
\[\err_{\calbS}(g_{\bA, \theta}) \le O\del{\sqrt{(\opt_\beta(\calF) + \eps) / \beta}} \quad \text{and} \quad \bal_{\calbS}(g_{\bA, \theta}) \ge \Omega(\beta^2).\] 
Since Step 4 checks all relevant thresholds, the PTF returned $g_{\bA, \theta}$ satisfies these bounds. By Lemma~\ref{lem:uniform-convergence}, the above guarantee can be translated, from $\err_{\calbS}(g_{\bA, \theta})$ and $\bal_{\calbS}(g_{\bA,\theta})$ to $\err_{\Dpair}(g_{\bA, \theta})$ and $\bal_{\Dpair}(g_{\bA, \theta})$, up to an additive error $\eps$ (which is less than $\sqrt{\eps / \beta}$) for the error guarantee and $c_0\beta^2$ for arbitrarily small $c_0$ on the balance guarantee. Since Lemma~\ref{lem:uniform-convergence} always holds, the output PTF $\bg$ satisfies $\bal_{\Dpair}(\bg) = \Omega(\beta^2)$, and $\err_{\Dpair}(\bg)$ is at most the minimum over all $T$ repetitions of $\err_{\calbS}(g_{\bA, \theta}) + \eps$, satisfies the desired guarantees since at least one of the $T$ repetitions is appropriately bounded.
\end{proof}

\subsection{Proof of Lemma~\ref{lem:sdp-algorithm}}
\label{subsec:sdp-algorithm}

The proof of Lemma~\ref{lem:sdp-algorithm} proceeds by solving and then rounding an appropriate semidefinite program. The semidefinite program is a relaxation of an optimization problem over the space of polynomials; hence, we will first define the idealized optimization problem over polynomials to establish that there exists a feasible solution for the subsequent semidefinite programming relaxation. We use the shorthand $k := k(\eps)$.

\begin{definition}\label{def:Pi}
Let $\calA_k$ denote the set of polynomials $\R^d \to \R$ of degree at most $k$. For a parameter $C > 1$, and a distribution over pairs $\calS$, we let $\Pi(\calS, C)$ be given by 
\begin{align*}
    \min_{A \in \calA_k} \Ex_{(\bx,\by)\sim\calS}[ (A(\bx) - A(\by))^2 ] \qquad \text{s.t.} \qquad\begin{array}{ll} \text{(i)} & \Ex_{\by\sim \calS_Y}[A(\by)] = 0 \\
            \text{(ii)} & \Ex_{\by\sim \calS_{Y}}[A(\by)^2] = 1 \\
            \text{(iii)} & \Ex_{\by\sim \calS_{Y}}[A(\by)^4] \leq C \end{array} .
\end{align*}
\end{definition}

The subsequent lemma, which we prove in Appendix~\ref{subsec:polynomial-approx}, directly implies that the optimal solution to $\Pi(\Dpair, c /\beta^2)$ (for a large enough constant factor $c > 0$) is at most $O(1/\beta) \cdot (\opt_{\beta}(\calF) + \eps)$. The translation comes from the fact that constraints (i) and (ii) in Definition~\ref{def:Pi} normalize the polynomial to variance under $\Dpair$ being $1$ (which fixes the denominator in both $\Rpair(A^{\star})$ and $\kurt(A^{\star}, \Dy)$ to $1$).

\begin{restatable}[Low-Energy and Kurtosis Polynomial, proof in Appendix~\ref{subsec:polynomial-approx}]{lemma}{polynomialExistence}  
\label{lem:polynomial-existence}
    Whenever $\Dpair$ has $\calF$-approximation degree $k(\cdot)$, for any $\beta \in (0,1/2]$ and $\eps\in(0, \beta/8)$, there exists a degree-$k(\eps)$ polynomial $A^\star$ over $\R^d$ satisfying \[\calR_{\Dpair}(A^\star) = O(1/\beta) \cdot (\opt_\beta(\calF)+\eps)\quad\text{and}\quad\kurt(A^\star;\,\Dy) \le O(1/\beta^2).\]
\end{restatable}

\begin{claim}
\label{claim:ideal-to-empirical}
    Let $A^\star$ be an optimal solution to $\Pi(\Dpair, C)$ with objective value $\eta$. Then if $m \ge 4000C$, with probability $0.9$ over $\calbS \sim \Dpair^m$, there exists a feasible solution $A$ to $\Pi(\calbS, 500C)$ whose objective value is $O(\eta)$.
\end{claim}
\begin{proof}
    Consider the random variables $\hat{\boldeta}, \hat{\bmu}, \hat{\bs}_{2}$ and $\hat{\bs}_4$ defined as a function of the draw $\calbS \sim \Dpair^m$ given by 
    \[\hat\boldeta := \Ex_{(\bx,\by)\sim \calbS}\left[ (A^\star(\bx)-A^\star(\by))^2\right] \quad\hat\bmu := \Ex_{\by\sim \Sy}\left[A^\star(\by)\right]\quad\hat\bs_2 := \Ex_{\by \sim \Sy}\left[ A^\star(\by)^2 \right] \quad\hat\bs_4 := \Ex_{\by\sim \calbS_{Y}}\left[A^\star(\by)^4\right], \]
    and consider the following four events, (I) $\hat{\boldeta} \leq 40 \eta$, (II) $\hat{\bs}_4 \leq 40C$, (III) $|\hat{\bmu}| \leq 0.1$ and (IV) $|\hat{\bs}_2 - 1| \leq 0.1$. We will show that all four events hold with probability at least $0.9$, but when they do, we now show that the polynomial $A \colon \R^d \to \R$ given by $A(x) = (A^{\star}(x) - \hat{\bmu}) / (\hat{\bs}_2 - \hat{\bmu}^2)^{1/2}$ is a feasible solution to $\Pi(\calbS, 500C)$ with objective at most $45 \eta$.
    \begin{align*}
        \Ex_{(\bx,\by) \sim \calbS}\left[ (A(\bx) - A(\by))^2\right] &= \frac{1}{\hat{\bs}_2 - \hat{\bmu}^2} \Ex_{(\bx,\by)\sim \calbS}\left[ (A^{\star}(\bx) - A^{\star}(\by))^2 \right] \leq \frac{1}{0.9 - 0.1^2} \cdot \hat{\boldeta} \leq \frac{40 \eta}{0.89} \leq 45 \eta. \\
        \Ex_{\by \sim \Sy}\left[ A(\by) \right] &= \frac{1}{(\hat{\bs}_2 - \hat{\bmu}^2)^{1/2}} \left( \Ex_{\by \sim \Sy}\left[A^{\star}(\by) \right] - \hat{\bmu} \right) = 0 \\
        \Ex_{\by \sim \Sy}\left[ A(\by)^2\right] &= \frac{1}{\hat{\bs}_2 - \hat{\bmu}^2} \Ex_{\by \sim \Sy}\left[ (A^{\star}(\by) - \hat{\bmu})^2 \right] = \frac{1}{\hat{\bs}_2 - \hat{\bmu}^2} \left(\Ex_{\by \sim \Sy}\left[ A^{\star}(\by)^2 \right] - \hat{\bmu}^2 \right) = 1 \\
        \Ex_{\by \sim \Sy}\left[ A(\by)^4\right] &= \frac{1}{(\hat{\bs}_2 - \hat{\bmu}^2)^2} \Ex_{\by \sim \Sy}\left[ (A^{\star}(\by) - \hat{\bmu})^4 \right] \leq \frac{8}{(0.9 - 0.1^2)^2} \cdot \left(\Ex_{\by \sim \Sy}\left[ A^{\star}(\by)^4\right] + \hat{\bmu}^4 \right) \leq 500 C,
    \end{align*}
    where the last line used $(a - b)^4 \leq 8a^4 + 8b^4$. These establish that $A$ is feasible for $\Pi(\calbS, 500C)$ and has objective value $\le 45 \eta$. For the events: (I) occurs except with probability at most $1/40$, since $\hat{\boldeta}$ is non-negative and its expectation is $\eta$ (since $\calbS$ is drawn independently from $\Dpair$); similarly, (II) occurs except with probability $1/40$ since $\hat{\bs}_4$ is non-negative with expectation at most $C$. For (III), note that $\hat{\bmu}$ has expectation $\Ex_{\by\sim \Dy}[A^{\star}(\by)] = 0$, and its variance is $(1/m) \Ex_{\by\sim\Dy}[A^{\star}(\by)^2] = 1/m$. By Chebyshev's inequality, (III) occurs except with probability $100/m \le 1/40$. Similarly, for (IV), the expectation of $\hat{\bs}_2$ is $\Ex_{\by \sim \Dy}[A^{\star}(\by)^2]=1$, and its variance at most $C/m$, so Chebyshev's inequality implies (IV) holds except with probability $100C/m \le 1/40$.
\end{proof}

To define the SDP relaxation for $\Pi(\cdot,\cdot)$, we first introduce some useful notation. For any $\alpha \in \Z_{\geq 0}^d$, we let $x^{\alpha}$ to refer to the monomial $\prod_{i=1}^d x_i^{\alpha_i}$ and $|\alpha| = \sum_{i=1}^d \alpha_i$ to the degree of the monomial. Note that there are $N = \binom{d+k}{k}$ monomials $\alpha \in \Z_{\geq 0}^d$ of degree at most $k$. We let $\phi_k(z) \in \R^{N}$ denote the embedding $\phi_k(z)_{\alpha} = z^{\alpha}$ among $|\alpha| \leq k$, and note that every degree-$k$ polynomial $A \colon \R^d \to \R$ can be uniquely represented by $A(x) = w_A^{\intercal} \phi_k(x)$, where $w_A \in \R^N$ denotes the vector of coefficients in each of the monomials.

\begin{definition}[SDP Relaxation for $\Pi(\cdot, \cdot)$]
     For the collection of sampled pairs $\calbS = \{(\bx_i,\by_i)\}_{i=1}^m$, define the vectors in $\R^N$ given by $\bV_i := \phi_k(\bx_i) - \phi_k(\by_i)$ and $\bY_i := \phi_k(\by_i)$ and $\bar{\bY} := \frac{1}{m}\sum_{i=1}^m {\bY_i}$. We let $\Pi_{\textup{sdp}}(\calbS, C)$ be
     \begin{align*}
         \min_{M \in \R^{N\times N}} \frac{1}{m} \sum_{i=1}^m \bV_i^{\intercal} M \bV_i \qquad \text{s.t.} \qquad \begin{array}{ll} \text{(i)} & \bar{\bY}^{\intercal} M \bar{\bY} = 0 \\
                            \text{(ii)} & (1/m) \sum_{i=1}^m \bY_i^{\intercal} M \bY_i = 1 \\
                            \text{(iii)} & (1/m) \sum_{i=1}^m (\bY_i^{\intercal} M \bY_i)^2 \leq C \\
                            \text{(iv)} & M \succeq 0 \end{array}. 
     \end{align*}       
\end{definition}

\begin{claim}
\label{claim:sdp-solvable}
    There is a $\poly((d+k)^k, m, \log(1/\eps))$-time algorithm which receives $\calbS \sim \Dpair^m$ and outputs a psd matrix $M \in \R^{N \times N}$ which is feasible for $\Pi_{\textup{sdp}}(\calbS, 500c/\beta^2)$ for a large enough constant $c > 1$. As long as $m \geq 4000 c/\beta^2$, with probability at least $0.9$ over $\calbS \sim \Dpair^m$,
    \[ \frac{1}{m} \sum_{i=1}^m \bV_i^{\intercal} M \bV_i \leq O(1/\beta) \cdot (\opt_{\beta}(\calF) + \eps). \]
\end{claim}
\begin{proof}
    We first verify that condition (iii) can be formulated as a semidefinite program. We linearize this by introducing auxiliary variables $t_1, \dots, t_m \in \R$ and replacing the constraint with $(1/m)\sum_{i=1}^m t_i \le C$ and $t_i \ge (\bY_i^\top M\, \bY_i)^2$ for all $i \in [m]$. Indeed, each of the latter $m$ inequalities can be written as \[\begin{pmatrix} t_i & \bY_i^\top M\, \bY_i \\ \bY_i^\top M\, \bY_i & 1 \end{pmatrix} \succeq 0.\] The existence of a $\poly((d+k)^k, m, \log(1/\eps))$-time algorithm to solve $\Pi_{\textup{sdp}}(\calbS, 500c/\beta^2)$ to additive error $\eps$ (which can be absorbed in the asymptotic notation) follows from standard SDP solvers~\cite{jiang2020}. To argue the desired quality, we first note that by Lemma~\ref{lem:polynomial-existence} and Claim~\ref{claim:ideal-to-empirical}, with probability at least $0.9$ over $\calbS \sim \Dpair^m$, there is a feasible polynomial $A \colon \R^d \to \R$ to $\Pi(\calbS, 500c / \beta^2)$ whose objective value is $O(1/\beta) \cdot (\opt_{\beta}(\calF) + \eps)$. Then, if we consider the rank-$1$ $N \times N$ matrix $M^{\star} = w_Aw_A^{\intercal}$ where $w_A \in \R^N$ is the vector of coefficients of monomials of $A$, then
    \begin{align*}
        \frac{1}{m} \sum_{i=1}^m \bV_i^{\intercal} M^{\star} \bV_i &= \frac{1}{m} \sum_{i=1}^m (w_A^{\intercal} \bV_i)^2 = \frac{1}{m} \sum_{i=1}^m (w_A^{\intercal} \phi_k(\bx_i) - w_A^{\intercal} \phi_k(\by_i))^2 = \Ex_{(\bx,\by)\sim\calbS}\left[ (A(\bx) - A(\by))^2 \right],\\
        \bar{\bY}^{\intercal} M^{\star} \bar{\bY} &= \left( \frac{1}{m} \sum_{i=1}^m w_{A}^{\intercal} \bY_i\right)^2 = \left( \frac{1}{m} \sum_{i=1}^m A(\by_i) \right)^2 = \Ex_{\by \sim \Sy}\left[ A(\by) \right]^2 = 0 \\
        \frac{1}{m} \sum_{i=1}^m \bY_i^{\intercal} M^{\star} \bY_i &= \frac{1}{m} \sum_{i=1}^m (w_{A}^{\intercal} \phi_k(\by_i))^2 = \Ex_{\by \sim \Sy}\left[ A(\by)^2 \right] = 1, \\
        \frac{1}{m} \sum_{i=1}^m (\bY_i^{\intercal} M^{\star} \bY_i)^2 &= \frac{1}{m} \sum_{i=1}^m (w_{A}^{\intercal} \phi_k(\by_i))^4 = \Ex_{\by \sim \Sy}\left[ A(\by)^4 \right] \leq 500 c/\beta^2,
    \end{align*}
    which gives the desired bound. 
\end{proof}

\begin{claim}
\label{claim:sdp-rounding}
    Let $M \in \R^{N \times N}$ be a psd matrix satisfying the conclusions of Claim~\ref{claim:sdp-solvable}. If $\bzeta\sim\calN(0,M)$, then with probability $1/8$ over $\bzeta$, the polynomial $A_{\bzeta}(x) := \bzeta^\top\phi_k(x)$ satisfies \[\calR_{\calbS}(A_{\bzeta}) = O(1/\beta) \cdot (\opt_{\beta}(\calF) + \eps) \quad\text{and}\quad\kurt(A_{\bzeta};\,\Sy) \le O(1/\beta^2).\] 
\end{claim}
\begin{proof}
    Consider the following random variables, where we consider a fixed setting of $\calbS = \{ (\bx_i, \by_i) \}_{i=1}^m$ and define the random variables with respect to the randomness in $\bzeta\sim \calN(0, M)$: 
    \[\boldeta_{\bzeta} := \frac{1}{m}\sum_{i=1}^m (A_{\bzeta}(\bx_i) - A_{\bzeta}(\by_i))^2 \quad \bmu_{\bzeta} := \frac{1}{m}\sum_{i=1}^m A_{\bzeta}(\by_i)\quad \bs_{2,\bzeta} := \frac{1}{m}\sum_{i=1}^m A_{\bzeta}(\by_i)^2\quad \bs_{4,\bzeta} := \frac{1}{m}\sum_{i=1}^m A_{\bzeta}(\by_i)^4.\]
    Note that $\boldeta_{\bzeta}$ is non-negative and has expectation exactly $(1/m) \sum_{i=1}^m \bV_i^{\intercal} M \bV_i$, which is at most $O(1/\beta) \cdot (\opt_{\beta}(\calF) + \eps)$ by Claim~\ref{claim:sdp-solvable} and hence $\boldeta_{\bzeta}$ is of the same order with high constant probability, by Markov's inequality. We will also have that $\bmu_{\bzeta} = 0$ with probability $1$. To see this, note
    \begin{align*}
        \Ex_{\bzeta}[\bmu_{\bzeta}^2] = \Ex_{\bzeta}\left[\del{\frac{1}{m}\sum_{i=1}^m A_{\bzeta}(\by_i)}^2\right] = \Ex_{\bzeta}\left[\left(\frac{1}{m} \sum_{i=1}^m \bzeta^{\intercal} \bY_i \right)^2\right] = \Ex_{\bzeta}[\bar{\bY}^\top \bzeta \bzeta^\top \bar{\bY}] = \bar{\bY}^\top M\, \bar{\bY} = 0,
    \end{align*}
    and therefore $\bmu_{\bzeta}$ is always $0$. To upper bound $\bs_{4,\bzeta} = O(1/\beta^2)$ with high constant probability, we use the fact that if $u\sim\calN(0,\sigma^2)$ then $\E[u^4] = 3\sigma^4$ to compute the expectation of $\bs_{4,\bzeta}$ below and apply Markov's inequality: 
    \[\Ex_{\bzeta}[\bs_{4,\bzeta}] = \frac{1}{m} \sum_{i=1}^m \Ex_{\bzeta}[A_{\bzeta}(\by_i)^4] = \frac{1}{m}\sum_{i=1}^m \E_{\bzeta}[(\bzeta^\top \bY_i)^4] = \frac{1}{m}\sum_{i=1}^m 3\cdot (\bY_i^\top M\, \bY_i)^2 \le 3 \cdot 500 c/ \beta^2.\] 
    Lastly, $\bs_{2,\bzeta}$ is a non-negative random variable which we show below has expectation $1$ and second moment at most $3$. 
    \begin{align*}
        \Ex_{\bzeta}[\bs_{2,\bzeta}] &= \Ex_{\bzeta}\left[\frac{1}{m}\sum_{i=1}^m A_{\bzeta}(\by_i)^2 \right] = \frac{1}{m}\sum_{i=1}^m \bY_i^\top M\, \bY_i = 1,
    \end{align*}
     To compute the second moment, we first note that for any $i \in [m]$, the fact that $A_{\bzeta}(\by_i) = \bzeta^{\intercal} \bY_i$ means that $A_{\bzeta}(\by_i)$ is, itself, distributed as a Gaussian with mean $0$ and variance $\bY_i^{\intercal} M \bY_i$. We thus obtain
     \begin{align*}
        \Ex_{\bzeta}[\bs_{2,\bzeta}^2] = \frac{1}{m^2}\sum_{i=1}^m \sum_{j=1}^m \E_{\bzeta}[A_{\bzeta}(\by_i)^2 \cdot A_{\bzeta}(\by_j)^2] &\leq \frac{1}{m^2} \sum_{i=1}^m \sum_{j=1}^m \left( \Ex_{\bzeta}[A_{\bzeta}(\by_i)^4] \Ex_{\bzeta}[A_{\bzeta}(\by_j)^4]\right)^{1/2} \\
                &\leq \frac{3}{m^2} \sum_{i=1}^m \sum_{j=1}^m (\bY_i^{\intercal} M \bY_i) \cdot (\bY_j^{\intercal} M \bY_j) = 3 \left(\frac{1}{m} \sum_{i=1}^m \bY_i^{\intercal} M \bY_i \right)^2 = 3.
     \end{align*}
    Applying the Paley-Zygmund inequality, we will conclude $\bs_{2,\bzeta} > 1/10$ with probability at least $1/4$. Via a union bound, with probability at least $1/8$ all events hold, in which case we obtain an upper bound on both $\calR_{\calbS}(A_{\bzeta})$ and $\kurt(A_{\bzeta}; \calbS_{Y})$, only losing a constant factor (from the re-normalization of the variance of $A_{\bzeta}(\by)$ over $\by \sim \calbS_{Y}$ potentially decreasing to $1/10$ and the fourth moment increasing by a constant factor).
\end{proof}

\begin{proof}[Proof of Lemma~\ref{lem:sdp-algorithm}]
    The probability that the conditions of Claim~\ref{claim:sdp-solvable} and Claim~\ref{claim:sdp-rounding} succeed is at least $0.9\cdot 1/8 > 1/10$, since $\bzeta$ is drawn independently from $\calbS$. The runtime of the algorithm is $\poly((d+k)^k, m, 1/\eps)$.
\end{proof}

\subsection{Tightness of Theorem~\ref{thm:improper-alg}}

We note that Theorem~\ref{thm:improper-alg} is essentially tight, in the sense that the square-root dependence in the error guarantee cannot be improved. In Corollary~\ref{cor:square-root-tight}, we show that the square-root dependence on $\opt_\beta(\calF)$ in the error term is optimal assuming the Small Set Expansion Hypothesis (\sseh), even when every cut can be written as a degree-$1$ polynomial (i.e., when $\calF$ is the class of all cuts and $\Dpair$ has $\calF$-approximation degree $k(0) = 1$). To establish this, we appeal to the following result of~\cite{raghavendra2012} regarding the hardness of the balanced cut problem in graphs.

\begin{lemma}[Corollary~IV.6 of~\cite{raghavendra2012}]
\label{lem:graph-hardness}
    Let $G=(V,E)$ be a graph. There exists a constant $c>0$ such that for any
    $\alpha>0$, it is \nphard (assuming \sseh) to distinguish between the following
    cases:
    \begin{itemize}
        \item There exists a cut $(S,V\setminus S)$ with $|S|=|V|/2$ such that
        \[
        \frac{|E(S,V\setminus S)|}{|E|} \le \alpha + o(\alpha).
        \]
    
        \item Every cut $(S,V\setminus S)$ with $|S|\in [\frac{|V|}{10}, 
        \frac{9|V|}{10}]$ satisfies
        \[
        \frac{|E(S,V\setminus S)|}{|E|} \ge c\sqrt{\alpha}.
        \]
    \end{itemize}
\end{lemma}

\begin{corollary}
\label{cor:square-root-tight}
Given a set of pairs $S = \{(x_1,y_1),\ldots,(x_m,y_m)\} \subset \R^d \times \R^d$, let $\Dpair$ be the uniform distribution over $S$, and let $\calF$ be the class of all functions $f : \R^d\to\{-1,1\}$. Even if $\Dpair$ has $\calF$-approximation degree $k(0) = 1$, there exists a constant $c>0$ such that for any $\alpha>0$, it is \nphard (assuming \sseh) to distinguish between the following cases:
\begin{itemize}
    \item There exists $f\in\calF$ such that
    \[
    \err_{\Dpair}(f) \le \alpha + o(\alpha).
    \]

    \item Every cut $f \in \calF$ with $\bal_{\Dpair}(f) \ge 1/10$ satisfies
    \[
    \err_{\Dpair}(f) \ge c\sqrt{\alpha}.
    \]
\end{itemize}
\end{corollary}
\begin{proof}
Let $G=(V,E)$ be a graph with $|V|=n$. Embed each vertex $i\in V$ as the
standard basis vector $e_i\in\R^n$. Let $\Dpair$ be the uniform distribution
over pairs $(e_i,e_j)$ corresponding to edges $(i,j)\in E$. Let $\calF$ be the class of all Boolean functions $f : \{e_1, \dots, e_n\} \to \{-1,1\}$.

For any cut $(S,V\setminus S)$, define the corresponding function $f_S \in \calF$ as
\[
f_S(e_i) =
\begin{cases}
1 & i\in S,\\
-1 & i\notin S .
\end{cases}
\]
Because the support vectors $\{e_1,\dots,e_n\}$ form the standard basis, they are linearly independent. Consequently, \emph{every} Boolean assignment $f \in \calF$ can be interpolated exactly by a linear function (a degree-$1$ polynomial). This implies that $\Dpair$ trivially has $\calF$-approximation degree $k(0)=1$.

Under this embedding, the error of any function $f_S \in \calF$ corresponds exactly to the fraction of cut edges:
\[
\err_{\Dpair}(f_S)
=
\Pr_{(i,j)\sim E}[f_S(e_i)\neq f_S(e_j)]
=
\frac{|E(S,V\setminus S)|}{|E|},
\]
and its balance is $\bal_{\Dpair}(f_S)=\min\{|S|/n, 1 - |S|/n\}$. Therefore, the two cases in the statement of the corollary map exactly to the two cases in Lemma~\ref{lem:graph-hardness}.
\end{proof}

%% file: sections/dasgupta.tex

\newcommand{\TestMass}{\textsf{\textup{TestMass}}\xspace}

\section{Learning Balanced Halfspace Cut Trees} \label{sec:dasgupta}

In this section, we use the \ptfcut algorithm from Theorem~\ref{thm:improper-alg} to build cut trees for nearest neighbor search. Throughout the section, we will use the following notation:
\begin{itemize}
    \item $P\subset\R^d$ refers to a size-$n$ dataset, and $\Dq$ is a distribution over \emph{queries} $\bq \in \R^d$.
    \item $\Dpair$ refers to the induced distribution over nearest neighbor pairs $(\bq,\nn)\in \R^d\times P$, given by 
    \[\bq\sim\Dq\quad\text{and}\quad\nn = \NN(\bq) := \argmin_{p \in P} \| \bq - p\|_2.\]
    \item $\Dnn$ refers to the $\nn$-marginal of $\Dpair$.
\end{itemize}

\begin{definition}[Cut Tree]
\label{def:cut-tree}
A cut tree $T$ for $\R^d$ is a rooted binary tree where each internal node $u$ is labeled with a cut $f_u : \R^d \to \{-1,1\}$, inducing a ``top-down'' hierarchical partition of $\R^d$. 
\begin{itemize}
    \item For a node $u$ in $T$, $\calR_{T}(u) \subset \R^d$ denotes the set of points $x \in \R^d$ that, walking down $T$ and evaluating cuts, arrive at $u$. We let $P^u = P \cap \calR_{T}(u)$.
    \item We say $T$ is $\beta$-balanced if for every internal node $u$ with child $v$, $\Prx_{\bp \sim \Dnn}[ \bp \in P^{v}] \geq \beta \Prx_{\bp \sim \Dnn}[ \bp \in P^u]$.
    \item We say $T$ is perfect if, conditional on $\bq,\nn\in\calR_T(u)$, one has $\Pr[f_u(\bq)=f_u(\nn)] = 1$.
    \item We say a cut $f_u \colon \R^d \to \{-1,1\}$ has complexity $\tau$ if evaluating $f_u$ may be done in time at most $\tau$, and we say the tree has complexity $\tau$ if all internal cuts have complexity at most $\tau$.
\end{itemize}
\end{definition}

We begin by computing a cost function for cut trees similar to the Dasgupta cost function for hierarchical clustering~\cite{dasgupta2016}.

\begin{definition}[Cost of a Cut Tree]
\label{def:tree-cost}
    The \emph{cost} of $T$ with respect to $\Dpair$ is defined as: 
    \[
    \cost_{\Dpair}(T) := \Ex_{(\bq,\nn)\sim\Dpair}\left[\Prx_{\bp \sim \Dnn}\left[ \bp \in P^{\lca_T(\bq, \nn)}\right]\right] \in [0,1].
    \] 
    where $\lca_T(\bq, \nn)$ is the least common ancestor of the leaves (in $T$) to which $\bq$ and $\nn$ are routed.
\end{definition}

One way to interpret $\cost_{\Dpair}(T)$ (useful for nearest neighbor search) is to first consider the case $\Dnn$ is uniform on $P$. In this case, $n \cdot \cost_{\Dpair}(T)$ is the expected number of comparisons required to find $\nn$ by an ``omniscient'' search algorithm, i.e., one that stops its traversal at the deepest node containing $\nn$.

\begin{claim}
\label{claim:cost-to-search}
    Let $T$ be a $\beta$-balanced cut tree of complexity $\tau$ and depth $d_T$. For any $\delta \in (0,1)$, there exists an algorithm with the following guarantees: 
    \begin{itemize}
        \item Given a query $\bq\sim\Dq$, the algorithm returns $\NN(\bq)$ with probability at least $1-\delta$ over the draw of $\bq$.\footnote{Even though the above search algorithm assumed knowledge of an upper bound for $\cost_{\calD}(T)$, one may obtain such estimates for a fixed tree $T$ via a few additional samples from $\Dpair$, given that $T$ is a halfspace or PTF tree.}
        \item The algorithm runs in time $O(\tau d_T) + O(nd/\delta^2\cdot \cost_{\calD}(T))$.
    \end{itemize}
\end{claim}

\begin{proof} 
    We store the tree $T$ and each leaf $u$ will store the subset of dataset nodes $P^{u}$ which appear in $\Dnn$ with probability at least $\delta / (2n)$. The search algorithm proceeds by (i) routing a query $\bq \sim \Dq$ down the tree (taking time at most $O(\tau \log(n/\delta)/\beta)$), and then (ii) scanning the first $4n \cdot \cost_{\Dpair}(T) / \delta^2$ dataset points in the bottom-up ordering of the tree (taking time at most $O(n d \cdot \cost_{\Dpair}(T) / \delta^2)$), and returning the closest point $\hat{\bp}$ among those scanned. To analyze the correctness of the algorithm, we upper bound the probability over $(\bq,\nn)\sim \Dpair$ that $\hat{\bp} \neq \nn$. Note, for a fixed draw $(\bq,\nn)$ this can happen for two reasons:
    \begin{itemize}
        \item Either $\nn$ appears in $\Dnn$ with probability at most $\delta / (2n)$, which occurs overall with probability at most $\delta/2$ since there are at most $n$ points,
        \item Or, $\nn$ does not appear within the first $4n \cdot \cost_{\Dpair}(T) / \delta^2$ points in the bottom-up order (all of which appear with probability at least $\delta / (2n)$ under $\Dnn$).
    \end{itemize}
    To upper bound the second event, consider the bottom-up order starting at the leaf to which $\bq$ is routed. The fact that $\nn$ does not appear among the first $4n \cdot \cost_{\Dpair}(T) / \delta^2$ points means these points are contained within $P^{\bu}$ for $\bu = \lca_{T}(\bq, \nn)$, and appear with probability at least $\delta / (2n)$ under $\Dnn$; so the probability that $\bp \sim \Dnn$ lies in $P^{\bu}$ is larger than $2\cost_{\Dpair}(T) / \delta$. Hence, 
    \begin{align*}
        \Prx_{(\bq,\nn)\sim\Dpair}\left[ \hat{\bp} \neq \nn \right] &\leq \delta/2 + \Prx_{(\bq,\nn)\sim \Dpair}\left[ \Prx_{\bp\sim\Dnn}\left[ \bp \in P^{\lca_T(\bq,\nn)}\right] > 2\cost_{\Dpair}(T)/\delta \right] \leq \delta,
    \end{align*}
    where the last line uses Markov's inequality and the definition of $\cost_{\Dpair}(T)$.
\end{proof}

Claim~\ref{claim:cost-to-search} shows that balanced, low-complexity cut trees $T$ with small cost under $\Dpair$ yield efficient nearest neighbor search algorithms. In the remainder of the section, we focus on learning a balanced PTF tree whose cost competes with that of the optimal halfspace tree. This can be viewed as a hierarchical generalization of the improper learning task in Section~\ref{sec:balanced-cut-algo}.

\begin{tcolorbox}[colback=cyan!05]
    \begin{problem}[Learning Balanced Halfspace Cut Trees in Realizable Setting]
    \label{problem:balanced-halfspace-tree}
        Given sample access to $\Dpair$ and an error parameter $\eps \in (0, 1)$, assume there exists a perfect $1/3$-balanced halfspace cut tree $T^\star$. The goal is to efficiently learn a low-complexity $\Omega(1)$-balanced tree $T$ satisfying 
        \[\cost_{\Dpair}(T) \le \eps.\]
    \end{problem}
\end{tcolorbox}

\subsection{Reduction to Balanced Halfspace Cut}

We present a black-box reduction from Problem~\ref{problem:balanced-halfspace-tree} to a generalization of the balanced halfspace cut problem. As we will see, the algorithm from Section~\ref{sec:balanced-cut-algo} will not change; we will increase the degree of the polynomial in order to compete against disjoint unions of intersections of halfspaces. First, we will need terminology for functions of a small number of halfspaces.

\begin{definition}
    A function $g : \R^d \to \{-1,1\}$ is a \emph{$J$-wise intersection of halfspaces} if there exist $J$ halfspaces $f_1, \dots, f_J : \R^d \to \{-1, 1\}$ such that for all $z \in \R^d$, $g(z) = 1$ if and only if $f_j(z) = 1$ for all $j \in [J]$.
\end{definition}

\begin{definition}
    A function $h : \R^d\to\{-1,1\}$ is an \emph{$(R,J)$-halfspace function} if there exist $R$ functions $g_1, \dots, g_R : \R^d \to \{-1,1\}$ which are $J$-wise intersection of halfspaces and disjoint, i.e., for all $z \in \R^d$, $g_r(z) = 1$ for at most one $r \in [R]$, and $h(z) = 1$ if and only if $g_r(z) = 1$ for some $r \in [R]$.
\end{definition}

To formalize the reduction, we introduce a \emph{balanced cut oracle}, which we will later instantiate using the \ptfcut algorithm from Theorem~\ref{thm:improper-alg}. Below, the query and nearest neighbor distribution is generalized (since we will make multiple different calls), and cuts will compete against optimal $(R, J)$-halfspace cuts.

\begin{definition}[Balanced Cut Oracle]
\label{def:cut-oracle}
For parameters $0<\beta_1 \leq \beta_0 \leq 1/2$, $R, J \in \N$, and error tolerance $\eps_{c} > 0$. An algorithm $\calO$ is a \emph{$(\beta_1,\beta_0, R, J, \eps_c)$-balanced cut oracle} if it satisfies the following guarantees:
    \begin{itemize}
        \item On input distribution $\Dpair'$ on $\R^d \times \R^d$, it outputs a cut $g \colon \R^d \to \{-1,1\}$ which satisfies $\bal_{\Dpair'}(g) \geq \beta_1$.
        \item For every $(R,J)$-halfspace function $h \colon \R^d \to \{-1,1\}$ with $\bal_{\Dpair'}(h) \geq \beta_0$,\footnote{We will always use $\beta_0 = 1/3$.} 
        \[ \err_{\Dpair'}(g) \leq O( \sqrt{\err_{\Dpair'}(h)}) + \eps_c.\]
    \end{itemize}
\end{definition}

The algorithm will call the balanced cut oracle $\calO$ until a node satisfies a termination condition, at which point we mark it a leaf. The recursion stops at node $u$ when its routing region $\calR_T(u)$ has negligible mass with respect to $\Dnn$. We utilize a lightweight testing procedure to efficiently estimate this probability from samples of $\Dpair$, via a simple Chernoff bound (deferred to Appendix~\ref{sec:missing-dasgupta}).

\begin{restatable}[Marginal Mass Tester]{lemma}{marginalTester}
\label{lem:marginal-tester}
    There exists an algorithm \TestMass that takes sample access to $\Dnn$, a node $u$ in a cut tree $T$, a threshold $\gamma \in (0, 1)$, and a failure probability $\delta \in (0, 1)$, and satisfies the following guarantees:
    \begin{itemize}
        \item If $\Pr_{\bp \sim \Dnn}[\bp \in P^u] \le \gamma$, it outputs \textsf{\textup{Small}} with probability at least $1-\delta$.
        \item If $\Pr_{\bp \sim \Dnn}[\bp \in P^u] \ge 2\gamma$, it outputs \textsf{\textup{Large}} with probability at least $1-\delta$.
    \end{itemize}
    The algorithm uses $O(\log(1/\delta) / \gamma)$ samples $\bp$ from $\Dnn$ and evaluates whether $\bp \in P^u$ by walking down $T$.
\end{restatable}

\begin{definition}[Conditional Mixture $\Dpair^{(u)}$]
    Let $u$ be a node in a cut tree $T$ where $\bp \sim \Dnn$ has $\bp \in P^{u}$ with non-zero probability. The distribution $\calD^{(u)}$ is supported on pairs $\calR_T(u) \times P^u$ and given by first drawing $(\bq, \nn) \sim \Dpair$ conditioned on $\nn \in P^u$; if $\bq \in \calR_T(u)$, output $(\bq,\nn)$ and otherwise, output $(\nn,\nn)$. We let $\Dnn^{(u)}$ be the second marginal of $\Dpair^{(u)}$, equivalently $\bp \sim \Dnn$ conditioned on $\bp \in \calR_T(u)$. 
\end{definition}

Given a balanced cut oracle and the marginal mass tester, we can construct a low-cost tree via the following greedy procedure:

\begin{tcolorbox}[colback=red!05]
\begin{center}
    \textbf{Algorithm:} $\GreedyTree(P, \Dpair, \delta \mid \calO)$
\end{center}
\textbf{Input:} A dataset $P\subset\R^d$, a balanced cut oracle $\calO$, sample access to a pair distribution $\Dpair$,  and a failure probability $\delta \in (0, 1)$. 
\begin{enumerate}
    \item Initialize the tree with a single root node $u$ where $\calR_T(u) = \R^d$.
    \item While there exists an active node $u$: \begin{enumerate}
        \item \textbf{Test Mass:} Call \TestMass on $\Dnn$ for $\calR_T(u)$ with parameters $\gamma$ and $\delta' = \Theta(\delta\beta_1\gamma)$. If \TestMass returns \textsf{Small}, mark $u$ as inactive (leaf).
        \item \textbf{Balanced Cut:} Otherwise, invoke $\calO$ on $\Dpair^{(u)}$ to obtain a balanced cut $g_u : \R^d \to \{-1,1\}$.
        \item \textbf{Recurse:} Create children $v$ and $v'$ corresponding to the regions $\calR_T(u) \cap g_u^{-1}(-1)$ and $\calR_T(u) \cap g_u^{-1}(1)$, respectively. Mark $v,v'$ as active and $u$ as inactive.
    \end{enumerate}
    \item Return the constructed tree $\bT$.
\end{enumerate}
    
\end{tcolorbox}

\begin{claim}\label{claim:greedy-tree-runtime}
    Suppose $\calO$ has balance parameters $\beta_1 \leq \beta_0 = 1/3$. With probability at least $1-\delta$, all invocations of \TestMass succeed and $\GreedyTree$ terminates after making at most $O(1/(\beta_1\gamma))$ calls to $\calO$ and $O(1/(\beta_1\gamma))$ calls to \TestMass. Excluding the time and samples required by $\calO$, the algorithm runs in time $O\left(n\tau\log(1/(\beta_1\gamma)) + \log(1/(\delta\beta_1\gamma))/\gamma^2\right)$ and uses $O\left(\log(1/(\delta\beta_1\gamma))/\gamma^2\right)$ samples.
\end{claim}
\begin{proof}
    First, consider an idealized process where every invocation of \TestMass is correct (i.e., satisfies the guarantees of Lemma~\ref{lem:marginal-tester}). Then, for every leaf node, it's parent's call to \TestMass output \textsf{Large}. Hence, the parent node $u$ had $\Prx_{\by \sim \Dnn}[\by \in P^u] \geq \gamma$ (since all invocations of \TestMass were correct), and from the cut oracle $\calO$, both children $v$ has $\Prx_{\by \sim \Dnn}[\by \in P^v] \geq \beta_1 \gamma$. Thus, in this idealized process, the collection of leaves form disjoint regions of $\R^d$ with mass under $\Dnn$ which is at least $\beta_1\gamma$, and hence, there are at most $1/(\beta_1\gamma)$ leaves. Since we build a binary tree, the total number of nodes is at most $2/(\beta_1 \gamma)$. Since each node uses invocation of \TestMass, the total number of invocations in the idealized process is at most $2/(\beta_1\gamma)$. Hence, a non-idealized process whose first $3/(\beta_1\gamma)$ invocations of \TestMass are correct produces the same tree as the idealized process. By a union bound and setting of $\delta'$ in \GreedyTree, the event that the first $3/(\beta_1\gamma)$ invocations satisfy the correctness guarantee holds with probability at least $1-\delta$, and hence the algorithm terminates with $O(1/(\beta_1\gamma))$ calls to $\TestMass$.
    
    To bound the runtime, observe that at any internal node $u$, the algorithm evaluates the cut $g_u$ on all $|P^u|$ points. Since $g_u$ has complexity $\tau$, this partitioning step takes $O(|P^u| \tau)$ time. Because the tree is $\beta_1$-balanced, its maximum depth is bounded by $\log_{1/(1-\beta_1)}(1/(\beta_1\gamma)) = O(\log(1/(\beta_1\gamma))/\beta_1)$. Thus, the total time spent partitioning is $O(n \tau \log(1/(\beta_1\gamma))/\beta_1)$. Finally, by Lemma~\ref{lem:marginal-tester}, the $O(1/(\beta_1\gamma))$ calls to \TestMass require a total of $O(\log(1/\delta')/(\beta_1\gamma^2)) = O(\log(1/(\delta\beta_1\gamma))/(\beta_1\gamma^2))$ time and samples, which yields the final bounds.
\end{proof}

\begin{lemma}
\label{lem:greedy-reduction}
    Let $\calO$ be a $(\beta_1, 1/3, R, J, \eps_c)$-balanced cut oracle for $R = O(1/(\beta_1\gamma))$ and $J = O(\log(1/(\beta_1\gamma))/\beta_1)$. Then, with probability at least $1-\delta$, $\GreedyTree(P, \Dpair, \delta\mid\calO)$ returns a tree $\bT$ satisfying
    \[ \cost_{\Dpair}(\bT) \leq 2\gamma + O(\eps_c/\beta_1), \]
    assuming the existence of a perfect $1/3$-balanced halfspace cut tree $T^\star$.
\end{lemma}

Before proving Lemma~\ref{lem:greedy-reduction}, we introduce useful notation for a cut tree $T$.
\begin{itemize}
    \item Let $\calI(T)$ and $\calL(T)$ denote the set of internal nodes and leaf nodes of $T$, respectively.
    \item For any $u\in\calI(T)$, let $g_u : \R^d \to \{-1,1\}$ denote the cut function that splits the points at $u$.
    \item For $\eta \in (0,1)$, let $T_{\le \eta}$ denote the tree formed by maximally contracting the leaves of $T$ so that every resulting leaf node $u$ satisfies $\Prx_{\bp \sim\Dnn}[\bp \in P^{u}] \leq \eta$.
\end{itemize}
\begin{proof}[Proof of Lemma~\ref{lem:greedy-reduction}]
We assume for the analysis the conclusions of Claim~\ref{claim:greedy-tree-runtime}, which implies that for every $u \in \calI(\bT)$, letting $\eta_u$ denote the probability $\bp \sim \Dnn$ satisfies $\bp \in P^{u}$, we have $\eta_u \geq \gamma$ (because $\TestMass$ outputs $\mathsf{Large}$ for every internal node). For $u \in \calI(\bT)$, consider the partition of $\R^d$ into parts $\calR_{T^{\star}}(u^{\star}_1), \dots, \calR_{T^{\star}}(u^{\star}_r)$ induced by the leaves $u_1^{\star}, \dots, u_r^{\star}$ of $T^{\star}_{\leq \eta_u /2}$. Importantly, we have that:
\begin{itemize}
    \item The number of parts $r$ is at most the number of leaves in $T^{\star}_{\leq \eta_u/2}$, which is at most $6/\eta_u$, since $T^{\star}$ is $1/3$-balanced, so each leaf of $T^{\star}_{\leq \eta_u/2}$ contains at least $\eta_u/6$ fraction of $\Dnn$.
    \item Each part $\calR_{T^{\star}}(u^{\star}_j)$ is defined by an intersection of at most $O(\log(1/\eta_u))$ halfspaces.
\end{itemize}
Letting $\xi_{1}, \dots, \xi_r \in [0,1]$ be given by
\[ \xi_j := \Prx_{\bp \sim \Dnn^{(u)}}\left[ \bp \in \calR_{T^{\star}}(u^{\star}_j) \right], \]
we have $\sum_{j=1}^r \xi_j = 1$, since these partition $\R^d$ and hence $\Dnn$, and in addition, that every $\xi_j \leq 1/2$ since \smash{$\Dnn^{(u)}$} is an $\eta_u$-fraction of $\Dnn$, and each \smash{$\calR_{T^{\star}}(u^{\star}_j)$} at most $\eta_u / 2$-fraction of $\Dnn$. We can always assemble the parts $\xi_1, \dots, \xi_r$ into $(S, [r] \setminus S)$ such that both $\sum_{j \in S} \xi_j$ and $\sum_{j \notin S} \xi_j$ are at least $1/3$ (via running a greedy largest-to-smallest assignment). Hence, this gives us a Boolean function $h_u \colon \R^d \to \{-1,1\}$ where $h_u(x) = 1$ iff $x \in \calR_{T^{\star}}(u^{\star}_j)$ for $j \in S$, which is a disjoint union of at most $r \leq R$ parts, each of which is an intersection of at most $J$ halfspaces; it satisfies $\bal_{\Dpair^{(u)}}(h_u) \geq 1/3$, and hence the output function $g_u \colon \R^d \to \{-1,1\}$ has $\bal_{\Dpair^{(u)}}(g_u) \geq \beta_1$ and 
\begin{align} 
\err_{\Dpair^{(u)}}(g_u) \leq \eps_c.  \label{eq:apprx}
\end{align}
We can now upper bound $\cost_{\Dpair}(\bT)$ by the contributions coming from the leaf nodes and internal nodes. The leaf node contributions correspond to $(\bq,\nn)\sim\Dpair$ having both $\bq,\nn \in \calR_{\bT}(u)$, and are at most
\begin{align*}
    \sum_{u \in \calL(\bT)} \Prx_{\bp \sim \Dnn}\left[ \bp \in P^u \right] \cdot \Prx_{(\bq,\nn)\sim\Dpair}\left[ \bq,\nn \in \calR_{\bT}(u)\right] \leq 2\gamma \sum_{u\in\calL(\bT)}\Prx_{(\bq,\nn)\sim\Dpair}\left[ \bq,\nn \in \calR_{\bT}(u) \right] \leq 2\gamma,
\end{align*}
since the leaves of $\bT$ form a partition of $\R^d$, and since $\TestMass$ declared $\mathsf{Small}$, must contain at most $2\gamma$-fraction of $\Dnn$. The internal node contributions require more care, and are at most
\begin{align}
\sum_{u \in \calI(\bT)} \eta_u \Prx_{(\bq, \nn) \sim \Dpair}\left[ u = \lca_{\bT}(\bq,\nn) \right] \leq \sum_{u \in \calI(\bT)} \eta_u \Prx_{(\bq,\nn)\sim\Dpair}\left[ \bq,\nn \in \calR_{\bT}(u)\right] \cdot \err_{\Dpair^{(u)}}(g_u),\label{eq:ha}
\end{align}    
since $u = \lca_{\bT}(\bq,\nn)$ whenever $\bq, \nn \in \calR_{\bT}(u)$ and are cut by $g_u$---furthermore, note the fact $(\bq,\nn) \sim \Dpair$ is conditioned on $\bq,\nn \in \calR_{\bT}(u)$ means the probability $g_u(\bq) \neq g_u(\nn)$ is equal to the probability $g_u(\bq) \neq g_u(\nn)$ when $(\bq,\nn) \sim \Dpair^{(u)}$. We may re-write (\ref{eq:ha}) as an expectation, to apply (\ref{eq:apprx}),
\begin{align}
\Ex_{(\bq,\nn)\sim \Dpair}[ \sum_{\substack{u \in \calI(\bT) \\ (\bq,\nn)\in \calR_{\bT}(u)}} \eta_u \cdot \err_{\calD^{(u)}}(g_u)] &\leq \Ex_{(\bq,\nn)\sim \Dpair}[ \sum_{\substack{u \in \calI(\bT) \\ (\bq,\nn)\in \calR_{\bT}(u)}} \eta_u \eps_c] \label{eq:hah2}
\end{align}
Notice that for a fixed $(\bq,\nn)$, the values $\eta_u$ over all $u \in \calI(\bT)$ where $(\bq,\nn) \in \calR_{\bT}(u)$ geometrically decay at rate at least $1-\beta_1$ and are at most $1$, meaning their sum is $O(1/\beta_1)$. Hence, the summation in (\ref{eq:hah2}) is upper bounded by $O(\eps_c / \beta_1)$, completing the proof.
\end{proof}
\begin{remark}
The same argument can be adapted to show that even if $T^\star$ is not perfect, $\cost_{\Dpair}(\bT) \le 2\gamma + O(\eps_c/\beta_1) + O(1/\beta_1)\cdot\sqrt{\cost_{\Dpair}(T^\star)}$. We omit this result because even a perfect halfspace tree $T^\star$ need not have subconstant cost (e.g., if some point in $P$ has constant mass under $\Dnn$, then $\cost_{\Dpair}(T^\star) \ge \Omega(1)$). Alternatively, if $\Dnn$ is assumed to be uniform over $P$, then any perfect tree $T^\star$ has $\cost_{\Dpair}(T^\star) = 1/n$.
\end{remark}

%% file: sections/final-learning.tex

\subsection{Final NNS Guarantee}

\begin{definition}
\label{def:gen-halfspace-approx-deg}
    Let $\Dpair'$ be a distribution over $\R^d \times P$, and let $R,J\in\N$. The \emph{$(R,J)$-halfspace approximation degree} is a function $k(\eps; R, J; \Dpair') \colon [0, 1] \to \N$ encoding the $\calF_{R,J}$-approximation degree of $\Dpair'$ (as per Definition~\ref{def:general-approx-deg}) where $\calF_{R,J}$ is the class of all $(R,J)$-halfspace functions $f \colon \R^d \to \{-1,1\}$.
\end{definition}

\begin{claim}\label{cl:conditioning}
    Let $\Dpair$ be a distribution over $\R^d \times P$ with $(R, J)$-halfspace approximation degree $k(\eps;R,J;\Dpair)$. If $u$ is a node in a cut tree $T$ where $\bp \sim \Dnn$ has $\bp\in P^{u}$ with probability at least $\gamma$, then $k(\eps; R,J;\Dpair^{(u)}) \leq k(\gamma\eps/2;R,J;\Dpair)$.
\end{claim}

\begin{proof}
    Consider any $(R, J)$-halfspace $f \colon \R^d \to \{-1,1\}$ and let $A \colon \R^d \to \R$ denote the degree $k(\gamma\eps/2;R,J;\Dpair)$ polynomial approximating $f$ with respect to $\Dpair$. Using the notation of Definition~\ref{def:general-approx-deg} to refer to $\Dpair^{(u)}_X$ and $\Dpair^{(u)}_Y$ as first and second marginals, for $\ell \in \{1,2,4\}$,
\begin{align*}
\Ex_{\by \sim \Dpair^{(u)}_{Y}}\left[ |A(\by) - f(\by)|^{\ell} \right] &\leq \frac{1}{\gamma} \Ex_{\by \sim \Dpair_Y}\left[ |A(\by) - f(\by)|^{\ell} \right] \leq \frac{1}{\gamma} \cdot \gamma\eps/2\leq \eps, \\
\Ex_{\bx \sim \Dpair^{(u)}_{X}}\left[ |A(\bx) - f(\bx)|^{\ell} \right] &\leq  \frac{1}{\gamma} \Ex_{(\bx,\bp) \sim \Dpair}\left[ \ind\{ \bx \in \calR_{\bT}(u)\} |A(\bx) - f(\bx)|^{\ell} + \ind\{\bx\notin \calR_{\bT}(u)\} |A(\bp) - f(\bp)|^{\ell} \right] \\
			&\leq \frac{1}{\gamma} \Ex_{\bx \sim \Dpair_X}\left[ |A(\bx) - f(\bx)|^{\ell} \right] + \frac{1}{\gamma} \Ex_{\bp \sim \Dpair_Y}\left[ |A(\bp) - f(\bp)|^{\ell} \right]  \leq \gamma\eps/(2\gamma) + \gamma \eps / (2\gamma) = \eps.
\end{align*}
\end{proof}

\begin{theorem}[Formalization of Theorem~\ref{thm:main-intro}]
\label{thm:main-formal}
    There exists a polynomial-time algorithm which receives a dataset $P \subset \R^d$ of $n$ points with $d = \polylog(n)$, and sample-access to a nearest neighbor pair distribution $\Dpair$.
    \begin{itemize}
	\item Suppose that there is a constant $c_1 > 0$ such that for any $\eps_a > 0$ and $R, J \in \N$, $\Dpair$ has $(R, J)$-halfspace approximation degree $k(\eps_a) \leq (R \cdot 2^J/\eps_a)^{c_1}$.
	\item $\Dpair$ admits a perfect $1/3$-balanced halfspace tree $T^{\star}$.
    \end{itemize}
    With probability $1-o(1)$, the algorithm outputs a nearest neighbor data structure for $P$ which uses space $O(nd)$ and finds a nearest neighbor in time $o(nd)$ with probability $1-o(1)$ over $(\bq,\nn) \sim \Dpair$.
\end{theorem}

\begin{proof}
We will use $\GreedyTree(P, \Dpair,\delta \mid \calO)$ with $\delta = o(1)$ (which encodes the failure probability of the data structure learning algorithm), and its guarantees (Lemma~\ref{lem:greedy-reduction}), and Theorem~\ref{thm:improper-alg} to implement the cut oracle $\calO$ (which is executed on conditional distributions $\Dpair^{(u)}$). 
\begin{enumerate}
\item[(a)] First, executing $\GreedyTree(P,\Dpair,\delta\mid\calO)$ with parameter $\gamma$ in $\TestMass$ (which we set later) guarantees that all calls to $\calO$ are on distributions $\Dpair^{(u)}$ whose second marginal always contains at least $\gamma$-fraction of $\Dnn$, and this implies that for any $R, J \in \N, \eps_a > 0$, the $(R, J)$-halfspace degree of every $\Dpair^{(u)}$ is at most $k(\gamma\eps_a/2; R, J; \Dpair)$ (Claim~\ref{cl:conditioning}). 
\item[(b)] For $\eps_c > 0$ (which we set later), we implement a $(\beta_1, 1/3, R, J, \eps_c)$-balanced cut oracle $\calO$ for $R = O(1/\gamma)$ and $J = O(\log(1/\gamma))$ as per Definition~\ref{def:cut-oracle} for each $\Dpair^{(u)}$ using Theorem~\ref{thm:improper-alg} on the class $\calF$ of $(R,J)$-halfspaces with input balance parameter $\beta \leftarrow 1/3$, error parameter $\eps\leftarrow c'\eps_c^2$ (for a small constant $c'$) and obtain $\beta_1 = \Omega(1)$. The performance guarantees, using (a) to say $k(c'\eps_c^2; R, J; \Dpair^{(u)}) \leq k(c'\gamma\eps_c^2/2; R,J;\Dpair)$, is 
\begin{align*} 
\text{sample complexity per call to $\calO$: }& O\left( (d + k(c'\gamma\eps_c^2/2;R,J;\Dpair))^{k(c'\gamma\eps_c^2/2;R,J;\Dpair)} / \eps_c^4 \right),
\end{align*}
whose running time is polynomial in the sample complexity (furthermore, $\calO$ is called $O(1/\gamma)$ times). By Claim~\ref{claim:greedy-tree-runtime} and inspection of $\GreedyTree(P, \Dpair,\delta\mid\calO)$, the total running time is polynomial in the above sample complexity, $\log(1/\delta)$ and $1/\gamma$, and linear in $n$.
\item[(c)] Using $\calO$ built above in $\GreedyTree(P,\Dpair,\delta\mid\calO)$ and applying Lemma~\ref{lem:greedy-reduction}, with probability $1-\delta$, we obtain a cut tree $\bT$ of depth $O(\log(1/\gamma))$ whose cuts are PTFs of degree $k(c'\gamma\eps_c^2/2; R, J;\Dpair)$ and $\cost_{\Dpair}(\bT)$ is at most $O(\gamma + \eps_c)$, and for any $\delta_s > 0$, Claim~\ref{claim:cost-to-search} implies a data structure which succeeds at finding the nearest neighbor of $(\bq,\nn) \sim \Dpair$ with probability at least $1-\delta_s$ and has query time at most
\[ O( (d + k(c'\gamma\eps_c^2/2;R,J;\Dpair))^{k(c'\gamma\eps_c^2/2;R,J;\Dpair)} \log(1/\gamma)) + O(nd / \delta_s^2) \cdot \left(\gamma + \eps_c \right).\]
\end{enumerate}
For the final parameter settings, assuming $d = \polylog(n)$, we set 
\[ \gamma = \eps_c = \delta_s^3 \qquad\text{and}\qquad \delta_s = (\log n)^{-C},\] 
for a small enough constant $C$. As long as $\cost_{\Dpair}(T^{\star}) \leq \delta_s^6$, the fact that $k(\cdot)$ grows polynomially in the error implies $k(c'\gamma\eps_c^2;R,J;\Dpair)$ is polynomial in at most $(\log n)^{O(C)}$. Since $d = \polylog(n)$, the query time of the data structure is \smash{$d^{{(\log n)}^{O(C)}} + nd / (\log n)^{\Omega(C)}$}, which is $o(nd)$ for a small enough constant $C > 0$. Finally, the data structure consists of a tree of $\polylog(n)$ nodes, each of which holds a polynomial $\R^d \to \R$ of degree $(\log n)^{O(C)}$, as well as the points, giving the total space bound.
\end{proof}

\begin{remark}
    In Appendix~\ref{sec:poly-approx} (specifically Theorem~\ref{thm:conditions-approx}), we show that distributions $\Dpair$ with Gaussian-like marginals have bounded $(R,J)$-halfspace approximation degree $k(\eps_a) \le (R \cdot 2^J/\eps_a)^{c_1}$.
\end{remark}

%% file: sections/poly-approx.tex

\newcommand{\Rint}{\calR^{\textup{int}}}
\newcommand{\Rtrans}{\calR^{\textup{trans}}}
\newcommand{\Rtail}{\calR^{\textup{tail}}}

\section{Additional Preliminaries}

\subsection{VC Dimension and Sample Complexity}

\begin{definition}
\label{def:vc-dim}
    For a set of elements $X$, let $\calC \subseteq 2^X$ denote a concept class. A set $S\subseteq X$ is \emph{shattered} by $\calC$ if for every $T\subseteq S$, there exists $C\in\calC$ for which $T = S \cap C$. The \emph{VC dimension} of $\calC$, denoted $\vc(\calC)$, is the size of the largest set shattered by $\calC$.    
\end{definition}

\begin{lemma}[Theorem~6.8 of~\cite{shalev2014}]
\label{lem:vc-samples}
    Let $\calC$ be a concept class. For any distribution $\calM$, if we draw $\calbS\sim\calM^m$ with
    \[
    m = O\!\left(\frac{\vc(\calC) + \log(1/\delta)}{\eps^2}\right),
    \]
    then with probability at least $1-\delta$ we have
    \[
    \sup_{C\in\calC}\left|\Prx_{\bx\sim S}[\bx\in C] - \Prx_{\bx\sim\calM}[\bx\in C]\right|\le\eps.
    \]
\end{lemma}

To apply this framework to geometric separators, we define the following standard concept classes over the domain $X = \R^d$.

\begin{definition}[Slabs]
    The class of \emph{slabs} (or \emph{strips}) in $d$ dimensions is defined as
    \[
    \mathrm{Slab}_1 = \left\{ \{x \in \R^d : |\langle w, x \rangle - \theta| \le \gamma\} : w \in \R^d, \theta \in \R, \gamma \ge 0 \right\}.
    \]
\end{definition}

\begin{fact}\label{fact:vc-halfspaces}
    The VC dimension of halfspaces in $\R^d$ is $\vc(\calP_1) = d + 1 = O(d)$.
\end{fact}

\begin{fact}\label{fact:vc-slabs}
    The VC dimension of slabs in $\R^d$ is $\vc(\mathrm{Slab}_1) = O(d)$.
\end{fact}

\begin{fact}\label{fact:vc-ptf}
    The VC dimension of the class of degree-$k$ PTFs in $d$ dimensions is 
    \[
    \vc(\calP_k) = \binom{d+k}{k} = O((d+k)^k).
    \]
\end{fact}

\section{Approximating Geometric Regions with Polynomials}
\label{sec:poly-approx}

Our final nearest neighbor search guarantee (Theorem~\ref{thm:main-formal}) makes an assumption on the $(R, J)$-halfspace approximation degree of the distribution $\Dpair$. Specifically, it assumes it is bounded by $(R \cdot 2^J/\eps)^{O(1)}$. The purpose of this appendix is to justify this assumption by demonstrating that this bound holds when the marginals of $\Dpair$ satisfy Gaussian-like properties.

In Subsection~\ref{subsec:poly-approx}, we show that the geometric regions of interest (halfspaces, intersections of halfspaces, and disjoint unions of intersections of halfspaces) are well-approximated by low-degree polynomials, whenever the underlying distribution satisfies a certain quantitative relaxation of anticoncentration and subgaussianity properties (see Definition~\ref{def:concentration-defs}). In Subsection~\ref{sec:extension-finite-sample}, we show that these anticoncentration and subgaussianity properties are preserved under finite sampling. The latter is relevant to our nearest neighbor search application, where the $\by$-marginal of $\Dpair$ is a discrete distribution supported on the dataset.

\begin{definition}
\label{def:concentration-defs}
    We say a distribution $\calM$ over $\R^d$ is \emph{$(L,\Delta)$-anticoncentrated} if for every unit vector $h\in\R^d$,
    \[
    \sup_{t\in\R}\cbr{\Prx_{\bz\sim\calM}[|\ip{h}{\bz} - t| \le r]} \le Lr+\Delta
    \quad \text{for all } r \in [0,1].
    \]
    We say that $\calM$ is \emph{$(B,r)$-subgaussian} if for every unit vector $h\in\R^d$ and integer $1\le p \le r$,
    \[
    \Ex_{\bz\sim\calM}[|\ip{h}{\bz}|^p]^{1/p} \le B\sqrt{p}.
    \]
    We say $\calM$ is $B$-subgaussian if it is $(B,r)$-subgaussian for all $r\in\N$.
\end{definition}

Recall the notion of $(R,J)$-halfspace approximation degree as defined in Definition~\ref{def:gen-halfspace-approx-deg}.

\begin{theorem}\label{thm:conditions-approx} Let $L,B>0$ be constants. Suppose that $\Dpair$ be a distribution over $\R^d \times \R^d$ whose first marginal $\Dx$ is $(L,0)$-anticoncentrated and $B$-subgaussian, and whose second marginal $\Dy$ is the uniform distribution over a dataset $\bP \sim \calM^m$, where $\calM$ is $(L,0)$-anticoncentrated and $B$-subgaussian.

There exists a constant $c_1 > 0$ such that for any target error $\eps \in (0, 1)$ and parameters $R, J \in \N$, if the dataset size satisfies $m \ge \poly(d, 1/\eps, \log(1/\delta), R, 2^J)$, then with probability at least $1-\delta$ over the draw of $\bP$, the distribution $\calD$ has an $(R,J)$-halfspace approximation degree bounded by:
$$ k(\eps; R, J; \Dpair) \le \left(\frac{R \cdot 2^J}{\eps}\right)^{c_1}. $$
\end{theorem}

A preliminary result, which we use to build polynomial approximations for high-dimensional geometric regions is the following result of~\cite{diakonikolas2010}, which gives a polynomial approximation of the one-dimensional sign function $\sign : \R \to \{-1,1\}$.

\begin{lemma}[Theorem~4.5 and Fact~3.11 of~\cite{diakonikolas2010}]
\label{lem:interval-approx}
    For any $\eta\in(0,1)$, there is a choice of $a = \Theta(\eta / \log(1/\eta))$ such that there exists a univariate polynomial $p$ of degree $k=\Theta(\log^2(1/\eta)/\eta)$ satisfying: \begin{itemize}
    \item \textbf{Interval:} $p(t) \in [\sign(t)-\eta,\, \sign(t)+\eta]$ for $|t|\in[a,1]$.
    \item \textbf{Transition:} $p(t) \in [-(1+\eta),\,1+\eta]$ for $|t| \le a$.
    \item \textbf{Tail:} $|p(t)| \le 2^{k+1}\cdot |t|^k$ for $|t| \ge 1$.
\end{itemize}
\end{lemma}

Using Lemma~\ref{lem:interval-approx}, we show that if a distribution $\calM$ over $\R^d$ satisfies anticoncentration and subgaussianity, then every halfspace can be approximated with a low-degree polynomial. Additionally, it will be helpful to explicitly identify the ``interval'', ``transition'', and ``tail'' regions in approximating any single halfspace and specify the error guarantees in each region.

\begin{lemma}
\label{lem:region-approx}
Fix a dimensionality $d \in \N$. Suppose that:
\begin{itemize}
\item $\calM$ is a distribution over $\R^d$ that is $(L, \Delta)$-anticoncentrated and $(B, r)$-subgaussian, for some parameters $L, \Delta, B \in \R_{\geq 0}$, and $r \in \N$. 
\item $f \colon \R^{d} \to \{-1,1\}$ is a halfspace, and $\eps \in (0,1)$ and $q \in \N$ are desired parameters.
\end{itemize}
Then, for $k = \tilde{O}( \max\{ q L^2 B^2 / \eps^2 , 1/\eps\} )$ and as long as $r \geq 4qk$, there exists a degree-$k$ polynomial $A \colon \R^d \to \R$ and a partition $\R^d = \Rint \sqcup\Rtrans\sqcup\Rtail$ such that, for all $\ell\in\{1,2,4\}$: \begin{itemize}
        \item \textbf{Interval:} For all $z\in\Rint$, it holds that $|A(z)-f(z)|^\ell = O(\eps)$.
        \item \textbf{Transition:} For all $z\in\Rtrans$, it holds that $|A(z)-f(z)|^\ell \le O(1)$, and $\Prx_{\bz\sim\calM}[\bz\in\Rtrans] = O(\eps+\Delta)$.
        \item \textbf{Tail:} $\Ex_{\bz\sim\calM}[|A(\bz)-f(z)|^{\ell \tilde{q}} \cdot \one\{\bz\in\Rtail\}] = O(\eps)$ for all $1 \leq \tilde{q} \leq q$.
    \end{itemize}
Moreover, the construction of $A$ depends only on the halfspace $f$ and the parameters $B, q, \eps$.
\end{lemma}

\begin{proof}
    Let $\calM$ be a distribution over $\R^d$ and $f(z) = \sign(\ip{h}{z}-\theta)$ a halfspace. We invoke Lemma~\ref{lem:interval-approx} with parameter $\eta = \min\{\eps^2 / (qL^2B^2), \eps\}$ to obtain a parameter $a = \Theta(\eta / \log(1/\eta))$ and polynomial $p$ of degree 
    \[k = O(\log(1/\eta)^2 / \eta) = \tilde{O}(\max\{qL^2B^2/\eps^2, 1/\eps\}).\] 
 We define a scale parameter $R = 32B\sqrt{qk}$. We consider two cases, based on whether $|\theta|$ is larger or smaller than $R/16$. 
 
 The easier case occurs when $|\theta|\ge R/16$; in this case, we ignore $p$ and show that the degree-0 polynomial $A(z) = \sign(-\theta)$ satisfies the necessary properties. We partition $\R^d$ by letting $\Rint = \{z \in \R^d : f(z) = \sign(-\theta)\}$, $\Rtrans = \emptyset$, and $\Rtail = \{z \in \R^d : f(z) \neq \sign(-\theta)\}$, and we verify that the three conditions hold. The interval property holds, because for all $z \in \Rint$, we have $|A(z) - f(z)| = 0$. The transition property holds vacuously since it is empty. For the tail property, note that $z \in \Rtail$ implies $|\ip{h}{z}| \ge |\theta| \ge R/16$. We use the upper bound on the $qk$-th moment from $(B, r)$-subgaussianity of $\calM$ and apply Markov's inequality, 
\[
    \Prx_{\bz\sim\calM}[\bz \in \Rtail] \le \Prx_{\bz\sim\calM}\left[|\ip{h}{\bz}| \ge \frac{R}{16}\right] \le \frac{\Ex_{\bz\sim\calM}[|\ip{h}{\bz}|^{qk}]}{(R/16)^{qk}} \le \left(\frac{16 B\sqrt{qk}}{R}\right)^{qk} \leq \frac{1}{2^{qk}},
\]
so using the fact $|A(z)-f(z)| \le 2$ as both $f$ and $A$ are functions in $\{-1,1\}$, we conclude
\[
    \Ex_{\bz\sim\calM}[|A(\bz)-f(z)|^{\ell \tilde{q}} \cdot \one\{\bz\in\Rtail\}] \leq 2^{\ell \tilde{q}} \cdot \Prx_{\bz\sim\calM}[\bz \in \Rtail] \le \frac{1}{2^{q(k - \ell)}} \le O(\eps)
\]
for $k$ being a sufficiently large constant factor of $\Omega(\log(1/\eps))$, which is always larger than $\ell$.

Next, we handle the case that $|\theta|\le R/16$. We define $t(z) = \frac{\ip{h}{z}-\theta}{R}$ and a degree-$k$ polynomial $A$ given by \[A(z) = p\del{t(z)}.\] We partition $\R^d$ as follows: \begin{itemize}
    \item Let $\Rint = \{z : |t(z)| \in [a,1]\}$. By construction and Lemma~\ref{lem:interval-approx}, $|A(z)-f(z)|^\ell\le\eta^\ell\le \eta = O(\eps)$.
    \item Let $\Rtrans = \{z : |t(z)| < a\}$. By construction and Lemma~\ref{lem:interval-approx} once more, $|A(z)-f(z)|^{\ell} \le (2+\eta)^{\ell} \le 3^4$ for all $z\in\Rtrans$. By $(L,\Delta)$-anticoncentration of $\calM$, $\Prx_{\bz\sim\calM}[\bz\in\Rtrans]\le aRL+\Delta$. Since $a = \Theta(\log(1/\eta) / k)$, we have \[aRL = O\del{\frac{\log(1/\eta)}{k} \cdot B\sqrt{qk} \cdot L} = O\del{\frac{BL\sqrt{q} \cdot \log(1/\eta)}{\sqrt{k}}} = O(\eps)\] by large enough choice of $k = \tilde{O}(qL^2B^2/\eps^2)$. Thus, $\Prx_{\bz\sim\calM}[\bz\in\Rtrans] = O(\eps+\Delta)$.
    \item Let $\Rtail = \{z : |t(z)| > 1\}$. By construction and Lemma~\ref{lem:interval-approx}, for all $z\in\Rtail$, $|A(z)|\le 2^{k+1}|t(z)|^k$ and thus $|A(z)|^{\ell \tilde{q}} \le 2^{\ell q(k+1)}|t(z)|^{\ell qk}$. To bound the tail expectation, we use the triangle inequality on the $\ell qk$-norm (losing a factor of $2^{\ell qk}$) and $(B, 4qk)$-subgaussianity of $\calM$: 
    \begin{align*}
    \Ex_{\bz\sim\calM}[|t(\bz)|^{\ell qk}] &\le  (2/R)^{\ell q k} \Ex_{\bz\sim\calM}[|\ip{h}{\bz}|^{\ell qk}] + |2\theta/R|^{\ell qk} \le \left(\dfrac{2 B \sqrt{\ell q k}}{R} \right)^{\ell qk} + \left(\frac{2\theta}{R}\right)^{\ell q k},    
\end{align*}   
Using the fact $f(\bz) \in \{-1,1\}$ and once more the triangle inequality on the $\ell q$-norm,
 \begin{align*}
 \Ex_{\bz\sim\calM}[|A(\bz) - f(\bz)|^{\ell \tilde{q}}\cdot\one\{\bz\in\Rtail\}] &\le 2^{\ell q} \left( \Ex_{\bz\sim \calM} [ | A(\bz)|^{\ell \tilde{q}} \cdot \ind\{ \bz \in \Rtail\} ] + \Prx_{\bz \sim \calM}[ \bz \in \Rtail] \right) \\
 	&\leq 2^{\ell q} \cdot  (2^{\ell q(k+1)} + 1) \cdot \Ex_{\bz \sim \calM}\left[ |t(\bz)|^{\ell q k} \right],
\end{align*}
where, we substituted the upper bound on $|A(\bz|^{\ell \tilde{q}}$, and used Markov's inequality to similarly upper bound the probability $\bz \in \Rtail$ by the $\ell qk$-th moment of $|t(\bz)|$. Substituting the $\ell qk$-th moment of $|t(\bz)|$, we obtain the upper bound
\[ \left( \dfrac{8 B \sqrt{\ell qk}}{R}\right)^{\ell qk} + \left( \frac{8 \theta}{R}\right)^{\ell q k} \leq O(\eps)\]
whenever $\ell \leq 4$, $R \geq 32 B \sqrt{qk}$, $R \geq 16 \theta$ and $k$ is a sufficiently large constant factor of $\log(1/\eps)$.
\end{itemize}

This completes the proof.
\end{proof}

\subsection{Polynomial Approximation Statements}\label{subsec:poly-approx}

In the upcoming lemmas, we use Lemma~\ref{lem:region-approx} in order to construct low-degree polynomial approximations of halfspaces, intersection of halfspaces and disjoint unions of intersections of halfspaces. We will consider distributions $\calM$ over $\R^d$ which are $(L, \Delta)$-anticoncentrated and $(B, r)$-subgaussian, for parameters $L, B, \Delta \ge 0$ and $r \in \N$. Conceptually, it is useful to distinguish the roles of these parameters as follows:
\begin{itemize}
\item $L$ and $B$ capture the anticoncentration and subgaussianity of the underlying distribution, respectively. They should be viewed as fixed values.
\item $\Delta \ge 0$ and $r \in \N$ capture the extent to which a \emph{discrete} distribution approximates anticoncentrated and subgaussianity, respectively. $\Delta$ should be viewed as arbitrarily small and $r$ arbitrarily large.
\end{itemize}

\begin{lemma}
Let $f \colon \R^d \to \{-1,1\}$ be a halfspace and $\eps \in (0, 1)$. For $k = \tilde{O}(\max\{ L^2 B^2/\eps^2, 1/\eps\})$, and as long as $\Delta \leq \eps$ and $r \geq 4k$, there exists a degree-$k$ polynomial $A \colon \R^d \to \R$ such that for $\ell \in \{1, 2, 4\}$: 
\[\Ex_{\bz\sim\calM}[|A(\bz)-f(\bz)|^\ell] = O(\eps).\] 
Moreover, $A$ depends only on the halfspace $f$ and the parameters $B$ and $\eps$.
\end{lemma}

\begin{proof}
    We invoke Lemma~\ref{lem:region-approx} with parameter $q=1$ to obtain a polynomial $A$ of degree $\Ot(\max\{L^2B^2/\eps^2,1/\eps\})$ and a partition of $\R^d$ into disjoint regions $\Rint, \Rtrans, \Rtail$ satisfying the stated properties. We bound the approximation error by decomposing it over these three regions:
    \[
        \Ex_{\bz\sim\calM}[|A(\bz)-f(\bz)|^\ell] = \sum_{\calR \in \{\Rint, \Rtrans, \Rtail\}} \Ex_{\bz\sim\calM}\left[|A(\bz)-f(\bz)|^\ell \cdot \one\{\bz \in \calR\}\right]. \tag{$\ast$}
    \]
We bound each contribution using Lemma~\ref{lem:region-approx} as follows. In the interval region, all $z\in\Rint$ satisfy $|A(z)-f(z)|^\ell = O(\eps)$ for all $z \in \Rint$, so the contribution to $(\ast)$ is at most $O(\eps)$. In the transition region, all $z\in\Rtrans$ satisfy $|A(z) - f(z)|^\ell \leq O(1)$ and $\Prx_{\bz\sim\calM}[\bz \in \Rtrans] = O(\eps)$ (since $\Delta \leq \eps$). Thus, the contribution to ($\ast$) is $O(\eps)$. Finally, the tail region satisfies $\Ex_{\bz\sim\calM}[|A(\bz)-f(\bz)|^\ell \cdot \one\{\bz\in\Rtail\}] = O(\eps)$, meaning the contribution to ($\ast$) from this region is $O(\eps)$. Summing the bounds over the three regions yields $\Ex_{\bz\sim\calM}[|A(\bz)-f(\bz)|^\ell] = O(\eps)$ for all $\ell \in \{1,2,4\}$. Finally, Lemma~\ref{lem:region-approx} ensures that the construction of $A$ depends only on $f$, $B$, and $\eps$, concluding the proof.
\end{proof}

\begin{lemma}
\label{lem:intersect-approx}
Let $f \colon \R^d \to \{-1,1\}$ be an intersection of $J$ halfspaces $f_1, \dots, f_J \colon \R^d \to \{-1,1\}$. That is, 
    \[ f(z) :=  \del{2\prod_{j=1}^J \frac{f_j(z)+1}{2}} - 1 = \begin{cases}
        1 & \text{if } f_j(z) = 1  \text{ for all $j\in[J]$}\\
        -1 & \text{otherwise} 
    \end{cases}, \] 
    and $\eps \in (0, 1)$. Then, for a sufficiently large constant $c > 1$ and $k = \tilde{O}(\max\{J^2 L^2B^2 c^{2J}/\eps^2, 1/\eps \})$, as long as $\Delta \leq \eps / c^{J}$ and $r \geq 4J k$, there exists a degree-$k$ polynomial $A \colon \R^d \to \R$ such that for all $\ell \in \{1, 2, 4\}$:
    \[ \Ex_{\bz\sim\calM}[|A(\bz)-f(\bz)|^\ell] = O(\eps). \]
    Moreover, $A$ depends only on the function $f$ and the parameters $B$, $J$, and $\eps$.
\end{lemma}

\begin{proof}
    Let $c>1$ be a large enough constant, and let $\eps' = \eps / c^J$. For each halfspace $f_j$, we invoke Lemma~\ref{lem:region-approx} with parameter $q=J$ and error $\eps'$, obtaining a polynomial $A_j$ of degree $k_j = \Ot(\max\{JL^2B^2/(\eps')^2,1/\eps\}) = \Ot(\max\{JL^2B^2c^{2J}/\eps^2, 1/\eps\})$ and a partition of $\R^d$ into regions $\Rint_j,\Rtrans_j,\Rtail_j$. We define a polynomial $A$ of degree $\sum_{j\in[J]} k_j = \Ot(J^2L^2B^2c^{2J}/\eps^2)$ as
\[A(z) := \del{2\prod_{j=1}^J \frac{A_j(z)+1}{2}} - 1,\]
such that we can upper bound $|A(z) - f(z)|$ by
\begin{align*}
2^{1-J} \left| \prod_{j=1}^J (A_j(z) + 1) - \prod_{j=1}^J (f_j(z) + 1) \right| &= 2^{1-J} \left| \prod_{j=1}^J \left(f_j(z) + 1 + (A_j(z) - f_j(z)) \right) - \prod_{j=1}^J (f_j(z) + 1)\right| \\
		&= 2^{1-J} \sum_{T \subset [J] \setminus \{ \emptyset\}} \prod_{j \in T} \left| A_j(z) - f_j(z) \right| \prod_{j \notin T} |f_j(z)+1| \\
		&\leq 2 \sum_{T \subset [J]\setminus \{ \emptyset\}} \prod_{j \in T} |A_j(z) - f_j(z)|,
\end{align*}
where the last line used the fact $|f_j(z) + 1|\leq 2$, and therefore, raising to power $\ell$ (and identifying the maximum $T$ term), we obtain the upper bound 
\begin{align*}
|A(z) - f(z)|^{\ell} \leq 2^{\ell(J + 1)} \sum_{T \subset [J] \setminus \{ \emptyset\}} \prod_{j \in T} |f_j(z) - A_j(z)|^{\ell}.
\end{align*}
Letting $\Rtail = \bigcup_{j=1}^J \Rtail_j$ and $\Rtrans = (\cup_{j=1}^J \Rtrans_j) \setminus \Rtail$ and $\Rint = \bigcap_{j=1}^J \Rint_j$, we consider the three regions. The easier case occurs whenever $z \in \Rint$, since $z \in \Rint_j$ for all $j \in [J]$ and Lemma~\ref{lem:region-approx} implies $|f_j(z) - A_j(z)|^{\ell} \leq \eps'$. Hence,
\begin{align*} 
\Ex_{\bz \sim \calM}\left[ |A(\bz) - f(\bz)|^{\ell} \cdot \ind\{ \bz \in \Rint\} \right] &\leq 2^{\ell(J+1)} \sum_{T \subset [J] \setminus \{\emptyset\}} (\eps')^{|T|} \leq \eps,
\end{align*}
since $\eps'$ is at most $\eps / c^J$ for a sufficiently large constant $c$, and $\ell \leq 4$. Next, if $z \in \Rtrans$, we must have $z \in \Rtrans_j$ for some $j$ while $z \notin \Rtail_{j'}$ for all $j' \in [J]$. Hence, we have
\begin{align*}
\Ex_{\bz \sim \calM}\left[ |A(\bz) - f(\bz)|^{\ell} \cdot \ind\{ \bz \in \Rtrans\} \right] \leq 2^{\ell(J+1)} \cdot 2^J \cdot O(1)^{J} \sum_{j \in [J]} \Prx_{\bz \sim \calM} \left[ \bz \in \Rtrans_j\right] \leq \eps,
\end{align*}
where we first use $|A_{j}(z) - f_j(z)| \leq O(1)$ for $z \in \Rtrans_j$ and the probability $\bz \in \Rtrans_j$ is at most $O(\eps' + \Delta)$, and each is at most $\eps / c^J$. Finally, for $z \in \Rtail$, we use two facts: at least one $j \in [J]$ has $z \in \Rtail_j$, and $|A_{j'}(z) - f_{j'}(z)| \leq O(1)$ whenever $z \notin \Rtail_{j'}$. This gives us the upper bound
\begin{align*}
|A(z) - f(z)|^{\ell} \cdot \ind\{ z \in \Rtail \} &\leq 2^{\ell(J+1)} \sum_{T \subset [J] \setminus \{ \emptyset \}} O(1)^{\ell |T|} \prod_{\substack{j \in T \\ z \in \Rtail_j}} |A_j(z) - f_j(z)|^{\ell}.
\end{align*}
To simplify the above, consider any $z \in \Rtail$ and any $T \subset [J] \setminus \{\emptyset\}$ where $K \subset T$ consists of indices $j$ where $z \in \Rtail_j$. Applying the AM-GM inequality, we can upper bound the product $\prod_{j \in K} |A_j(z) - f_j(z)|^{\ell}$ by $\sum_{j \in K} |A_j(z) - f_j(z)|^{\ell |K|}$, giving
\begin{align*}
|A(z) - f(z)|^{\ell} \cdot \ind\{ z \in \Rtail \} \leq O(1)^{\ell(J+1)} \sum_{\kappa=1}^J \sum_{j=1}^J |A_j(z) - f_j(z)|^{\ell \kappa} \cdot \ind\{ z \in \Rtail_j \} ,
\end{align*}
which implies
\begin{align*}
\Ex_{\bz \sim \calM}\left[ |A(\bz) - f(\bz)|^{\ell} \cdot \ind\{ \bz \in \Rtail \}\right] \leq O(1)^{\ell (J+1)} \sum_{\kappa=1}^J \sum_{j=1}^J \Ex_{\bz \sim \calM}\left[ |A_j(\bz) - f_j(\bz)|^{\ell \kappa} \cdot \ind\{ \bz \in \Rtail_j\}\right].
\end{align*}
This is at most $O(1)^{\ell(J+1)} \cdot J^2 \cdot O(\eps') \leq \eps$ as desired, where we have used the last condition of Lemma~\ref{lem:region-approx}. This completes the proof.
\end{proof}

\begin{lemma}
\label{lem:union-intersect-approx}
Let $f \colon \R^d \to \{-1,1\}$ be a union of $R$ disjoint intersections of $J$ halfspaces. That is, $f^{(1)}, \dots, f^{(R)} \colon \R^d \to \{-1,1\}$ are each intersections of $J$ halfspaces forming disjoint regions, and 
\[ f(z) := \sum_{r=1}^R (f^{(r)}(z) + 1) - 1 = \begin{cases}
        1 & \text{if } f^{(r)}(z) = 1  \text{ for some $r\in[R]$}\\
        -1 & \text{otherwise} 
    \end{cases}, \]
    and $\eps \in (0, 1)$. Then, for a sufficiently large constant $c > 1$ and $k = \tilde{O}(J^2L^2B^2 R^8 c^{2J} / \eps^2)$, as long as $\Delta \leq \eps / (R^4 c^J)$ and $r \geq 4J k$, there exists a polynomial $A$ over $\R^d$ of degree $k$ such that for $\ell \in \{1, 2, 4\}$:
    \[ \Ex_{\bz\sim\calM}[|A(\bz)-f(\bz)|^\ell] = O(\eps). \]
    Moreover, $A$ depends only on the function $f$ and the parameters $B,J,R,\eps$.
\end{lemma}

\begin{proof}
    For each $r \in [R]$, we invoke Lemma~\ref{lem:intersect-approx} with accuracy $\eps' = \eps / R^4$ to obtain a polynomial $A^{(r)}$ of degree $k_r = \Ot(J^2 L^2 B^2 c^{2J} / (\eps')^2) = \Ot(J^2 L^2 B^2 R^8 c^{2J} / \eps^2)$, and we let $A \colon \R^d \to \R$ be the polynomial given by $A(z) := \sum_{r=1}^R (A^{(r)}(z) + 1) - 1$. By the power mean inequality,
 \[ |A(z) - f(z)|^{\ell} \leq R^{\ell-1} \sum_{r=1}^R |A^{(r)}(z) - f^{(r)}(z)|^{\ell}.  \]
 Hence, 
 \begin{align*}
 \Ex_{\bz \sim \calM}\left[ |A(\bz) - f(\bz)|^{\ell} \right] \leq R^{\ell-1} \sum_{r=1}^R \Ex_{\bz \sim \calM}\left[ |A^{(r)}(\bz) - f^{(r)}(\bz)|^{\ell} \right] \leq R^{\ell} \eps ' \leq \eps,
 \end{align*}
 by the guarantees of Lemma~\ref{lem:intersect-approx} and setting of $\eps'$.
\end{proof}

\subsection{Extension to a Finite Sample} \label{sec:extension-finite-sample}

We will formally show that given a sufficiently large random sample from an $(L,0)$-anticoncentrated and $B$-subgaussian distribution, the uniform distribution over these samples will be $(L,\Delta)$-anticoncentrated and $(B,r)$-subgaussian, with high probability. Crucially, with enough samples, $\Delta$ can be made arbitrarily small, and $r$ arbitrarily large.

We begin with a claim bounding the maximum norm of sampled vectors from a subgaussian distribution.

\begin{claim}
\label{claim:subgaussian-max-norm}
    Let $\calM$ be a $B$-subgaussian distribution over $\R^d$. Then, with probability at least $1 - \delta$ over $\calbS\sim\calM^m$, 
    \[
    \max_{\bz \in \calbS} \|\bz\|_2 \le C B \sqrt{d + \log(m/\delta)}
    \]
    for a sufficiently large constant $C > 0$.
\end{claim}
\begin{proof}
    Fix any unit vector $h\in\mathbb{S}^{d-1}\subset\R^d$. By $B$-subgaussianity of $\calM$, for all $p\in\N$, we have \[\Ex_{\bz\sim\calM}[|\ip{h}{z}|^p] \le (B\sqrt{p})^p.\] Applying Markov's inequality, this gives \[\Prx_{\bz\sim\calM}[|\ip{h}{\bz}|\ge t] \le \frac{(B\sqrt{p})^p}{t^p}.\] Setting $p = \floor{t^2 / (e^2B^2)}$ yields \[\Prx_{\bz\sim\calM}[|\ip{h}{\bz}| \ge t] \le \exp(-\Omega(t^2/B^2)).\] Now, let $\Gamma\subset\mathbb{S}^{d-1}$ denote a $1/2$-net of the unit sphere, with $|\Gamma|\le \exp(O(d))$. Using that $\ip{u}{v} = 1-\frac{\|u-v\|_2^2}{2}$ for unit vectors $u,v\in\R^d$, the following holds for any $z\in\R^d$, \[\|z\|_2 \le \frac{8}{7}\cdot\max_{h\in\Gamma}|\ip{h}{\bz}|.\] Taking a union bound over $\Gamma$, \[\Prx_{\bz\sim\calM}[\|\bz\|_2 \ge t] \le |\Gamma| \cdot \exp(-\Omega(t^2/B^2)) \le \exp(O(d) - \Omega(t^2/B^2)).\] Taking $t = CB\sqrt{d+\log(m/\delta)}$ for a sufficiently large constant $C$ gives \[\Prx_{\bz\sim\calM}[\|\bz\|_2 \ge t] \le \delta/m.\] Finally, for $\calbS\sim\calM^m$, taking a union bound over the $m$ samples yields \[\Prx_{\calbS\sim\calM^m}[\max_{\bz\in\bS}\|\bz\|_2 \ge t] \le \delta.\] 
\end{proof}

The following is the main lemma, where we show that taking finite samples preserves anticoncentration and subgaussianity.

\begin{lemma}
\label{lem:finite-sample-extension}
    Let $\calM$ be a distribution over $\R^d$ that is $(L,0)$-anticoncentrated and $B$-subgaussian. For $r\in\N$, $\delta>0$, and $\Delta > 0$, let $\calbS \sim\calM^m$ where 
    \[
    m \ge \max\left\{ C \cdot  \frac{d + \log(1/\delta)}{\Delta^2},\, d \cdot \left(\frac{2B^2d\log(1/\delta)}{r}\right)^r \right\}
    \]
    for a sufficiently large constant $C > 0$. Then, with probability at least $1-\delta$ over the draw of $\calbS$, the empirical distribution $\Unif(\calbS)$ is $(L,\Delta)$-anticoncentrated and $(2B,r)$-subgaussian.
\end{lemma}

\begin{proof}
    We bound the failures of anticoncentration and subgaussianity separately, showing that both properties hold simultaneously with probability at least $1-\delta$.

    \paragraph{Anticoncentration.}
    Let $\mathrm{Slab}_1$ denote the class of slabs in $\R^d$, given by sets $\{z : |\ip{h}{z} - \theta| \le \gamma\}$ for vectors $h \in \R^d$, $\theta \in \R$, and $\gamma\ge 0$. By Fact~\ref{fact:vc-slabs}, $\vc(\mathrm{Slab}_1)=O(d)$. Applying Lemma~\ref{lem:vc-samples}, with probability at least $1 - \delta/2$ over the draw of $\calbS$, 
    \[
    \sup_{\substack{h\in\R^d \\ t,r\in\R}}\cbr{\left|\Prx_{\bz\sim\calbS}[|\ip{h}{\bz} - t| \le r] - \Prx_{\bz\sim\calM}[|\ip{h}{\bz} - t|\le r]\right|} \le O\left(\sqrt{\frac{d + \log(1/\delta)}{m}}\right).
    \]
    Because $m \ge C\cdot \frac{d + \log(1/\delta)}{\Delta^2}$ for a sufficiently large $C$, the above expression is at most $\Delta$. Since $\calM$ is $(L,0)$-anticoncentrated, for any $h\in\R^d$ and $t,r\in\R$, we have 
    \[
    \Prx_{\bz\sim\calM}[|\ip{h}{\bz} -t| \le r] \le Lr,
    \] 
    which implies 
    \[
    \Prx_{\bz\sim\calbS}[|\ip{h}{\bz} -t| \le r] \le Lr+\Delta.
    \] 
    Thus, $\Unif(\calbS)$ is $(L, \Delta)$-anticoncentrated.

    \paragraph{Subgaussianity.}
    We must show that for all unit vectors $h \in \R^d$ and all integers $p \in [1, r]$, the empirical moments satisfy $\Ex_{\bz\sim\calbS}[|\ip{h}{\bz}|^p] \le (2B\sqrt{p})^p$. 
    
    Let $\mathrm{Tail}_1$ denote the class of tail regions given by sets $\{z : |\ip{h}{\bz}| \ge \tau\}$ for vectors $h \in \R^d$ and $\tau \ge 0$. Because tails are complements of slabs, $\vc(\mathrm{Tail}_1) = O(d)$ by Fact~\ref{fact:vc-slabs}. Applying Lemma~\ref{lem:vc-samples}, with probability at least $1 - \delta/4$ over the draw of $\calbS$, we have
    \[
    \sup_{\substack{h\in\R^d \\ \tau \ge 0}} \cbr{\left|\Prx_{\bz\sim\calbS}[|\ip{h}{\bz}| \ge \tau] - \Prx_{\bz\sim\calM}[|\ip{h}{\bz}| \ge \tau]\right|} \le O\left(\sqrt{\frac{d + \log(1/\delta)}{m}}\right).
    \]
    
    Let $R = C B \sqrt{d + \log(m/\delta)}$ be the bound from Claim~\ref{claim:subgaussian-max-norm} (which also holds with probability $1-\delta/2$). Conditioned on these events, $|\ip{h}{\bz}| < \tau$ for all $\bz\in\calbS$. Substituting $\tau = t^{1/p}$ we have:
    \begin{align*}
        \Ex_{\bz\sim\calbS}[|\ip{h}{\bz}|^p] &= \int_0^\infty \Prx_{\bz\sim\calbS}[|\ip{h}{\bz}|^p \ge t] \,dt \\
        &= \int_0^{R^p} \Prx_{\bz\sim\calbS}[|\ip{h}{\bz}| \ge t^{1/p}] \,dt \\
        &\le \int_0^{R^p} \left( \Prx_{\bz\sim\calM}[|\ip{h}{\bz}| \ge t^{1/p}] + O\left(\sqrt{\frac{d + \log(1/\delta)}{m}}\right) \right) \,dt \\
        &\le \int_0^\infty \Prx_{\bz\sim\calM}[|\ip{h}{\bz}|^p \ge t] \,dt + R^p \cdot O\left(\sqrt{\frac{d + \log(1/\delta)}{m}}\right) \\
        &\le \Ex_{\bz\sim\calM}[|\ip{h}{\bz}|^p] + \left(CB\sqrt{d + \log(m/\delta)}\right)^p \cdot O\left(\sqrt{\frac{d + \log(1/\delta)}{m}}\right).
    \end{align*}
    By the $B$-subgaussianity of $\calM$, the true moment $\Ex_{\bz\sim\calM}[|\ip{h}{\bz}|^p]$ is at most $(B\sqrt{p})^p$. For the remaining error term, we use $m \ge d \cdot \left(\frac{2B^2d\log(1/\delta)}{r}\right)^r$ to conclude the bound of $(B\sqrt{p})^p$. 
    
    Adding these together yields:
    \[
    \Ex_{\bz\sim\calbS}[|\ip{h}{\bz}|^p] \le 2(B\sqrt{p})^p \le (2B\sqrt{p})^p,
    \]
    which confirms that $\Unif(\calbS)$ is $(2B, r)$-subgaussian, completing the proof.
\end{proof}

We are now ready to prove the main theorem.

\begin{proof}[Proof of Theorem~\ref{thm:conditions-approx}]
To show that $\Dpair$ achieves the stated $(R,J)$-halfspace approximation degree, it suffices to show that both marginals $\Dx$ and $\Dy$ are $(O(1),\Delta)$-anticoncentrated and $(O(1),r)$-subgaussian, $\Delta \le \eps / (R^4 c^J)$ is sufficiently small, and $r\ge \Omega(J)$ is sufficiently large. By assumption, for constants $L,B>0$, the first marginal $\Dx$ is $(L,0)$-anticoncentrated and $B$-subgaussian. Thus, it immediately satisfies the requirements of Lemma~\ref{lem:union-intersect-approx} (i.e., with $\Delta = 0$ and all $r$ arbitrarily large). For the second marginal $\Dy = \Unif(\bP)$, we apply Lemma~\ref{lem:finite-sample-extension}, which shows that a sample size of $m \ge \poly(d, 1/\eps, \log(1/\delta), R, 2^J)$ is sufficient to guarantee that $\Dy$ is $(L,\Delta)$-anticoncentrated and $(2B,r)$-subgaussian for sufficiently small $\Delta$ and sufficiently large $r$, with probability at least $1-\delta$. 

Conditioned on this event, we apply Lemma~\ref{lem:union-intersect-approx} to both $\Dx$ and $\Dy$. This guarantees that any union of $R$ disjoint intersections of $J$ halfspaces can be approximated (up to error $O(\eps)$) by a polynomial of degree
$$ k = \tilde{O}\left(\frac{J^2 R^8 c^{2J}}{\eps^2}\right)$$ up to the required moments. Choosing an appropriately large constant $c_1$, we conclude:
$$ k(\eps; R, J; \Dpair) \le \left(\frac{R \cdot 2^J}{\eps}\right)^{c_1}. $$

Crucially, Lemma~\ref{lem:union-intersect-approx} ensures that the constructed polynomial depends \emph{only} on the target $(R,J)$-halfspace function and parameters ($B, J, R,$ and $\eps$), so we obtain a single polynomial for both $\Dx$ and $\Dy$. 
\end{proof}

%% file: sections/missing-cut-algo.tex

\section{Omitted Proofs from Section~\ref{sec:balanced-cut-algo}}
\label{sec:missing-proofs}

\subsection{Rounding to a Sparse Cut (Proof of Lemma~\ref{lem:shift-existence})}
\label{subsec:cheeger-rounding}
In this section, we prove that any function with low energy and constant kurtosis induces a threshold cut with low error and constant balance.

\shiftExistence*

Observe that $\Rpair(F)$ and $\kurt(F)$ are invariant under affine transformations of $F$. Moreover, $\err(F,\theta)$ and $\bal(F,\theta)$ are preserved under affine transformations of $F$, provided that the same transformation is applied to the threshold $\theta$. Hence, it suffices to prove Lemma~\ref{lem:shift-existence} for function $F : \R^d\to\R$ with mean $0$ and variance $1$ on $\Dy$, in which case \[\Rpair(F) = \Ex_{(\bx,\by)\sim\calD}[(F(\bx)-F(\by))^2]\quad\text{and}\quad\kurt(F;\;\Dy) = \Ex_{\by\sim\Dy}[F(\by)^4].\] We will use this simplification for the remainder of the section. Also, we will henceforth omit the subscript $\Dpair$ from $\calR$, $\err$, and $\bal$.

We begin by defining the following relaxation of $\calR(F)$: \[\calS(F) := \Ex_{(\bx,\by)\sim\Dpair}[|F(\bx)^2-F(\by)^2|].\]

\begin{claim}
\label{claim:relating-S-and-R}
    For any $F: \R^d\to\R$ with variance $1$ on $\Dy$,\, $\calS(F) \le 8\cdot \del{\sqrt{\calR(F)} + \calR(F)}$.
\end{claim}
\begin{proof}
    By Cauchy-Schwarz, we have \[\calS(F) \le \sqrt{\Ex_{(\bx,\by)\sim\Dpair}[(F(\bx)-F(\by))^2]} \cdot \sqrt{\Ex_{(\bx,\by)\sim\Dpair}[(F(\bx)+F(\by))^2]}.\] The left factor is precisely $\sqrt{\calR(F)}$. To bound the right factor, we use the inequality $(a+b)^2 \le 2a^2+2b^2$, obtaining \[\Ex_{(\bx,\by)\sim\Dpair}[(F(\bx)+F(\by))^2] \le 2\cdot \Ex_{(\bx,\by)\sim\Dpair}[(F(\bx)-F(\by))^2] + 2\cdot \Ex_{\by\sim\Dy}[4\cdot F(\by)^2] = 2\cdot\calR(F) + 8.\] Therefore, $\calS(F) \le 8\cdot \del{\sqrt{\calR(F)} + \calR(F)}$, as desired.
\end{proof}

Now, we show the existence of a sparse threshold cut for nonnegative functions $G : \R^d\to\R_{\ge 0}$. For convenience, we define $\tail(G,t) = \Pr_{\by\sim\Dy}[G(\by) \ge t]$ so that $\bal(G,t) = \min\{\tail(G,t), 1-\tail(G,t)\}$.

\begin{claim}
\label{claim:nonnegative-shift}
    For any $G : \R^d \to \R_{\ge 0}$ with $\E_{\by\sim\Dy}[G(\by)^2] = 1$ and $\E_{\by\sim\Dy}[G(\by)^4] = M_4$, there is a threshold $t' > 0$ such that \[\frac{\err(G,t')}{\tail(G,t')} \le 2\cdot \calS(G)\quad\text{and}\quad\tail(G,t') \ge  \frac{1}{4M_4}.\]
\end{claim}

\begin{proof}
    For brevity, we will omit $\Dpair$ and $\Dy$ from expectations and probabilities, as they are clear from context. We begin by writing \[\frac{\calS(G)}{1} = \frac{\Ex_{(\bx,\by)}[|G(\bx)^2 - G(\by)^2|]}{\Ex_{\by}[G(\by)^2]} = \frac{\Ex_{(\bx,\by)}\left[\int_0^\infty \one\{\sign(G(\bx)^2 - t) \neq \sign(G(\by)^2 - t)\}\,dt\right]}{\Ex_{\by}\left[\int_0^\infty \one\{G(\by)^2 > t\}\,dt\right]}.\] Changing the order of expectations and integrals yields \[\frac{\int_0^\infty \err(G^2, t)\,dt}{\int_0^\infty \tail(G^2, t)\,dt} = \calS(G).\] This is enough to show the existence of $t' = \sqrt{t}\in\R$ where $\err(G,t') / \tail(G,t') \le \calS(G) / 1$, but it gives no lower bound on $\tail(G,t')$. Instead, we restrict the bounds of integration to a finite interval $[0,L]$ for a parameter $L$ to be chosen later. Observe that \[\frac{\int_0^L \err(G^2, t)\,dt}{\int_0^L \tail(G^2, t)\,dt} \le \calS(G) \cdot \frac{\int_0^\infty \tail(G^2,t)\,dt}{\int_0^L \tail(G^2,t)\,dt} = \calS(G)\cdot \frac{1}{\Ex_{\by}[\min\{G(\by)^2, L\}]}.\] Now, we aim to find a setting of $L$ for which $\Ex_{\by}[\min\{G(\by)^2, L\}]\ge 1 / 2$ and $\tail(G^2, L) \ge (1) / (4M_4)$. We compute \[\Ex_{\by}[\min\{G(\by)^2, L\}] \ge \Ex_{\by}[G(\by)^2] - \Ex_{\by}[G(\by)^2\cdot\one\{G(\by)^2 > L\}] \ge 1 - \sqrt{M_4 \cdot \Pr_{\by}[G(\by)^2 > L]},\] where the last inequality follows from Cauchy-Schwarz. We pick $L$ as follows: \[L = \sup\cbr{r\ge 0 : \Prx_{\by}[G(\by)^2 > r] \ge \frac{1}{4M_4}}\] so that $\Prx_{\by}[G(\by)^2 > L] \le 1 /(4M_4)$. This implies $\Ex_{\by}[\min\{G(\by)^2, L\}] \ge 1/2$, which gives \[\frac{\int_0^L \err(G^2,t)\,dt}{\int_0^L \tail(G^2,t)\,dt} \le 2\cdot\calS(G).\] By an averaging argument, there exists $t\in(0,L)$ such that $\err(G^2,t) / \tail(G^2,t)$ is at most $2\cdot \calS(G)$. Importantly, $\tail(G^2,t) \ge 1/(4M_4)$ by definition of $L$. Substituting $t' = \sqrt{t}$ completes the proof.
\end{proof}

We are now ready to prove the main lemma.

\begin{proof}[Proof of Lemma~\ref{lem:shift-existence}]
    Let $F : \R^d \to \R$ have mean $0$ and variance $1$ on $\Dy$, and write $M_4 = \Ex_{\by\sim\Dy}[F(\by)^4]$. For $y\in\R^d$, define $\tilde{F}(y) := F(y) - \nu$ where $\nu$ is a median of $F(\by)$. By Chebyshev's inequality, $|\nu|\le \sqrt{2}$.

    Now, take the positive and negative parts \[F^+ := \max\{\tilde{F},0\}\quad F^- := \max\{-\tilde{F}, 0\}\] so that $\tilde{F} = F^+ - F^-$ and $\del{\tilde{F}}^2 = (F^+)^2 + (F^-)^2$. Without loss of generality, assume that \[\Ex_{\by\sim\Dy}[F^+(\by)^2] \ge \Ex_{\by\sim\Dy}[F^-(\by)^2]\] so that \[\Ex_{\by\sim\Dy}[F^+(\by)^2] \ge \frac{\Ex_{\by}[\tilde{F}(\by)^2]}{2} = \frac{\Ex_{\by}[F(\by)^2] + \nu^2}{2} \ge \frac{1}{2}\] where we have used that $\Ex_{\by}[F(\by)] = 0$. Also, note that $\nu\le \sqrt{2}$ since $\Ex_{\by\sim\Dy}[F^+(\by)^2] \le 1$. Using $(a+b)^4 \le 8(a^4+b^4)$, this gives \[\Ex_{\by\sim\Dy}[F^+(\by)^4] \le 8M_4 + 8\nu^4 \le 8M_4 + 32 \le 40M_4\] where we used that $M_4 \ge \Ex_{\by\sim\Dy}[F(\by)^2] = 1$.
        
    Let $G:= C\cdot F^+$ for a constant $C = 1 /\sqrt{\Ex_{\by\sim\Dy}[F^+(\by)^2]}\in[1,\sqrt{2}]$, ensuring $\Ex_{\by}[G(\by)^2] = 1$ and $\Ex_{\by}[G(\by)^4] \le 4\cdot \Ex_{\by}[F^+(\by)^4]$.

    To bound $\calR(G)$, we must lower bound the variance of $G$. Since $\nu$ is a median, $\Pr(F \le \nu) \ge 1/2$, which implies $\Pr(F^+ = 0) \ge 1/2$. Thus $G=0$ with probability at least $1/2$. Using Cauchy-Schwarz:
    \[ \Ex[G]^2 = \Ex[G \cdot \one\{G>0\}]^2 \le \Ex[G^2] \cdot \Pr(G>0) \le 1 \cdot \frac{1}{2} = \frac{1}{2}. \]
    Therefore, $\Var[G] = \Ex[G^2] - \Ex[G]^2 \ge 1/2$.
    
    Now we can bound the energy ratio. Using the fact that $x \mapsto \max(x, 0)$ is a contraction and the variance bound above:
    \[ \calR(G) = \frac{\Ex[(G(\bx)-G(\by))^2]}{\Var[G]} \le \frac{C^2 \Ex[(F(\bx)-F(\by))^2]}{1/2} \le 2 C^2\cdot \calR(F) \le 4\cdot\calR(F). \]

    Applying Claim~\ref{claim:relating-S-and-R}, this gives \[\calS(G) \le 8\cdot\del{\sqrt{\calR(G)} + \calR(G)} \le 32\cdot \del{\sqrt{\calR(F)} + \calR(F)}.\] Then, by Claim~\ref{claim:nonnegative-shift}, there exists a shift $t'\in\R$ such that \[\frac{\err(G,t')}{\tail(G,t')} \le 2\cdot\calS(G)\quad\text{and}\quad\tail(G,t') \geq \frac{1}{16\cdot\Ex_{\by\sim\Dy}[F^+(\by)^4]}\]

    Since $G$ has median $0$ and $t' > 0$, $\tail(G,t') \le 1/2$, implying the following bound: \[\err(G,t') \le \calS(G) \le 32\cdot \del{\sqrt{\calR(F)} + \calR(F)} \implies \err(G,t') \le 64\sqrt{\calR(F)}.\] The implication holds because either $\calR(F)\ge 1$, in which case the statement follows from $\err(G,t')\le 1$, or $\calR(F) < 1$, in which case $\sqrt{\calR(F)} \ge \calR(F)$. 
    
    Moreover, since $\tail(G,t')\le 1/2$ we have $\bal(G,t') = \tail(G,t')$, giving: \[\bal(G,t') \ge \frac{1}{16\cdot \Ex_{\by}[F^+(\by)^4]} \ge \frac{1}{640M_4}.\]

    To conclude, we translate the cut function back to $F$. Namely, for $x\in\R^d$, the following holds: \[G(x) \ge t' \iff F^{+}(x) \ge t' / C \iff \tilde{F}(x) \ge t' / C \iff F(x) \ge \nu + t'/C.\] Thus, setting $\theta := \nu + t'/C$, we have $\err(F,\theta) \le 64\sqrt{\calR(F)}$ and $\bal(F,\theta) \ge 1/(640M_4)$, as desired. Note that if we instead selected $F^-$, the argument would be the same, but we would instead set $\theta := \nu - t' / C$.
\end{proof}

\subsection{Existence of a Good Polynomial (Proof of Lemma~\ref{lem:polynomial-existence})}
\label{subsec:polynomial-approx}
\polynomialExistence*
\begin{proof}
    Let $f\in\calF$ satisfy $\bal_{\Dpair}(f) \ge\beta$ and $\err_{\Dpair}(f) = \opt_\beta := \opt_\beta(\calF)$. We have the following bounds on the moments of $f$: \begin{itemize}
        \item \textbf{Mean.} $|\Ex_{\Dy}[f]| \le 1-2\beta$.
        \item \textbf{Variance.} $|\Var_{\Dy}[f]| \ge 4\beta\cdot(1-\beta) \ge 2\beta$.
        \item \textbf{Fourth moment.} Since the range of $f$ is $\{-1,1\}$, $\Ex_{\Dy}[f^4] = 1$
    \end{itemize}

    Moreover, the energy of $f$ is given by \[
        \calR_{\Dpair}(f) = \frac{\Ex_{(\bx,\by)\sim\Dpair}[(f(\bx)-f(\by))^2]}{\Var_{\Dy}[f]} \le \frac{4\cdot\opt_\beta}{2\beta} = (2/\beta)\cdot\opt_\beta.
    \]

    Now, by $(k,\eps)$-smoothness of $\Dpair$, there exists a polynomial $\tilde{A}$ such that for $\ell\in\{1,2,4\}$:
    \[
        \Ex_{\bx\sim\Dx}[|\tilde{A}(\bx)-f(\bx)|^\ell] \le \eps \quad\text{and}\quad \Ex_{\by\sim\Dy}[|\tilde{A}(\by)-f(\by)|^\ell]\le\eps.
    \]

    We now analyze the moments and energy of $\tilde{A}$.
    \begin{enumerate}
        \item \textbf{Mean.} $|\Ex_{\Dy}[\tilde{A}]| \le \Ex_{\Dy}[|\tilde{A}-f|] + |\Ex_{\Dy}[f]| \le (1-2\beta)+\eps < 1-\beta$.
        \item \textbf{Variance.} Using the fact that $|\sqrt{\Var(\bX)} - \sqrt{\Var(\bY)}| \le \E[(\bX-\bY)^2]^{1/2}$, we have:
    \[
        \sqrt{\Var_{\Dy}[\tilde{A}]} \ge \sqrt{\Var_{\Dy}[f]} - \sqrt{\Ex_{\Dy}[(\tilde{A}-f)^2]} \ge \sqrt{2\beta} - \sqrt{\eps} \ge \frac{\sqrt{\beta}}{3}
    \]
    Squaring this expression
    \[
        \Var_{\Dy}[\tilde{A}] \ge \beta/9.
    \]
        \item \textbf{Fourth moment.} Using the inequality $(a+b)^4 \le 8a^4 + 8b^4$, we have \[\Ex_{\Dy}[\tilde{A}^4] \le 8\cdot\Ex_{\Dy}[(\tilde{A}-f)^4] + 8\cdot\Ex_{\Dy}[f^4] \le 8\eps + 8 \le 9.\]
        \item \textbf{Energy.} Using the inequality $(a+b+c)^2 \le 3(a^2+b^2+c^2)$, we have: \[\Ex_{\Dpair}[(\tilde{A}(\bx)-\tilde{A}(\by))^2] \le 3 \cdot \del{\Ex_{\bx\sim\Dx}[(\tilde{A}(\bx)-f(\bx))^2] + \Ex_{(\bx,\by)\sim\Dpair}[(f(\bx)-f(\by))^2] + \Ex_{\by\sim\Dy}[(f(\by)-\tilde{A}(\by))^2]}\] Thus, \[\calR(\tilde{A}) \le \frac{6\eps + 12\cdot\opt_\beta}{\beta/9}\le (108/\beta)\cdot (\opt_\beta+\eps).\]
    \end{enumerate}

    Finally, we define \[A^\star := \frac{\tilde{A} - \Ex_{\Dy}[\tilde{A}]}{\sqrt{\Var_{\Dy}[\tilde{A}]}}\]
    so that $\Ex_{\Dy}[A^\star] = 0$ and $\Ex_{\Dy}[(A^\star)^2] = 1$. For the fourth moment, we use \[\Ex_{\Dy}[(A^\star)^4]\le \frac{8\cdot \Ex_{\Dy}[(\tilde{A})^4] + 8\cdot \Ex_{\Dy}[\tilde{A}]^4}{\Var_{\Dy}[\tilde{A}]^2} \le \frac{72 + 8}{\beta^2 / 81} \le 6480/\beta^2.\]
    Finally, since $\calR$ is invariant under shifting and scaling, $\calR(A^\star) =\calR(\tilde{A}) \le (108/\beta)\cdot(\opt_\beta+\eps)$.
\end{proof}

\subsection{Uniform Convergence (Proof of Lemma~\ref{lem:uniform-convergence})}
\label{subsec:proof-uniform-convergence}

\uniformConvergence*

For a degree-$k$ polynomial threshold function $g$, define
$\Set(g) := g^{-1}(1) \subseteq \R^d$ and $\Cut(g) := (g^{-1}(1)\times g^{-1}(0)) \cup (g^{-1}(0)\times g^{-1}(1))$.
Observe that $\err_{\Dpair}(g) = \Prx_{(\bx,\by)\sim\Dpair}[(\bx,\by)\in\Cut(g)]$ and
\[\bal_{\Dpair}(g) = \frac12 - \left|\frac12 - \Prx_{\by\sim\Dy}[\by\in\Set(g)]\right|.
\]

Define the following concept classes:
\[
\calC_k^{\Set} = \{\Set(g) : g \text{ is a degree-$k$ PTF}\}
\qquad\text{and}\qquad
\calC_k^{\Cut} = \{\Cut(g) : g \text{ is a degree-$k$ PTF}\}.
\]

\begin{proof}[Proof of Lemma~\ref{lem:uniform-convergence}]
    
We begin by showing that both $\vc(\calC_k^{\Set})$ and $\vc(\calC_k^{\Cut})$ are $O((d+k)^k)$.

\paragraph{VC dimension of $\calC_k^{\Set}$.}
A degree-$k$ PTF $g$ is equivalent to a halfspace in $\R^N$ where $N = \binom{d+k}{k} < (d+k)^k$. Since halfspaces in $\R^N$ have VC dimension $N+1$, it follows that
$\vc(\calC_k^{\Set}) \le N+1 \le (d+k)^k$.
\paragraph{VC dimension of $\calC_k^{\Cut}$.}
Suppose a set of pairs
$S=\{(\bx_1,\by_1),\ldots,(\bx_T,\by_T)$ is shattered by $\calC_k^{\Cut}$. Consider the set of points
$X=\{\bx_1,\ldots,\bx_T,\by_1,\ldots,\by_T\}$.

Whether a pair $(\bx_i,\by_i)$ lies in $\Cut(g)$ depends only on the labels $g(\bx_i)$ and $g(\by_i)$. Thus every labeling of the pairs corresponds to a labeling of the points in $X$ induced by $\calC_k^{\Set}$. In particular, if $S$ is shattered then $\calC_k^{\Set}$ must induce at least $2^T$ distinct labelings of $X$.

By Sauer's lemma, the number of labelings of $X$ realizable by $\calC_k^{\Set}$ is at most $(2eT/V)^V$ where $V=\vc(\calC_k^{\Set})$. Thus, $2^T \le (2eT/V)^V$, which implies $T=O(V)$, and hence $\vc(\calC_k^{\Cut}) = O((d+k)^k)$.

Finally, we apply Lemma~\ref{lem:vc-samples} on $\calC_k^\Set$ and $\calC_k^\Cut$, giving the desired sample complexity
\[
m = O\left(\frac{(d+k)^k+\log(1/\delta)}{\eps^2}\right),
\] concluding the proof.
\end{proof}

%% file: sections/missing-dasgupta.tex

\section{Omitted Proofs from Section~\ref{sec:dasgupta}}
\label{sec:missing-dasgupta}

\marginalTester*

\begin{proof}
    Consider the following algorithm \TestMass: \begin{enumerate}
        \item Draw $m = \ceil{C\log(1/\delta)/\gamma}$ independent samples $\bz_1,\ldots,\bz_m\sim\Dnn$, where $C>0$ is a sufficiently large constant.
        \item Let $\hat{\mu} = \frac{1}{m}\sum_{i=1}^m \one\{\bz_i\in R\}$ denote the fraction of samples in $R$
        \item Return \textsf{Large} if $\hat{\bmu} > 1.5\gamma$, and \textsf{Small} otherwise.
    \end{enumerate}

    Now, let $\mu := \Pr_{\bz\sim\Dnn}[\bz\in R]$ denote the true marginal mass of $R$ and consider the two cases:

    \paragraph{Case 1 $(\mu \le \gamma)$.} Since $\mu \le \gamma$, the probability that $\hat{\bmu} > 1.5\gamma$ is at most: 
    \[ 
    \Pr[\hat{\bmu} > 1.5\gamma] \le \exp\left(-\Omega(\gamma m)\right) = \exp\left(-\Omega\left(C \log(1/\delta)\right)\right) \le \delta. 
    \]
    where the constant $C$ is chosen appropriately.

    \paragraph{Case 2: $(\mu \ge 2\gamma)$.} Since $\mu \ge 2\gamma$, the probability that $\hat{\bmu} \le 1.5\gamma$ is at most
    \[ 
    \Pr[\hat{\bmu} \le 1.5\gamma] \le \Pr[\hat{\bmu} \le 0.75\mu] \le \exp\left(-\Omega(\mu m)\right) \le \exp\left(-\Omega(\gamma m)\right) = \exp\left(-\Omega\left(C \log(1/\delta)\right)\right) \le \delta. 
    \]
    for an appropriate $C$.

    $\TestMass$ evaluates inclusion in $R$ exactly once per sample, completing the proof.
\end{proof}

%% file: arxiv-refs.bib
@misc{IX23,
      title={Worst-case Performance of Popular Approximate Nearest Neighbor Search Implementations: Guarantees and Limitations}, 
      author={Piotr Indyk and Haike Xu},
      year={2023},
      eprint={2310.19126},
      archivePrefix={arXiv},
      primaryClass={cs.DS},
      url={https://arxiv.org/abs/2310.19126}, 
}

@inproceedings{JSDS19,
 author = {Jayaram Subramanya, Suhas and Devvrit, Fnu and Simhadri, Harsha Vardhan and Krishnawamy, Ravishankar and Kadekodi, Rohan},
 booktitle = {Advances in Neural Information Processing Systems},
 editor = {H. Wallach and H. Larochelle and A. Beygelzimer and F. d\textquotesingle Alch\'{e}-Buc and E. Fox and R. Garnett},
 pages = {},
 publisher = {Curran Associates, Inc.},
 title = {DiskANN: Fast Accurate Billion-point Nearest Neighbor Search on a Single Node},
 url = {https://proceedings.neurips.cc/paper_files/paper/2019/file/09853c7fb1d3f8ee67a61b6bf4a7f8e6-Paper.pdf},
 volume = {32},
 year = {2019}
}

@article{guruswami03,
author = {Guruswami, Venkatesan},
title = {Inapproximability Results for Set Splitting and Satisfiability
Problems with No Mixed Clauses},
year = {2003},
issue_date = {Mar 2004},
journal = {Algorithmica},
month = dec,
pages = {451–469},
numpages = {19}
}

@article{charikar2005,
title = {Clustering with qualitative information},
journal = {Journal of Computer and System Sciences},
volume = {71},
number = {3},
pages = {360-383},
year = {2005},
note = {Learning Theory 2003},
issn = {0022-0000},
doi = {https://doi.org/10.1016/j.jcss.2004.10.012},
url = {https://www.sciencedirect.com/science/article/pii/S0022000004001424},
author = {Moses Charikar and Venkatesan Guruswami and Anthony Wirth},
}

@inproceedings{klivans2002,
author = {Klivans, Adam and O'Donnell, Ryan and Servedio, Rocco A.},
title = {Learning Intersections and Thresholds of Halfspaces},
year = {2002},
isbn = {0769518222},
publisher = {IEEE Computer Society},
address = {USA},
booktitle = {Proceedings of the 43rd Symposium on Foundations of Computer Science},
pages = {177–186},
numpages = {10},
series = {FOCS '02}
}

@book{shalev2014, place={Cambridge}, title={Understanding Machine Learning: From Theory to Algorithms}, publisher={Cambridge University Press}, author={Shalev-Shwartz, Shai and Ben-David, Shai}, year={2014}}

@INPROCEEDINGS{jiang2020,
  author={Jiang, Haotian and Kathuria, Tarun and Lee, Yin Tat and Padmanabhan, Swati and Song, Zhao},
  booktitle={2020 IEEE 61st Annual Symposium on Foundations of Computer Science (FOCS)}, 
  title={A Faster Interior Point Method for Semidefinite Programming}, 
  year={2020},
  volume={},
  number={},
  pages={910-918},
  doi={10.1109/FOCS46700.2020.00089}}

@inproceedings{raghavendra2012,
author = {Raghavendra, Prasad and Steurer, David and Tulsiani, Madhur},
title = {Reductions between Expansion Problems},
year = {2012},
isbn = {9780769547084},
publisher = {IEEE Computer Society},
address = {USA},
url = {https://doi.org/10.1109/CCC.2012.43},
doi = {10.1109/CCC.2012.43},
booktitle = {Proceedings of the 2012 IEEE Conference on Computational Complexity (CCC)},
pages = {64–73},
numpages = {10},
series = {CCC '12}
}

@article{diakonikolas2010,
author = {Diakonikolas, Ilias and Gopalan, Parikshit and Jaiswal, Ragesh and Servedio, Rocco A. and Viola, Emanuele},
title = {Bounded Independence Fools Halfspaces},
journal = {SIAM Journal on Computing},
volume = {39},
number = {8},
pages = {3441-3462},
year = {2010},
doi = {10.1137/100783030}
}

@inproceedings{dasgupta2016,
author = {Dasgupta, Sanjoy},
title = {A cost function for similarity-based hierarchical clustering},
year = {2016},
isbn = {9781450341325},
publisher = {Association for Computing Machinery},
address = {New York, NY, USA},
url = {https://doi.org/10.1145/2897518.2897527},
doi = {10.1145/2897518.2897527},
booktitle = {Proceedings of the Forty-Eighth Annual ACM Symposium on Theory of Computing},
pages = {118–127},
numpages = {10},
location = {Cambridge, MA, USA},
}

@inproceedings{charikar2017,
author = {Charikar, Moses and Chatziafratis, Vaggos},
title = {Approximate hierarchical clustering via sparsest cut and spreading metrics},
year = {2017},
publisher = {Society for Industrial and Applied Mathematics},
address = {USA},
booktitle = {Proceedings of the Twenty-Eighth Annual ACM-SIAM Symposium on Discrete Algorithms},
pages = {841–854},
numpages = {14},
location = {Barcelona, Spain},
}

@book{knuth1973,
  author = {Knuth, Donald},
  pages = {391-392},
  publisher = {Addison-Wesley},
  title = {The Art Of Computer Programming, vol. 3: Sorting And Searching},
  year = 1973
}

@inbook{balcan21, place={Cambridge}, title={Data-Driven Algorithm Design}, booktitle={Beyond the Worst-Case Analysis of Algorithms}, publisher={Cambridge University Press}, author={Balcan, Maria-Florina}, year={2021}, pages={626–645}}


%% file: waingarten.bib
@String{algorithmica     = {Algorithmica}}

@String{arxiv            = {arXiv preprint arXiv:}}

@String{focs2006         = {Proceedings of the 47th Annual {IEEE} Symposium on Foundations of Computer Science ({FOCS}~'2006)}}

@String{focs2014         = {Proceedings of the 55th Annual {IEEE} Symposium on Foundations of Computer Science ({FOCS}~'2014)}}

@String{focs2018         = {Proceedings of the 59th Annual {IEEE} Symposium on Foundations of Computer Science ({FOCS}~'2018)}}

@String{icm2018          = {Proceedings of the International Congress of Mathematicians ({ICM}~'2018)}}

@String{infcomp          = {Information and Computation}}

@String{jacm             = {Journal of the ACM}}

@String{neurips2021      = {Proceedings of Advances in Neural Information Processing Systems~34 ({NeurIPS}~'2021)}}

@String{NEURIPS2024      = {Proceedings of Advances in Neural Information Processing Systems~38 ({NeurIPS}~'2024)}}

@String{sicomp           = {{SIAM} Journal on Computing}}

@String{socg2004         = {Proceedings of the 20th {ACM} Symposium on Computational Geometry ({SoCG}~'2004)}}

@String{soda2014         = {Proceedings of the 25th {ACM-SIAM} Symposium on Discrete Algorithms ({SODA}~'2014)}}

@String{soda2017         = {Proceedings of the 28th {ACM-SIAM} Symposium on Discrete Algorithms ({SODA}~'2017)}}

@String{soda2021         = {Proceedings of the 32nd {ACM-SIAM} Symposium on Discrete Algorithms ({SODA}~'2021)}}

@String{stoc1998         = {Proceedings of the 30th {ACM} Symposium on the Theory of Computing ({STOC}~'1998)}}

@String{stoc2002         = {Proceedings of the 34th {ACM} Symposium on the Theory of Computing ({STOC}~'2002)}}

@String{stoc2015         = {Proceedings of the 47th {ACM} Symposium on the Theory of Computing ({STOC}~'2015)}}

@String{stoc2018         = {Proceedings of the 50th {ACM} Symposium on the Theory of Computing ({STOC}~'2018)}}

@String{stoc2024         = {Proceedings of the 56th {ACM} Symposium on the Theory of Computing ({STOC}~'2024)}}

@String{tpami            = {{IEEE} Transactions on Pattern Analysis and Machine Intelligence}}

@inproceedings{AAKK14,
	author = {Amirali Abdullah and Alexandr Andoni and Ravindran Kannan and Robert Krauthgamer},
	booktitle = FOCS2014,
	pages = {581--590},
	title = {Spectral Approaches to Nearest Neighbor Search},
	year = {2014}}

@inproceedings{AI06,
	author = {Alexandr Andoni and Piotr Indyk},
	booktitle = FOCS2006,
	pages = {459--468},
	title = {Near-Optimal Hashing Algorithms for Approximate Nearest Neighbor in High Dimensions},
	year = {2006}}

@inproceedings{AJW25,
author = {Andoni, Alexandr and Jiang, Shunhua and Weinstein, Omri},
title = {A Framework for Building Data Structures from Communication Protocols},
year = {2025},
booktitle = {Proceedings of the 57th Annual ACM Symposium on Theory of Computing},
pages = {256–267},
numpages = {12},
location = {Prague, Czechia},
series = {STOC '25}
}

@article{GR17,
    author = {Gupta, Rishi and Roughgarden, Tim},
    title = {A PAC Approach to Application-Specific Algorithm Selection},
    journal = {SIAM Journal on Computing},
    volume = {46},
    number = {3},
    pages = {992-1017},
    year = {2017},
}

@inproceedings{AINR14,
	author = {Alexandr Andoni and Piotr Indyk and Huy L. Nguyen and Ilya Razenshteyn},
	booktitle = SODA2014,
	pages = {1018--1028},
	title = {Beyond Locality-Sensitive Hashing},
	year = {2014}}

@inproceedings{ALRW17,
	author = {Alexandr Andoni and Thijs Laarhoven and Ilya Razenshteyn and Erik Waingarten},
	booktitle = SODA2017,
	pages = {47--66},
	title = {Optimal Hashing-based Time--Space Trade-offs for Approximate Near Neighbors},
	year = {2017}}

@inproceedings{AR15,
	author = {Alexandr Andoni and Ilya Razenshteyn},
	booktitle = STOC2015,
	pages = {793--801},
	title = {Optimal Data-Dependent Hashing for Approximate Near Neighbors},
	year = {2015}}

@inproceedings{C02,
	author = {Moses Charikar},
	booktitle = STOC2002,
	pages = {380--388},
	title = {Similarity estimation techniques from rounding algorithms},
	year = {2002}}

@article{C88,
	author = {Kenneth L. Clarkson},
	journal = SICOMP,
	number = {4},
	pages = {830--847},
	title = {A Randomized Algorithm for Closest-Point Queries},
	volume = {17},
	year = {1988}}

@inproceedings{DIIM04,
	author = {Mayur Datar and Nicole Immorlica and Piotr Indyk and Vahab S. Mirrokni},
	booktitle = SOCG2004,
	pages = {253--262},
	title = {Locality-sensitive hashing scheme based on p-stable distributions},
	year = {2004}}

@inproceedings{IM98,
	author = {Piotr Indyk and Rajeev Motwani},
	booktitle = STOC1998,
	pages = {604--613},
	title = {Approximate Nearest Neighbors: Towards Removing the Curse of Dimensionality},
	year = {1998}}

@article{M93,
	author = {Stefan Meiser},
	journal = INFCOMP,
	number = {2},
	pages = {286--303},
	title = {Point Location in Arrangements of Hyperplanes},
	volume = {106},
	year = {1993}}

@misc{T,
	howpublished = {Available as \texttt{https://twitter.com/}},
	title = {Twitter}}

@inproceedings{ANNRW18,
	author = {Alexandr Andoni and Assaf Naor and Aleksandar Nikolov and Ilya Razenshteyn and Erik Waingarten},
	booktitle = STOC2018,
	pages = {787--800},
	title = {Data-dependent Hashing via Non-linear Spectral Gaps},
	year = {2018}}

@article{LR99,
	author = {Tom Leighton and Satish Rao},
	journal = JACM,
	number = {6},
	pages = {787--832},
	title = {Multicommodity max-flow min-cut theorems and their use in designing approximation algorithms},
	volume = {46},
	year = {1999}}

@article{ARV09,
	author = {Sanjeev Arora and Satish Rao and Umesh Vazirani},
	journal = JACM,
	number = {2},
	pages = {1--37},
	title = {Expander flows, geometric embeddings and graph partitioning},
	volume = {56},
	year = {2009}}

@misc{SUBLINEAR,
	howpublished = {\url{https://sublinear.info/index.php?title=Main_Page}},
	title = {List of Open Problems in Sublinear Algorithms}}

@article{AM85,
	author = {Noga Alon and Vitaly D. Milman},
	journal = {Journal of Combinatorial Theory, Series B},
	number = {1},
	pages = {73--88},
	title = {$\lambda$1, isoperimetric inequalities for graphs, and superconcentrators},
	volume = {38},
	year = {1985}}

@inproceedings{KPW26,
author = {Sanjeev Khanna and Ashwin Padaki and Erik Waingarten},
title = {Sparse Navigable Graphs for Nearest Neighbor Search: Algorithms and Hardness},
booktitle = {Proceedings of the 2026 Annual ACM-SIAM Symposium on Discrete Algorithms (SODA)},
chapter = {},
pages = {1606-1630},
doi = {10.1137/1.9781611978971.58},
year = 2026,
}

@article{DKZ20,
  title={Near-optimal sq lower bounds for agnostically learning halfspaces and relus under gaussian marginals},
  author={Diakonikolas, Ilias and Kane, Daniel and Zarifis, Nikos},
  journal={Advances in Neural Information Processing Systems},
  volume={33},
  pages={13586--13596},
  year={2020}
}

@article{MV22,
author = {Mitzenmacher, Michael and Vassilvitskii, Sergei},
title = {Algorithms with predictions},
year = {2022},
issue_date = {July 2022},
publisher = {Association for Computing Machinery},
address = {New York, NY, USA},
volume = {65},
number = {7},
issn = {0001-0782},
url = {https://doi.org/10.1145/3528087},
doi = {10.1145/3528087},
journal = {Commun. ACM},
month = jun,
pages = {33–35},
numpages = {3}
}

@inproceedings{CDFJLMSSW26,
author = {Alex Conway and Laxman Dhulipala and Martin Farach-Colton and Rob Johnson and Ben Landrum and Christopher Musco and Yarin Shechter and Torsten Suel and Richard Wen},
title = {Sparse Navigable Graphs for Nearest Neighbor Search: Algorithms and Hardness},
booktitle = {Proceedings of the 2026 Annual ACM-SIAM Symposium on Discrete Algorithms (SODA)},
pages = {1606-1630},
year = 2026,
}

@inproceedings{CD07,
 author = {Cayton, Lawrence and Dasgupta, Sanjoy},
 booktitle = {Advances in Neural Information Processing Systems},
 editor = {J. Platt and D. Koller and Y. Singer and S. Roweis},
 pages = {},
 publisher = {Curran Associates, Inc.},
 title = {A learning framework for nearest neighbor search},
 url = {https://proceedings.neurips.cc/paper_files/paper/2007/file/0d7de1aca9299fe63f3e0041f02638a3-Paper.pdf},
 volume = {20},
 year = {2007}
}

@inproceedings{DIRW20,
title={Learning Space Partitions for Nearest Neighbor Search},
author={Yihe Dong and Piotr Indyk and Ilya Razenshteyn and Tal Wagner},
booktitle={International Conference on Learning Representations},
year={2020},
url={https://openreview.net/forum?id=rkenmREFDr}
}

@inproceedings{AIR18,
	author = {Alexandr Andoni and Piotr Indyk and Ilya Razenshteyn},
	booktitle = ICM2018,
	pages = {3287--3318},
	title = {Approximate nearest neighbor search in high dimensions},
	year = {2018}}

@book{MP69,
	author = {Marvin Minsky and Seymour Papert},
	publisher = {{MIT} Press},
	title = {Perceptrons},
	year = {1969}}

@inproceedings{ANRW21,
	author = {Alexandr Andoni and Aleksandar Nikolov and Ilya Razenshteyn and Erik Waingarten},
	booktitle = SODA2021,
	title = {Approximate nearest neighbors beyond space partitions},
	year = {2021}}

@article{S84,
	author = {Hanan Samet},
	journal = {ACM Computing Surveys (CSUR)},
	number = {2},
	pages = {187--260},
	title = {The quadtree and related hierarchical data structures},
	volume = {12},
	year = {1984}}

@article{LT80,
	author = {Richard J. Lipton and Robert E. Tarjan},
	journal = SICOMP,
	number = {3},
	pages = {615--627},
	title = {Applications of a planar separator theorem},
	volume = {9},
	year = {1980}}

@inproceedings{R18,
	author = {Aviad Rubinstein},
	booktitle = FOCS2018,
	pages = {1260--1268},
	title = {Hardness of approximate nearest neighbor search},
	year = {2018}}

@inproceedings{BIGANN21,
	author = {Harsha Vardhan Simhadri and George Williams and Martin Aum\"{u}ler and Matthijs Douze and Artem Babenko and Dmitry Baranchuk and Qi Chen and Lucas Hosseini and Ravishankar Krishnaswamy and Gopal Srinivasa and Shuas Jayaram Subramanya and Jingdong Wang},
	booktitle = NEURIPS2021,
	pages = {177--189},
	title = {Results of the {NeurIPS'21} Challenge on Billion-Scale Approximate Nearest Neighbor Search},
	year = {2021}}

@InProceedings{JWZ24,
  author    = {Rajesh Jayaram and Erik Waingarten and Tian Zhang},
  booktitle = STOC2024,
  title     = {Data-dependent LSH for the Earth Mover's Distance},
  year      = {2024},
}

@article{KKMS08,
author = {Kalai, Adam Tauman and Klivans, Adam R. and Mansour, Yishay and Servedio, Rocco A.},
title = {Agnostically Learning Halfspaces},
journal = {SIAM Journal on Computing},
volume = {37},
number = {6},
pages = {1777-1805},
year = {2008},
doi = {10.1137/060649057},
}

@inproceedings{BBMWWZ25,
    author = {Yiqiao Bao and Anubhav Baweja and Nicolas Menand and Erik Waingarten and Nathan White and Tian Zhang},
    title = {Average-Distortion Sketching},
    year    = {2025},
    booktitle = {2025 IEEE 66th Annual Symposium on Foundations of Computer Science (FOCS)}, 
}

@ARTICLE{WZSSS17,
author={Wang, Jingdong and Zhang, Ting and Song, Jingkuan and Sebe, Nicu and Shen, Heng Tao},
journal={ IEEE Transactions on Pattern Analysis \& Machine Intelligence },
title={{ A Survey on Learning to Hash }},
year={2018},
volume={40},
number={04},
ISSN={1939-3539},
pages={769-790},
publisher={IEEE Computer Society},
address={Los Alamitos, CA, USA},
month=apr}

@INPROCEEDINGS{AN25,
  author={Andoni, Alexandr and Nosatzki, Negev Shekel},
  booktitle={2025 IEEE 66th Annual Symposium on Foundations of Computer Science (FOCS)}, 
  title={Embeddings into Similarity Measures for Nearest Neighbor Search}, 
  year={2025},
  volume={},
  number={},
  pages={864-893},
}

@Article{MY18,
  author  = {Yu A Malkov and Dmitry A Yashunin},
  journal = TPAMI,
  title   = {Efficient and robust approximate nearest neighbor search using hierarchical navigable small world graphs},
  year    = {2018},
  number  = {4},
  pages   = {824--836},
  volume  = {42},
}

@InProceedings{DGMMS24,
  author    = {Haya Diwan and Jinrui Guo and Cameron Musco and Christopher Musco and Tosten Suel},
  booktitle = NEURIPS2024,
  title     = {Navigable graphs for high-dimensional nearest neighbor search: constructions and limits},
  year      = {2024},
}

@misc{BIGANN23,
      title={Results of the Big ANN: NeurIPS'23 competition}, 
      author={Harsha Vardhan Simhadri and Martin Aumüller and Amir Ingber and Matthijs Douze and George Williams and Magdalen Dobson Manohar and Dmitry Baranchuk and Edo Liberty and Frank Liu and Ben Landrum and Mazin Karjikar and Laxman Dhulipala and Meng Chen and Yue Chen and Rui Ma and Kai Zhang and Yuzheng Cai and Jiayang Shi and Yizhuo Chen and Weiguo Zheng and Zihao Wan and Jie Yin and Ben Huang},
      year={2024},
      eprint={2409.17424},
      archivePrefix={arXiv},
      primaryClass={cs.IR},
      url={https://arxiv.org/abs/2409.17424}, 
}
